\documentclass[10pt]{article}

\usepackage{amsmath,amssymb,amsthm,color,epsfig,mathrsfs}
\usepackage{amsbsy,array,float,caption}
\usepackage{hyperref}
\usepackage{bookmark}

\newtheorem{theorem}{Theorem}

\newtheorem{assumption}[theorem]{Assumption}

\newcounter{spslist}

\newcommand{\col}[3]{ \renewcommand{\arraystretch}{#1}
                \left[\!\! \begin{array}{c} #2 \\ #3 \end{array} \!\!\right] }

\newcommand{\mat}[5]{ \renewcommand{\arraystretch}{#1}
                    \left[\!\! \begin{array}{cc}
                            #2 & #3 \\
                            #4 & #5 \end{array} \!\!\right] }

\newcounter{geqncount}
    {\refstepcounter{equation}%
     \setcounter{geqncount}{\value{equation}}%
     \setcounter{equation}{0}%
  }%
    {\setcounter{equation}{\value{geqncount}}}

\newtheorem{Theorem}{Theorem}[section]
\newtheorem{Lemma}{Lemma}[section]

\newtheorem{Proposition}{Proposition}[section]

\newcommand{\lf}{{\!-\!}}
\newcommand{\rt}{{+}}

\newcommand{\rte}{{+e}}
\newcommand{\rtp}{{+p}}

\newcommand{\Pl}{P_{\!-\!}} 
\newcommand{\Ple}{P_{-e}} 
\newcommand{\vle}{v_{-e}} 
\newcommand{\wle}{w_{-e}} 
\newcommand{\Plp}{P_{-p}} 
\newcommand{\vlp}{v_{-p}} 
\newcommand{\wlp}{w_{-p}} 
\renewcommand{\Pr}{P_+} 
\newcommand{\Pre}{P_{+e}} 
\newcommand{\vre}{v_{+e}} 
\newcommand{\wre}{w_{+e}} 
\newcommand{\Prp}{P_{+p}} 
\newcommand{\vrp}{v_{+p}} 
\newcommand{\wrp}{w_{+p}} 
\newcommand{\Plr}{P_\mp} 

\newcommand{\Pres}{P_\text{\scriptsize res}}
\newcommand{\Preg}{P_\text{\scriptsize reg}}

\renewcommand{\O}{{O}}

\newcommand{\slantfrac}[2]{{\hbox{\tiny$\raisebox{0.6pt}{#1}\!/_{\!#2}$}}}

\newcommand{\threehalves}{\slantfrac{3}{2}}
\newcommand{\onehalf}{\slantfrac{1}{2}}

\newcommand{\kw}{(\kk,\omega)}
\newcommand{\kwz}{(\kk_0,\omega_0)}
\newcommand{\tomega}{{\tilde\omega}}

\newcommand{\tkk}{{\tilde\kk}}
\newcommand{\tkw}{(\tkk,\tomega)}

\newcommand{\Psil}{\Psi_-}
\newcommand{\Psir}{\Psi_+}
\newcommand{\psig}{\psi^\guided}
\newcommand{\Psig}{\Psi^\guided}
\newcommand{\Psigr}{\Psig_+}
\newcommand{\Psigl}{\Psig_-}

\newcommand{\inc}{\text{\scriptsize in}}

\newcommand{\out}{\text{\scriptsize out}}
\newcommand{\guided}{\text{\itshape g}}

\newcommand{\Range}{{\mathrm{Ran}}}
\newcommand{\Null}{{\mathrm{Null}}}
\newcommand{\sspan}{{\mathrm{span}}}

\newcommand{\fadj}{{\text{\tiny $[$}*{\text{\tiny $]$}}}}
\newcommand{\fperp}{{\text{\tiny $[$}\perp{\text{\tiny $]$}}}}

\newcommand{\rr}{{\mathbf r}}
\newcommand{\HH}{{\mathbf H}}
\newcommand{\EE}{{\mathbf E}}
\newcommand{\BB}{{\mathbf B}}
\newcommand{\DD}{{\mathbf D}}
\renewcommand{\SS}{{\mathbf S}}
\newcommand{\kk}{{\boldsymbol{\kappa}}}
\newcommand{\xxi}{{\boldsymbol{\xi}}}
\newcommand{\ZZ}{\mathbb{Z}}
\newcommand{\RR}{\mathbb{R}}
\newcommand{\CC}{\mathbb{C}}

\renewcommand{\Re}{\text{Re}\,}
\renewcommand{\Im}{\mathrm{Im}\,}

\newcommand{\half}{\frac{1}{2}}

\newcommand{\commentsps}[1]{}
\newcommand{\commentatw}[1]{}

\begin{document}

\begin{center}
{\bfseries \Large  Resonant electromagnetic scattering in anisotropic layered media}
\end{center}

\vspace{0ex}

\begin{center}
{\scshape \large Stephen P. Shipman${}^\dagger$ and Aaron T. Welters${}^\ddagger$} \\
\vspace{2ex}
{\itshape ${}^\dagger$Department of Mathematics\\
Louisiana State University\\
Baton Rouge, Louisiana 70803, USA\\
\vspace{1ex}
${}^\ddagger$Department of Mathematics\\
Massachusetts Institute of Technology\\
Cambridge, Massachusetts 02139, USA
}
\end{center}

\vspace{2ex}
\centerline{\parbox{0.9\textwidth}{
{\bf Abstract.}\
The resonant excitation of an electromagnetic guided mode of a slab structure by exterior radiation results in anomalous scattering behavior, including sharp energy-transmission anomalies and field amplification around the frequency of the slab mode.  In the case of a periodically layered ambient medium, anisotropy serves to couple the slab mode to radiation.  Exact expressions for scattering phenomena are proved by analyzing a pole of the full scattering matrix as it moves off the real frequency axis into the lower half complex plane under a detuning of the wavevector parallel to the slab.  The real pole is the frequency of a perfect (infinite Q) guided mode, which becomes lossy as the frequency gains an imaginary part.
This work extends results of Shipman and Venakides to evanescent source fields and two-dimensional parallel wavevector and demonstrates by example how the latter allows one to control independently the width and central frequency of a resonance by varying the angle of incidence of the source field.
The analysis relies on two nondegeneracy conditions of the complex dispersion relation for slab modes (relating poles of the scattering matrix to wavevector), which were assumed in previous works and are proved in this work for layered media.  One of them asserts that the dispersion relation near the wavevector $\kk$ and frequency $\omega$ of a perfect guided mode is the zero set of a {\em simple} eigenvalue $\ell\kw$, and the other relates $\partial\ell/\partial\omega$ to the total energy of the mode, thereby implying that this derivative is~nonzero.
}}

\vspace{3ex}
\noindent
\begin{mbox}
\end{mbox}
\vspace{3ex}

\hrule
\vspace{1ex}


\vspace{1.5ex}

When an electromagnetic mode of a slab or film is excited by an exterior source field, delicate resonance phenomena occur.
The most notable are high field amplification and sharp variations of the transmitted energy across the slab as the frequency and angle of incidence of the source field are tuned.  These phenomena are utilized in photonic devices such as lasers~\cite{KanskarPaddonPacradouni1997} and light-emitting diodes~\cite{FanVilleneuveJoannopoul2000}.

Excitation of a guided mode is typically achieved by periodic variation of the dielectric properties of the slab in directions parallel to it. 
The periodicity couples radiating Rayleigh-Bloch waves with evanescent ones composing a guided mode~\cite{FanJoannopoul2002}.
In an explicit example, we show that coupling can be achieved without periodicity of the slab by replacing the air with an anisotropic ambient medium that supports radiation and evanescent modes at the same frequency and wavevector along the slab.

This paper analyzes resonant scattering phenomena through the perturbation of a pole of a scattering matrix representing the complex frequency of a generalized (leaky) guided mode, or scattering resonance.
The real frequency of a perfect guided mode (infinite quality factor), attains a small imaginary part as the wavevector parallel to the slab is detuned from the precise value required to support the perfect mode.
This is a manifestation of the instability of a perfect guided mode whose frequency is embedded in the continuous spectrum of radiation states.  Spectrally embedded guided modes have been demonstrated in periodic photonic structures~\cite{Bonnet-BeStarling1994,ShipmanVenakides2003,ShipmanVenakides2005,TikhodeevYablonskiMuljarov2002,Wei-HsuZhenChua2013,Wei-HsuZhenLee2013}, discrete systems~\cite{PtitsynaShipman2012,ShipmanRibbeckSmith2010}, 
and anisotropic layered media~\cite{ShipmanWelters2012}.

The characteristic peak-dip shape of resonant transmission anomalies (Fig.~\ref{fig:resonance},\,\ref{fig:transmission}) is often called a ``Fano resonance" or ``Fano lineshape".
There are several formulae for the Fano lineshape in the literature that are based on the underlying principle of coupling between an oscillatory mode of a structure and radiation states~\cite{Fano1961,DurandPaidarovaGadea2001,FanSuhJoannopoul2003,FanJoannopoul2002}.
Our analysis of the scattering matrix yields a formula for the transmission anomaly and the associated field amplification that is based solely on the Maxwell equations, without invoking a heuristic model.  The formula shows explicitly how the frequencies of the peaks and dips depend on the angle of incidence.  The reduction of the problem to the scattering matrix is an expression of the universal applicability of the formulae to very general linear scattering problems.

The descriptions of resonances in this work extend previous formulae, involving one angle of incidence~\cite{ShipmanVenakides2005,Shipman2010}, to two angles of incidence, and it is shown that this permits independent control over the  width and central frequency of a resonance; this is important, for example, in the tuning of LED structures~\cite{FanVilleneuveJoannopoul2000}.
Resonances can also be tuned by structural mechanisms, such as by rotating the anisotropic slab in the example of section~\ref{sec:example}.  We do not pursue structural perturbations in this paper, but they can be incorporated into the analysis by treating structural parameters on par with the wavevector.

Our analysis allows the ambient medium to be a 1D photonic crystal, that is, a periodically layered medium.  The slab or film consists of a defective layer, or slab, embedded in the ambient medium~(Fig.~\ref{fig:layered}).   
Rigorous derivation of resonances based on the Maxwell equations requires a careful treatment of energy density and flux in periodically layered media and new results on the non-degeneracy of the dispersion relation for generalized guided modes.  Theorems~\ref{thm:nondegeneracy} relates the frequency derivative of the dispersion relation for slab modes to the total energy of the mode and leads to a justification of a genericity assumption that was made in previous works~\cite{ShipmanVenakides2005,Shipman2010}.  

\begin{center}
{\scalebox{0.55}{\includegraphics{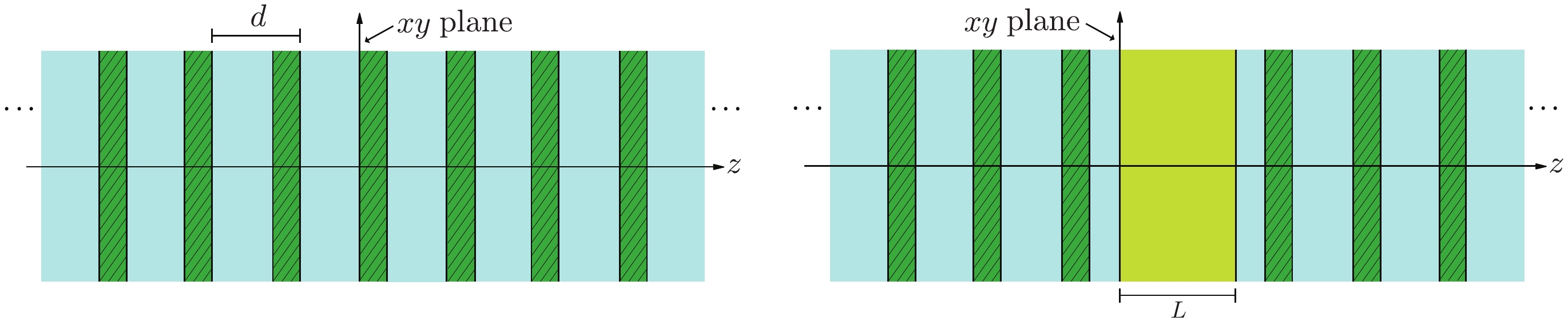}}}
\vspace{-1.8ex}
\captionof{figure}{\small
{\bfseries Left.}  A periodic layered medium with anisotropic layers.
{\bfseries Right.}  A defective layer, or slab, embedded in an ambient periodic layered medium.  In this example, the structure is symmetric about the centerline of the slab.
}
\label{fig:layered}
\end{center}

Layered media offer a scenario for resonant scattering by a planar waveguide that is physically realistic yet simple enough to permit essentially exact calculations and quick numerical computations.   They provide model problems for computing, for instance, the effects of random fabrication errors on resonant phenomena or scattering in the slow light regime.

This paper starts with an elaboration of an example of resonant scattering introduced in~\cite{ShipmanWelters2012}.  A second example in that work elucidates scattering in a ``slow light" medium---a periodically layered medium in which the group velocity of the propagating mode in one direction vanishes at a frequency in the interior of a spectral band~\cite{FigotinVitebskiy2006,Welters2011a}.  This will be treated in detail in a another communication.

Here is an overview of the content and main results of this paper.

\smallskip
{\bfseries Section~\ref{sec:example}: An example.}  An explicit construction demonstrates concretely how a spectrally embedded guided mode can be created in an anisotropic layered medium and how perturbations of the system result in resonant scattering.  The example introduces key concepts, such as the complex dispersion relation for guided slab modes, formulae for transmission anomalies, and the quality factor.

{\bfseries Section~\ref{sec:ODEs}: Energy in layered media.}
Reduction of the Maxwell equations to an ODE system in lossless anisotropic layered media is reviewed.
Theorems~\ref{thm:EnergyConservationLaw} and~\ref{thm:EnergyDensity} relate energy density, flux, and conservation for harmonic fields at complex frequency directly to the Maxwell ODEs.  These results, which don't appear to be in the literature, provide the basis for the designation of rightward and leftward modes in the scattering problem and the basis for Theorems~\ref{thm:scattering} and~\ref{thm:nondegeneracy} on the nondegeneracy of the complex dispersion relation for slab modes, used in the analysis of scattering anomalies.

{\bfseries Section~\ref{sec:scattering}: Resonant scattering.}  Detailed descriptions of transmission anomalies (Fig.~\ref{fig:resonance},\,\ref{fig:transmission} and sec.~\ref{sec:transmission}, eq.~\ref{T2b}) and resonant field amplification (sec.~\ref{sec:amplification}) are derived by perturbation analysis of scattering resonances, as they depend on the wavevector~$\kk$ parallel to the slab.  The complex frequencies $\omega$ of the resonances are poles of the full scattering matrix $S\kw$ for fixed $\kk$.  Transmission anomalies are revealed by reduction to the far-field scattering matrix~$S_0\kw$ (\ref{S0}), whose entries are ratios of analytic functions that both vanish at the parameters $\kwz$ of a perfect guided mode, as proved first in~\cite{ShipmanVenakides2005}.

The analysis relies on certain generic nondegeneracy conditions on the guided-mode dispersion relation that were assumed in previous works~\cite{Shipman2010,ShipmanVenakides2005}.  These conditions are proved for layered media in two theorems stated in section~\ref{sec:scattering} and proved in section~\ref{sec:nondegeneracy}.

{\bfseries Section~\ref{sec:nondegeneracy}: Nondegeneracy of guided modes}---proofs.
Theorem~\ref{thm:scattering} asserts that the dispersion relation is the zero set of an analytic and algebraically simple eigenvalue of a matrix whose nullspace corresponds to the guided modes of the slab.  It also identifies the space of incoming fields for which the scattering problem has a (nonunique) solution; this space includes propagating harmonics.

Theorem~\ref{thm:nondegeneracy} gives a formula for the $\omega$-derivative of the dispersion relation at the wavevector-frequency 
pair of a perfect guided mode in terms of the total energy of the mode. The derivative cannot be zero because 
the energy density is positive.

{\bfseries Section~\ref{sec:reduction}: Appendix on electromagnetics in layered media.}  
The dependence of the Maxwell ODEs, their solutions, and energy density, on the system parameters $(z,\kk,\omega)$ is discussed.  Theorem~\ref{thm:EnergyIndRep} expresses the energy of a field in a periodic layered medium (1D photonic crystal) in terms of the energy of a solution of an ``effective" homogeneous ODE $\psi'\!=\!iK\psi$, where $K$ is the matrix Floquet exponent.  This theorem is an extension to periodic media of Theorems~\ref{thm:EnergyConservationLaw} and~\ref{thm:EnergyDensity} and is needed in the proof of Theorem~\ref{thm:nondegeneracy}.

\section{An example of resonant scattering}\label{sec:example} 

The objective of this section is to introduce the reader, through a concrete example, to the resonance phenomena investigated in this article.

We explicitly construct spectrally embedded guided modes of a defect layer, or slab, embedded in a homogeneous anisotropic ambient medium (Fig.~\ref{fig:slab}).  Perturbation of the exact parameters that permit the construction of a guided mode results in resonant field amplification and sharp variations in the transmission of energy of an incident plane wave across the slab---we call this resonant scattering.
Controlling the central frequency and the spectral width of a resonance is important for applications~\cite{FanJoannopoul2002}.  {\itshape This example shows how these characteristics can be controlled independently by varying the two-dimensional angle of incidence of the source field.}

\begin{figure}[h]
\centerline{\scalebox{0.27}{\includegraphics{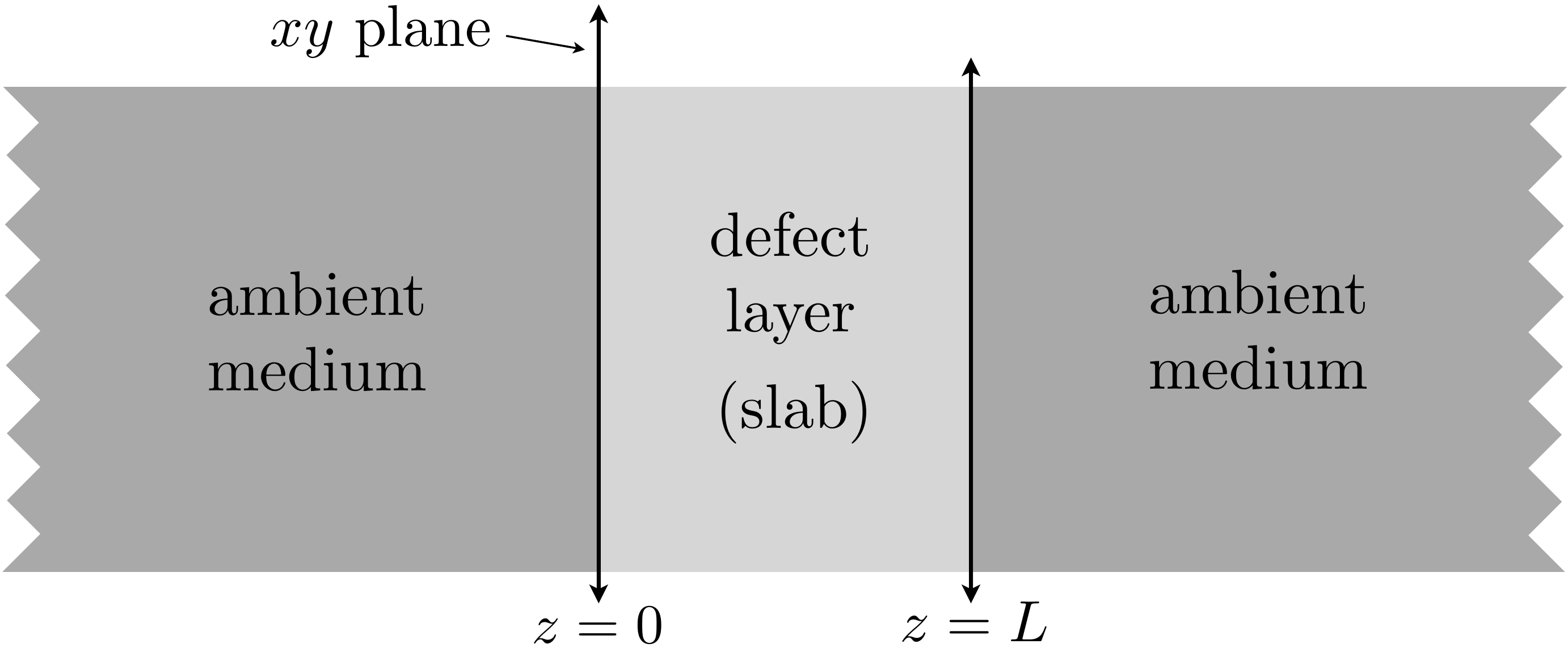}}}
\vspace{-1ex}
\caption{\small An infinite, homogeneous or periodically layered, anisotropic ambient material is interrupted by a defect layer, or slab, of length $L$ of a contrasting medium.}
\label{fig:slab}
\end{figure}

{\bfseries Modes of the ambient medium.}
In Fig.~\ref{fig:slab}, the $xy$-plane is parallel to the layers, and the $z$-axis is perpendicular.
Electromagnetic fields at a fixed frequency $\omega$ and wavevector $\kk=(k_1,k_2)$ parallel to the defect layer have the form
\begin{eqnarray*}
  \EE(x,y,z;t) &=&  [E_1(z), E_2(z), E_3(z)]^T e^{i(k_1x+k_2y-\omega t)}\,, \\
  \HH(x,y,z;t) &=& [H_1(z), H_2(z), H_3(z)]^T e^{i(k_1x+k_2y-\omega t)}\,,
\end{eqnarray*}
and the Maxwell equations reduce to a four-dimensional ODE system for the tangential components
$\psi(z)=[E_1(z),E_2(z),H_1(z),H_2(z)]^T$, described in detail in Appendix~\ref{sec:reduction},
\begin{equation*}
  \frac{d\psi}{dz} = iJA(z,\kk,\omega)\psi\,.
\end{equation*}
The $z$ dependence of $A$ comes through the dielectric and magnetic tensors $\epsilon(z)$ and $\mu(z)$.  In general, these are allowed to be periodic in the ambient medium, but, in the example we construct, they are constant both in the ambient medium and in the defect layer.

At a typical pair $\kw$, the ambient medium admits four modes---solutions of the form $\psi(z)=\psi_0^{ik_3z}$, where $k_3$, the wavenumber perpendicular to the layer, is an eigenvalue of $JA\kw$.

\medskip
{\bfseries Spectrally embedded guided slab modes.}
The strategy for creating a spectrally embedded guided mode is to choose $\epsilon$, $\mu$, $\kk$, and $\omega$ so that $JA\kw$ in the ambient medium admits two real wavenumbers $\pm k_3^{0p}\in\RR$, corresponding to propagating modes, and two imaginary wavenumbers $\pm k_3^{0e}\in i\RR$, corresponding to exponential modes, and then to arrange the defect layer just right so that the exponentially growing mode to the left of the defect matches the exponentially decreasing mode to the right (see Fig.~\ref{fig:fields}, top).  Because the ambient medium admits a propagating mode, $\omega$ will lie within the continuous spectrum for the given wavevector $\kk$.

Anisotropy is the key to creating simultaneous propagating and exponential modes.
Let the electric and magnetic tensors $\epsilon^{(0)}$ and $\mu^{(0)}$ of the ambient medium be
\begin{equation*}
  \epsilon^{(0)} = \left[
  \begin{array}{ccc}
      \epsilon_1 & 0 & 0 \\
      0 & \epsilon_2 & 0 \\
      0 & 0 & 1
   \end{array}
   \right],
  \qquad
  \mu^{(0)} = \left[
  \begin{array}{ccc}
      \mu_1 & 0 & 0 \\
      0 & \mu_2 & 0 \\
      0 & 0 & 1
   \end{array}
   \right],
   \qquad \text{(ambient space)}
\end{equation*}
and let us consider fields that propagate parallel to the $xz$-plane, which means $k_2=0$.
One computes the wavenumbers
\begin{equation*}
  k_3^{0e} = \left[ \epsilon_1 \left( \frac{\omega^2}{c^2} \mu_2 - k_1^2 \right) \right]^{1/2},
   \qquad
   k_3^{0p} = \left[ \mu_1 \left( \frac{\omega^2}{c^2} \epsilon_2 - k_1^2 \right) \right]^{1/2}.
   \qquad (k_2=0)
\end{equation*}
Their associated eigenspaces are given by the relations
\begin{eqnarray*}
  \big\{ {\textstyle-\frac{\omega}{c}}\epsilon_1 E_1 \pm k_3^{0e} H_2 \,=\, 0\,,\;
    E_2 = 0\,,\;
    H_1 = 0 \big\} & \text{for} & \pm k_3^{0e}\,, \\
 \big\{ {\textstyle\frac{\omega}{c}} \mu_1 H_1 \pm k_3^{0p} E_2 \,=\, 0\,,\;
    H_2 = 0\,,\;
    E_1 = 0 \big\} & \text{for} & \pm k_3^{0p}\,,
\end{eqnarray*}
which place them in mutually orthogonal polarizations.
Within a certain $(k_1,\omega)$ region, one has $k_3^{0e}=i|k_3^{0e}|$ and $k_3^{0p}=|k_3^{0p}|$.

The trick to matching the evanescent modes across the defect layer is to build it of the same material, but rotated by a right angle in the $xy$-plane so that its material tensors $\epsilon^{(1)}$ and $\mu^{(1)}$ are
\begin{equation*}
  \epsilon^{(1)} = \left[
  \begin{array}{ccc}
      \epsilon_2 & 0 & 0 \\
      0 & \epsilon_1 & 0 \\
      0 & 0 & 1
   \end{array}
   \right],
  \qquad
  \mu^{(1)} = \left[
  \begin{array}{ccc}
      \mu_2 & 0 & 0 \\
      0 & \mu_1 & 0 \\
      0 & 0 & 1
   \end{array}
   \right].
   \qquad \text{(defect layer)}
\end{equation*}
The $z$-directional wavenumbers in this medium are $\big\{ k_3^{1p}, -k_3^{1p}, k_3^{1e}, -k_3^{1e} \big\}$, given by
\begin{equation*}
    k_3^{1p} = \left[ \epsilon_2 \left( \frac{\omega^2}{c^2} \mu_1 - k_1^2 \right) \right]^{1/2},
    \qquad
    k_3^{1e} = \left[ \mu_2 \left( \frac{\omega^2}{c^2} \epsilon_1 - k_1^2 \right) \right]^{1/2},
\end{equation*}
and their associate eigenspaces are given by the relations
\begin{eqnarray*}
   \big\{ {\textstyle-\frac{\omega}{c}}\epsilon_2 E_1 \pm k_3^{1p} H_2 \,=\, 0\,, \;
    E_2 = 0\,,\;
    H_1 = 0 \big\} & \text{for} & \pm k_3^{1p}\,, \\
   \big\{ {\textstyle\frac{\omega}{c}} \mu_2 H_1 \pm k_3^{1e} E_2 \,=\, 0\,,\;
    H_2 = 0\,,\;
    E_1 = 0 \big\} & \text{for} & \pm k_3^{1e}\,.
\end{eqnarray*}

Figure~\ref{fig:resonance} (left) shows the dispersion relations ($\omega$ {\itshape vs.}~real $k_3$) for the propagating modes of the ambient space (superscript~$0$) and the defect layer (superscript~$1$) for hypothetical material coefficients.  In the frequency interval ${I}$ indicated in the figure, the modes $0e$ and $1e$ are exponential and the modes $0p$ and $1p$ are propagating.  This situation is attained under the condition
\vspace{-0.5ex}
\begin{equation*}
  \max\{\epsilon_1,\mu_2\} \leq \min\{\epsilon_2,\mu_1\}
  \quad \text{and} \quad
  k_1\not=0.
  \qquad \text{(assuming $k_2=0$)}
\end{equation*}
For frequencies in the interval~${I}$, the vector span of the ambient exponential modes ($0e$) coincides with that of the propagating modes in the slab ($1p$).  This allows the construction of guided modes by matching evanescent fields outside the slab with oscillatory fields in the slab:
\begin{equation*}
\begin{array}{ccll}
\left[
  \begin{array}{c}
    E_1 \\ E_2 \\ H_1 \\ H_2
  \end{array}
\right]
&=&  
C_1 \left[
  \begin{array}{c}
    -k_3^{0e} \\ 0 \\ 0 \\ \frac{\omega}{c}\epsilon_1
  \end{array}
\right]
e^{-ik_3^{0e}z},
&
\quad
z<0\,,
\quad
\text{(leftward evanescent)} \\
\vspace{-1ex}\\
&=&
B_1 \left[
  \begin{array}{c}
    k_3^{1p} \\ 0 \\ 0 \\ \frac{\omega}{c}\epsilon_2
  \end{array}
\right]
e^{ik_3^{1p}z} +
B_2 \left[
  \begin{array}{c}
    -k_3^{1p} \\ 0 \\ 0 \\ \frac{\omega}{c}\epsilon_2
  \end{array}
\right]
e^{-ik_3^{1p}z},
&
\quad
0<z<L\,,
\quad
\text{(oscillatory)} \\
\vspace{-1ex}\\
&=&
C_2 \left[
  \begin{array}{c}
    k_3^{0e} \\ 0 \\ 0 \\ \frac{\omega}{c}\epsilon_1  
  \end{array}
\right]
e^{ik_3^{0e}(z-L)},
&
\quad
L<z\,.
\quad
\text{(rightward evanescent)}
\end{array}
\end{equation*}
By imposing continuity of this solution at the interfaces $z=0$ and $z=L$, one obtains
\begin{equation}\label{trapped}
  {2\cos(k_3^{1p}L) - i\left( \frac{k_3^{0e}}{k_3^{1p}}\frac{\epsilon_2}{\epsilon_1} - \frac{k_3^{1p}}{k_3^{0e}}\frac{\epsilon_1}{\epsilon_2} \right)\sin(k_3^{1p}L) \ = \ 0.}
  \qquad \text{(guided-mode condition)}
\end{equation}
When plotted in the $\omega$-$L$ plane, this relation has multiple branches, which are shown in Fig.~\ref{fig:resonance} (middle) for $\kk=(0.5,0)$.  For $L=7$, for example, there are four frequencies in $I$ that admit a guided mode.  The first and third modes are plotted in Fig.~\ref{fig:fields} (top).

\begin{figure}
\centerline{
\scalebox{0.45}{\includegraphics{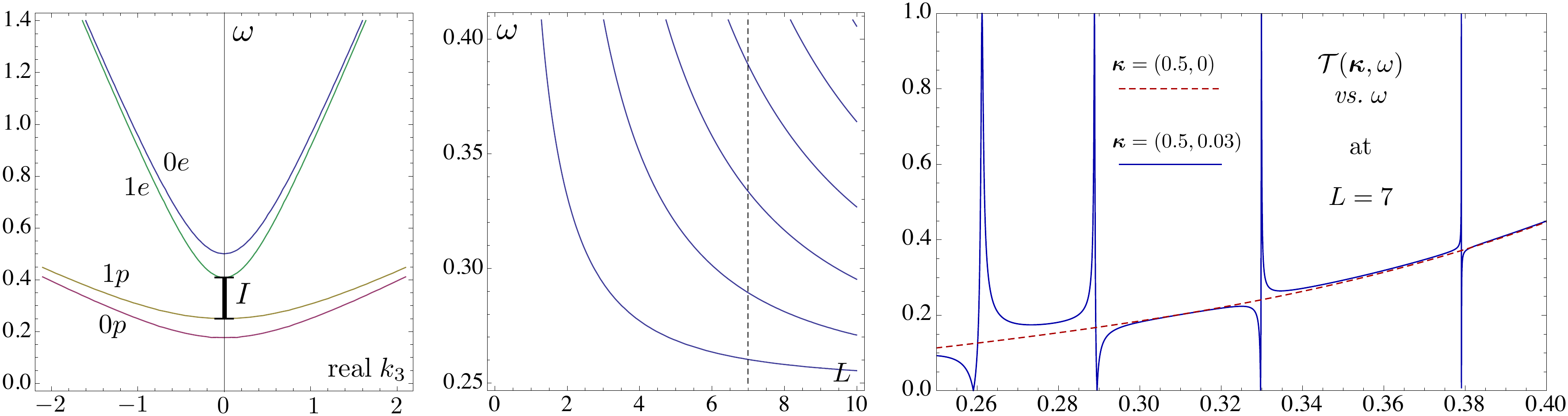}}
}
\caption{\small {\bfseries Left:}\, These dispersion relations show frequency $\omega$ {\itshape vs.}~propagation wavenumber $k_3$ (when it is real) perpendicular to the layers ($z$-direction) for the ambient medium ($0p$ and $0e$) and for the slab medium ($1p$ and $1e$).  For the material coefficients $\epsilon_1=1.5$, $\epsilon_2=8$, $\mu_1=4$, $\mu_2=1$, and wavevector $\kk_0=(k_1,k_2)=(0.5,0)$ parallel to the layers, there is a frequency interval $I\approx[0.25, 0.408248]$ in which each medium admits one propagating mode ($0p$ and $1p$) and one evanescent mode ($0e$ and $1e$).
{\bfseries Middle:}\, When the length of the slab $L$ and the frequency $\omega\in I$ satisfy this multi-branched relation, the slab supports a perfect guided mode that falls off exponentially as $|z|\to\infty$.  Because the ambient medium supports a propagating mode for $\omega\in I$, the frequency of the guided mode is embedded in the continuous spectrum.  The dotted line shows that, for $L=7$, there are four guided-mode frequencies in $I$.
{\bfseries Right:}\, The square root ${\cal T}\kw$ of the transmission is shown for a slab of length $L=7$ for $\kk=(0.5,0)$ (dotted) and $\kk=(0.5,0.03)$ (solid).  The slab at $\kk=(k_1,k_2)=(0.5,0)$ admits guided modes at four frequencies within the interval $I$ indicated by the intersection of the dashed line the middle graph with the four curves.  When $k_2$ is perturbed, sharp transmission anomalies appear near the guided-mode frequencies.}
\label{fig:resonance}
\end{figure}

\begin{center}
\scalebox{0.55}{\includegraphics{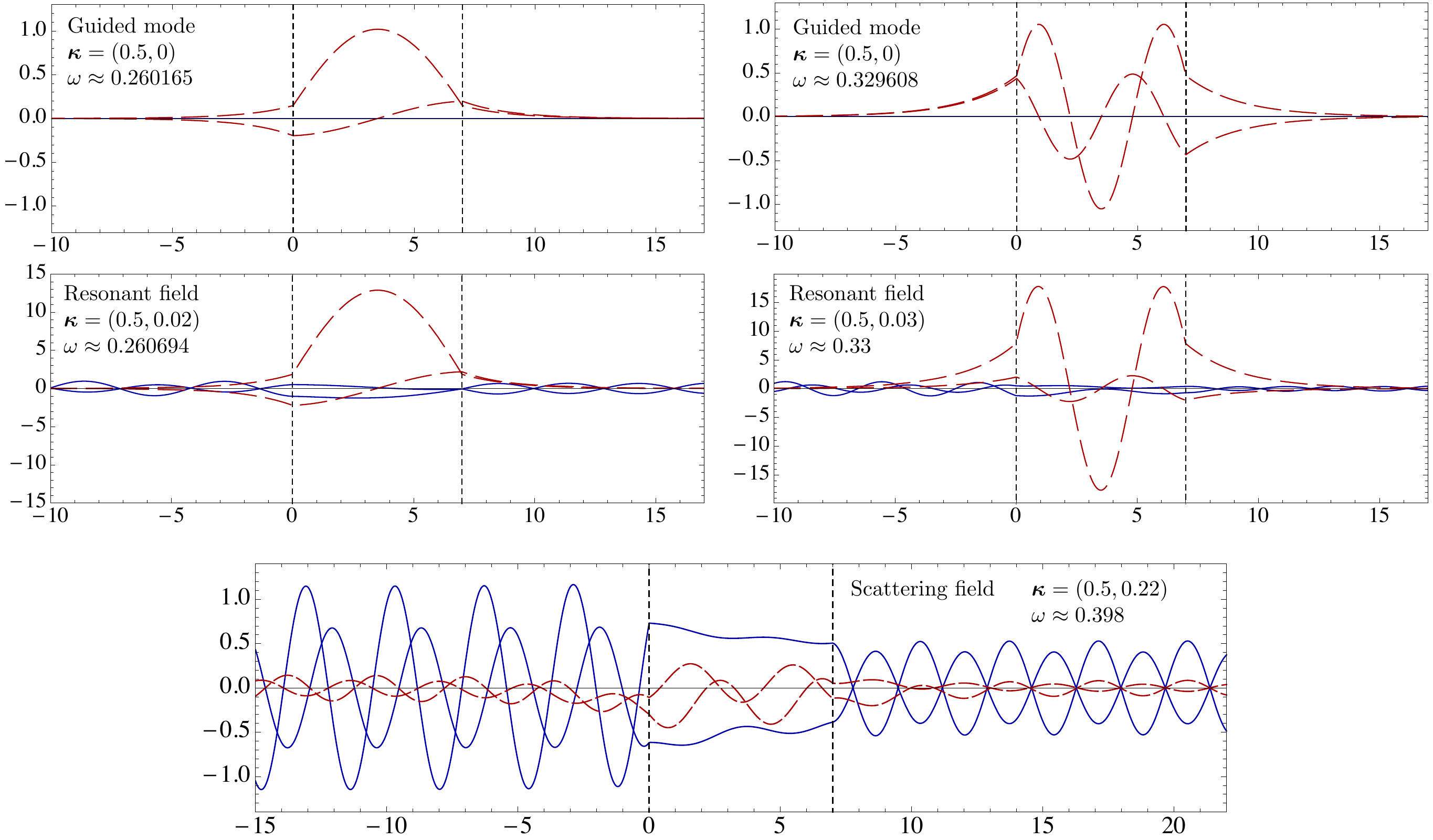}}
\captionof{figure}{\small The solid curves show the $E_2$ and $H_1$ components of electromagnetic field, and the dashed curves show the $E_1$ and $H_2$ components.  The slab (defect layer) lies between the vertical dashed lines.  {\bfseries Top.} The guided modes corresponding to the first and third of the four guided-mode frequencies indicated in Fig.~\ref{fig:resonance} for $L=7$.  {\bfseries Middle.} Resonant amplification when an incident plane wave is scattered at parameters $\kw$ close to those of a guided mode.  {\bfseries Bottom.} A non-resonant scattering field; the source field is incident is from the left.}
\label{fig:fields}
\end{center}

\medskip
{\bfseries Scattering of a propagating mode and transmission anomalies.}
When an electromagnetic mode of the ambient medium impinges upon the slab, say from the left, it is scattered, resulting in a transmitted field on the right and a reflected field on the left:
\begin{equation}\label{leftscattering}
  \psi(z) =
  \left\{
  \begin{array}{rl}
  v_{+p}\,e^{ik_3^{0p}z} + r_{-p}\,v_{-p}\,e^{-ik_3^{0p}z} + r_{-e}\,v_{-e}\,e^{-ik_3^{0e}z}\,, & (z<0)\\
  \vspace{-1.5ex}\\
                                 t_{+p}\,v_{+p}\,e^{\,ik_3^{+0p}z} + t_{+e}\,v_{+e}\,e^{\,ik_3^{+0e}z}\,. & (z>L)
  \end{array}
  \right.
\end{equation}
Here, $v_{\pm p,e}$ are the eigenvectors corresponding to the eigenvalues $\pm k_3^{0p,e}$ and $v_\rtp e^{ik_3^{0p}z}$ is the incident field.  The resulting total field $\psi(z)$ is called a scattering field; some are shown in Fig.~\ref{fig:fields} (middle, bottom).

The {\itshape transmission} ${\cal T}\kw^2$ is the ratio of energy flux of the transmitted field to that of the incident field, and it is equal to $|t_\rtp|^2$; its value lies in the interval $[0,1]$.  
Figure~\ref{fig:resonance} (right) shows ${\cal T}\kw$ {\itshape vs.} $\omega$ as $\omega$ traverses the interval $I$, which contains four guided mode frequencies for the wavevector $\kk=\kk_0=(0.5,0)$.  If $\kk$ is set exactly to $\kk_0$, the graph of ${\cal T}\kw$ is smooth, and the guided mode frequencies cannot be detected.  If $k_2$ is perturbed from $0$, the construction of perfect, exponentially decaying guided modes at real frequencies breaks down.  The destruction of a guided mode at a frequency $\omega_0$ is marked by a sharp anomaly in the transmission graph, characterized by a peak of $100\%$ transmission and a dip of $0\%$ transmission at frequencies separated by a spectral deviation of on the order of $(k_2)^2$.  
This feature is often called a Fano transmission resonance or a Fano lineshape and is accompanied by the excitation of a ``guided resonance" along the slab.

Another way to excite a guided resonance is to rotate the slab so that the two polarizations no longer match exactly at the interface between the ambient medium and the slab, so that again the construction of a perfect guided mode breaks down.  This results in a similar transmission graph to that in Fig.~\ref{fig:resonance}.

\medskip
{\bfseries Controling resonances.}
The components of the parallel wavevector $\kk=(k_1,k_2)$ independently control the central frequency and the width of a transmission resonance.  This is because a perfect guided mode is {\em robust} with respect to $k_1$ and {\em non-robust} with respect to $k_2$.  As $k_1$ varies, the mode does not couple with radiation and retains its infinite Q-value, but its frequency changes.  As $k_2$ varies, the conditions for a perfect mode are destroyed and the mode becomes leaky, but the central frequency of the resonance remains fixed (to leading order).

To be concrete, let us consider the pair $\kk_0=(0.5,0)$ and a frequency $\omega_0$ corresponding to Fig.~\ref{fig:resonance}, for which a guided slab mode can be constructed.  The loci of $100\%$ and $0\%$ transmission in real $\kw$ space near $\kwz$ are given by power series in $\tilde k_1=k_1\!-\!0.5$ and~$k_2$ for~$k_2\not=0$,
\begin{eqnarray*}
  {\cal T}\kw = 1 &\iff& \omega = \omega_\text{\tiny max}(\kk) := \omega_0 - \ell_1\tilde k_1 - \ell_{21}\tilde k_1^2 - r_2 k_2^2 + \dots \,,\\
  {\cal T}\kw = 0 &\iff& \omega = \omega_\text{\tiny min}(\kk) := \omega_0 - \ell_1\tilde k_1 - \ell_{21}\tilde k_1^2 + t_2 k_2^2 + \dots \,,
\end{eqnarray*}
with $r_2\not=t_2$, as illustrated in Fig.~\ref{fig:transmission} (left), and the ellipses indicate $O(|(\tilde k_1,k_2)|^3)$.
In this example, all the coefficients are real---$\omega$ is a real-analytic function of $\kk$---and when $k_2$=0, the expressions for $\omega_\text{\tiny max}(\kk)$ and $\omega_\text{\tiny min}(\kk)$ are identical and the anomaly reduces to the single $k_1$-dependent frequency of the guided mode.  The reality of the coefficients can be proved for a slab that is symmetric in~$z$~\cite{ShipmanTu2012}.

The wavenumber $k_1$ controls the central frequency of an anomaly by shifting the frequency $\omega_0$ of the guided mode.  Thus no anomaly emerges with a perturbation of $k_1$; the spectral width of the resonance remains zero.
On the other hand, if $k_2$ is perturbed from $0$, a transmission anomaly of width on the order of $|t_2-r_2|(k_2)^2$ opens up.  The central frequency of the anomaly is, to linear order, unchanged because of the symmetric dependence of the scattering problem on~$k_2$ coming from the reflection symmetry of the structure and the Maxwell equations under the map $(x,y,z)\to(x,-y,z)$.

\medskip
{\bfseries Complex dispersion relation and generalized guided modes.}
As we have discussed, the perfect guided slab mode we have constructed is an eigenstate at a frequency $\omega_0$ embedded in the continuous spectrum for the Maxwell equations corresponding to a fixed wavenumber, say $\kk_0=(0.5,0)$ as in the figures.  One can consider the mode to be a finite-energy state, or a bound state, when viewed as a function of $z$ alone, forgetting its infinite extent in $x$ and $y$.

The destruction of the perfect guided mode under a perturbation of $k_2$ from $0$ is a manifestation of the instability of embedded eigenvalues under generic perturbations of a system.  The annihilation of a positive eigenfrequency corresponds to a pole of a scattering matrix moving off of the real $\omega$ axis as it attains a small negative imaginary part, and this marks the onset of resonance.  The complex poles are called ``scattering resonances".
They have a long history in scattering theory; see for example~\cite[Vol.\,IV, \S\,XXII]{ReedSimon1980d} on the Auger states and~\cite{Zworski1999,Zworski2011}.

A relation $D\kw=0$ that defines the (complex) frequency of a scattering matrix parameterized by~$\kk$ is known as the dispersion relation for generalized guided slab modes.  It is depicted in Fig.~\ref{fig:ComplexDispersion} for real~$\kk$.  A~guided mode corresponding to a real pair $\kwz$ is always exponentially confined to the slab and has no attenuation temporally or along the slab---it is a {\em perfect guided mode}, experiencing no damping (infinite quality factor).  A generalized guided mode, where either $\kk$ or $\omega$ has a nonzero imaginary part, has either temporal attenuation or spatial attenuation along the slab.  These modes underly the theory of ``leaky modes" as described, for example, in \cite{HausMiller1986,TikhodeevYablonskiMuljarov2002,PengTamirBertoni1975}.

In section~\ref{sec:scattering} of this paper, we analyze the case of real parallel wavevectors, i.e., for $\kk\in\RR^2$.  A generalized guided mode for real~$\kk$ and complex~$\omega$ always has exponential growth in $|z|$ away from the slab.  This spatial growth is well known in scattering theory, as in the Helmholtz resonator~\cite{Beale1973} or the Lamb model of a spring-mass system attached to a string,~\cite{Lamb1900}.

The analysis of transmission anomalies in section~\ref{sec:scattering} centers on perturbation of the scattering matrix around a {\em real} point $\kwz$ satisfying the dispersion relation $D\kwz=0$.  In the example of this section, one has, near the guided mode parameters $\kwz$,
\begin{eqnarray*}
  D\kw=0 &\iff& \omega = \omega_g(\kk) := \omega_0 - \ell_1\tilde k_1 - \ell_{21}\tilde k_1^2 - \ell_{22} k_2^2 + \dots \,
\end{eqnarray*}
with $\ell_1$ and $\ell_{21}$ real valued.
Observe that {\em the linear part of the the frequency $\omega$ of a generalized guided mode, as a function of $\kk-\kk_0=(\tilde k_1,k_2)$, is real and coincides with the loci of $100\%$ and $0\%$ transmission}.  This is proved in general in section~\ref{sec:transmission} following \cite{PtitsynaShipmanVenakides2008,Shipman2010}.  In the example of this section, $\omega$ is real when $k_2=0$.  In addition, $\Im\ell_{22}>0$, which means that when $k_2$ is perturbed from $0$, $\omega$ enters the lower-half complex plane, becoming a scattering resonance.

Denote by $\omega_\text{\tiny cent}$ the real part of the generalized guided-mode frequency as a function of real $\kk$,
\begin{equation*}
  \omega_\text{\tiny cent}(\kk) \,:=\, \omega_0 - \ell_1\tilde k_1 - \ell_{21}\,\tilde k_1^2 - \Re\ell_{22}\, k_2^2 + \dots.
\end{equation*}
{\em The central frequency $\omega_\text{\tiny cent}$ lies between the frequencies of 0\% and 100\% transmission}.  This is seen through the relation
\begin{equation*}
  (r_2-\Re\ell_{22})(t_2-\Re\ell_{22}) = -(\Im\ell_{22})^2\,.
\end{equation*}
Thus $\Im\ell_{22}$ is the geometric average of the differences $\left| r_2-\Re\ell_{22} \right|$ and $\left| t_2-\Re\ell_{22} \right|$, and
Proposition~\ref{prop:coefficients}(g) expresses them explicitly in terms of $\Im\ell_{22}$.

\medskip
{\bfseries Fano-type resonance.}
The formulas for anomalies presented in section~\ref{sec:anomalies} generalize the Fano peak-dip shape and were first proved in~\cite{ShipmanVenakides2005,PtitsynaShipmanVenakides2008} in the case of a one-dimensional parallel wavevector.  The transmission has the form
\begin{equation*}
  {\cal T}\kw^2 \,=\, t_0^2\,\frac{\,\left| \varpi + (t_2-\Re\ell_{22})k_2^2 + \dots\right|^2 \,}{\left| \varpi + ik_2^2\,\Im\ell_{22}+\dots \right|^2 }(1+\dots)\,,
\end{equation*}
in which $\varpi=\omega-\omega_\text{\tiny cent}(\kk)$ is the deviation of $\omega$ from the center of the resonance and the ellipses indicate higher-order terms.  Ignoring the higher-order terms, one obtains the Fano lineshape~\cite{Fano1961}
\begin{equation}\label{Fano}
  {\cal T}\kw^2 \,\approx\, \frac{(q+e)^2}{1+e^2}\,,
\end{equation}
in which $q$ and $e$ are defined through
\begin{equation*}
  \renewcommand{\arraystretch}{1.1}
\left.
  \begin{array}{ll}
    \varpi = \omega-\omega_\text{\tiny cent}(\kk) & \text{(deviation from central frequency)} \\
    \vspace{-2ex}\\
    \Gamma = 2k_2^2\,\Im\ell_{22}>0 & \text{(resonance width)} \\
    \vspace{-2ex}\\
    e = \displaystyle\frac{\varpi}{\Gamma/2} & \text{(normalized frequency)} \\
    \vspace{-2ex}\\
    q = \displaystyle\frac{t_2-\Re\ell_{22}}{\Im\ell_{22}} & \text{(asymmetry parameter).}
  \end{array}
\right.
\end{equation*}
The relation between the width of the resonance and the imaginary part of $\ell_{22}$ is a form of the Fermi Golden Rule~\cite[\S 12.6]{ReedSimon1980d}.

The term ``Fano resonance" originated as a description of peak-dip features of atomic and molecular spectra.  It is characterized universally by the coupling between a mode of a structure and radiation states, which results in extreme sensitivity of scattered fields around the frequency of the mode.  There are several formulas in the literature based on heuristics of mode-radiation coupling 
\cite{DurandPaidarovaGadea2001,FanJoannopoul2002,FanSuhJoannopoul2003}, including that of Fano (\ref{Fano})~\cite{Fano1961}.
The approach in~\cite{ShipmanVenakides2005} and in this paper assumes only the underlying equations (the Maxwell equations here) and proves rigorous formulas for scattering anomalies from them.  The analysis centers around the perturbation of poles of a scattering matrix (here the frequencies $\omega_g$ of a guided mode) around a pole on the real $\omega$ axis, as the system parameters are varied ($\kk$ or structural parameters).  This is an expression of universal applicability of the formulas to linear systems in electromagnetics, acoustics, and other continuous and discrete~\cite{PtitsynaShipman2012,ShipmanRibbeckSmith2010} systems.

\medskip
{\bfseries Quality factor.}
The quality factor (Q-factor) of a resonant mode tends to infinity as the mode approaches a perfect guided mode with no damping.
It can be defined in terms of the generalized guided mode associated with the frequency $\omega_g$, as the ratio of the {\em energy stored} in the mode within a volume to the {\em energy dissipated} from the mode within that volume in one temporal cycle.
Equation (\ref{Qfactor}) in section~\ref{sec:EnergyFluxAndDensity} allows one to relate this quantity to the resonant width and frequency,
\begin{equation}
  Q = \frac{\left|\Re\omega_g(\kk)\right|}{-2\,\Im\omega_g(\kk)} 
  \sim \frac{\,\left|\omega_0\right|\,}{\Gamma} = O(|k_2|^{-2})
  \qquad
  \text{(quality factor).}
\end{equation}

\medskip
{\bfseries Resonant field amplification.}
When an incident propagating field at real parameters $\kw$ near a pair $\kwz$ of a perfect guided slab mode is scattered, the field in the slab is highly amplified and resembles the guided mode, as shown Fig.~\ref{fig:fields} (middle).  These fields are called {\em guided resonances} \cite{FanJoannopoul2002}, and one thinks of them as a coupling between radiation (propagating modes in the ambient medium) and a guided mode.  Field amplification occurs around the real part of the frequency of the generalized guided mode $\omega=\omega_\text{\tiny cent}(\kk)$ for real perturbations $(\tilde k_1,k_2)$ from $\kk_0$, as shown in Fig.~\ref{fig:amplification} (top).

The frequency interval of amplification shrinks to the single guided-mode frequency $\omega_0$ as $k_2\to0$.  At $k_2=0$, no amplification is observed at frequencies near $\omega_0$.  This is because, at $\kk=\kk_0=(0.5,0)$ (or more generally $\kk_0\!=\!(k_1,0)$), the guided mode is a perfect, infinite-lifetime, finite-energy, spectrally embedded state and is thus decoupled from radiation states.
Along the real part of the generalized guided-mode frequency relation
%
%
$\omega = \omega_\text{\tiny cent}(\kk)$, field amplification is on the order of $c/|k_2|$, as observed in in Fig.~\ref{fig:amplification} (top).  This law is proved in a general setting in section~\ref{sec:amplification}.

\medskip
{\bfseries Scattering of an evanescent field.}
We have just seen that, when the parallel wavevector $\kk$ of an incident {\em propagating} field is set exactly to that of a real guided mode pair $\kwz$, but the frequency is allowed to vary from $\omega_0$, no high-amplitude response is excited in the slab.  On the other hand, an incident evanescent field at wavevector-frequency pairs $(\kk_0,\omega)$ produces amplification that blows up as $c/|\omega-\omega_0|$, a law proved in section~\ref{sec:amplification}, as shown in Fig.~\ref{fig:amplification} (bottom right).  This is because the scattering problem (\ref{leftscattering}), with $v_\rtp e^{ik_3^{0p}z}$ replaced by $v_\rte e^{ik_3^{0e}z}$, admits no solution (this is made precise in Theorem~\ref{thm:scattering}).

Along the relation $\omega=\omega_\text{\tiny cent}(\kk)$, field amplification is on the order of $c/|k_2|^2$, a law proved in section~\ref{sec:amplification}, as shown in Fig.~\ref{fig:amplification} (bottom left).

\begin{center}
\scalebox{0.52}{\includegraphics{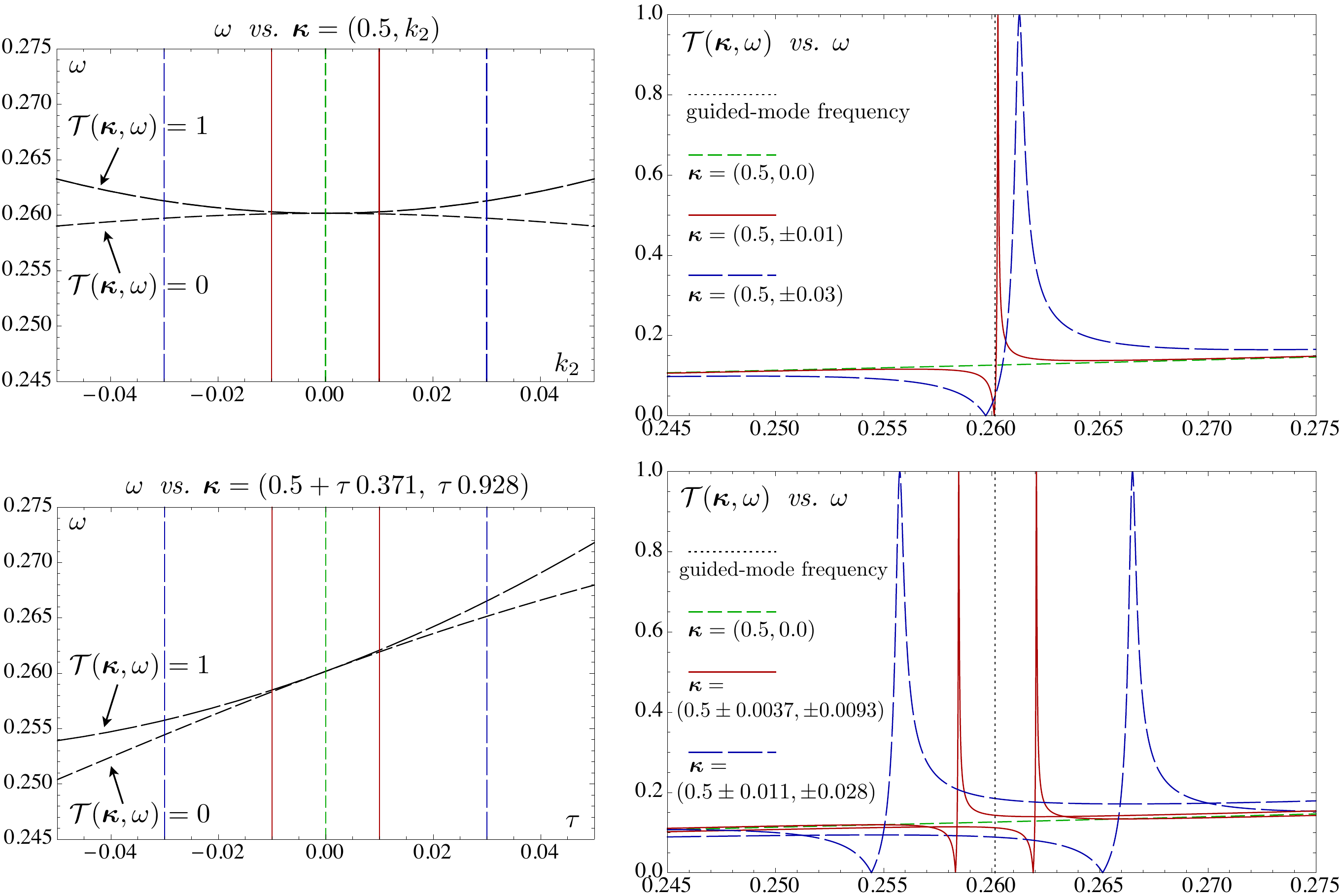}}
\captionof{figure}{\small 
The central frequency and width of transmission anomaly can be controlled by varying $\kk=(k_1,k_2)$ about the pair $\kk_0=(0.5,0)$ that admits a perfect guided mode at $\omega_0\approx0.26015$.  
The black curves in the left-hand plots are the loci of $100\%$ transmission ($a=0$) and $0\%$ transmission ($b=0$) along two vertical planes in real $(k_1,k_2,\omega)$ space.  They intersect quadratically at the guided-mode pair $\kwz$.
The vertical lines are the constant-$\kk$ lines along which the square root ${\cal T}\kw$ of the transmission is evaluated in the right-hand plots.
The wavenumber $k_1$ linearly controls the center of a resonance by changing the frequency $\omega_0$ of the perfect guided mode, and $k_2$ controls the width of the anomaly like $ck_2^2$ as the perfect guided mode is destroyed.}
\label{fig:transmission}
\end{center}

\begin{center}
\scalebox{0.37}{\includegraphics{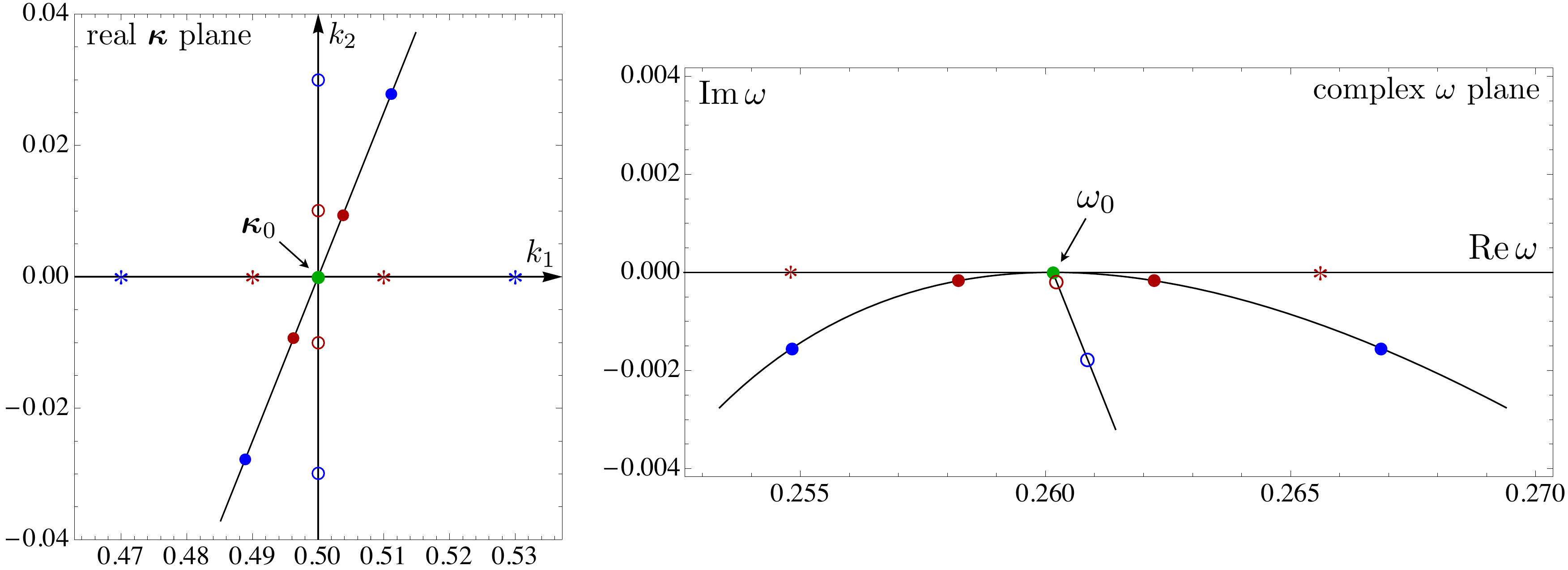}}
\captionof{figure}{\small This a depiction of the complex dispersion relation $D\kw=0$ for generalized guided modes for real wavevector $\kk\in\RR^2$ and complex frequency $\omega\in\CC$ near a real wavevector-frequency pair $\kwz\approx(0.5,0,0.26015)$ (green) for which the slab admits a true, exponentially decaying guided mode.
The relation is symmetric in $k_2$ and has the form $\omega-\omega_0 + \ell_1k_1 + \ell_{21}k_1^2 + \ell_{22}k_2^2 + \dots=0$, with $\ell_1$ real, $\Im\ell_{21}>0$ and $\Im\ell_{22}>0$.
{\bfseries (1)} A segment of the $k_1$-axis is mapped implicitly by $D\kw=0$ to the real $\omega$ axis.  This is due to the explicit construction, which decouples the two types of fields $(E_1,0,0,H_2)$ and $(0,E_2,H_1,0)$ when $k_2=0$.
{\bfseries (2)} The $k_2$-axis is mapped to a curve in the lower half $\omega$-plane emanating from $\omega_0$ in the direction of $-\ell_{22}$.  The hollow dots indicate the $\kk$ values used in Fig.~\ref{fig:transmission} (top) and the corresponding $\omega$ values satisfying $D\kw=0$.
{\bfseries (3)} The line $\kk=\tau(2/\sqrt{29},5/\sqrt{29})\approx\tau(0.371,0.928)$ is mapped to the curve in the lower half $\omega$-plane that is tangent to the real axis at $\kw$.  The solid dots indicate the $\kk$ values used in Fig.~\ref{fig:transmission} (bottom) and the corresponding $\omega$ values satisfying $D\kw=0$.}
\label{fig:ComplexDispersion}
\end{center}

\vspace{1ex}
\begin{center}
\scalebox{0.54}{\includegraphics{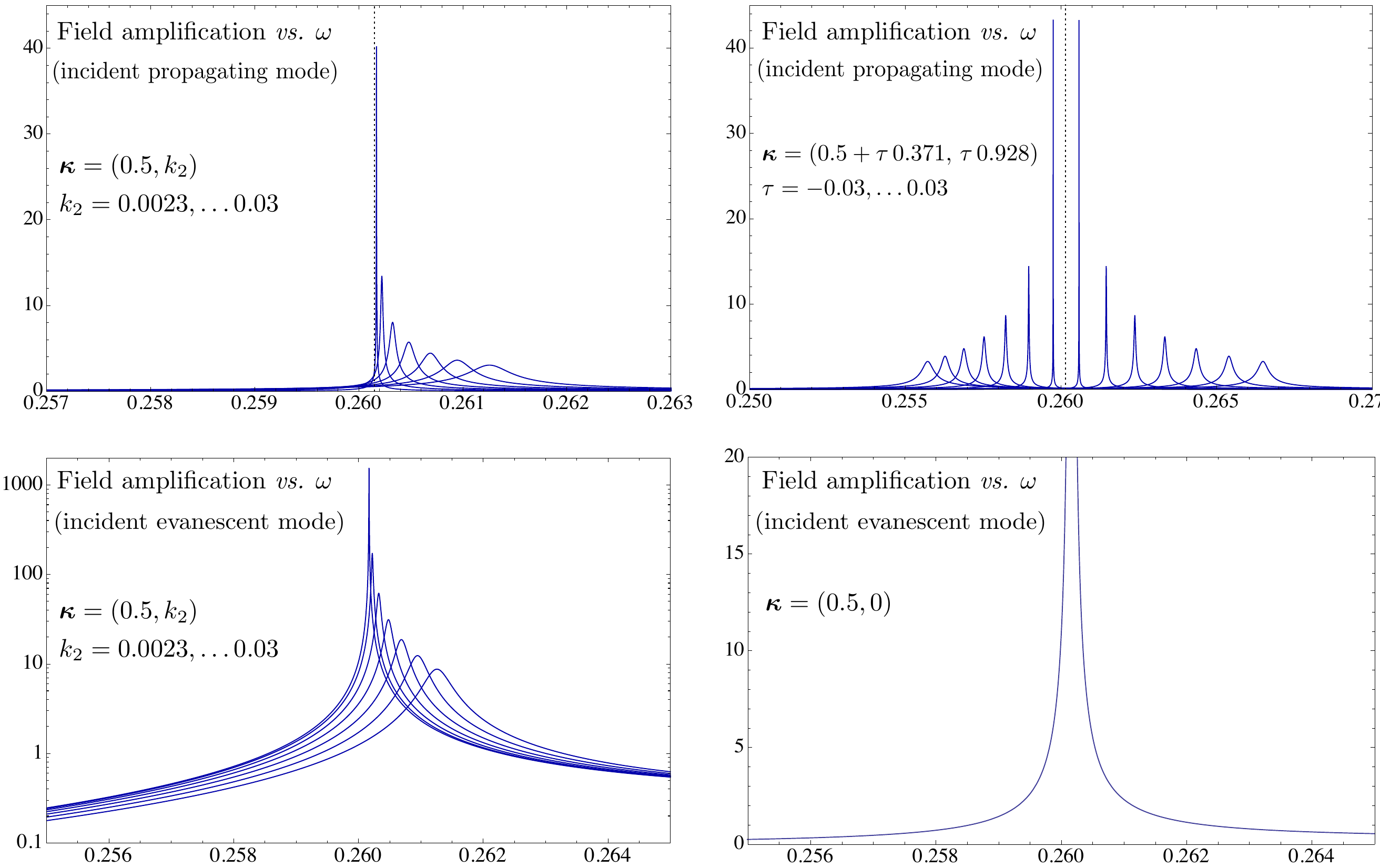}}
\captionof{figure}{\small
Resonant field amplification that occurs when an incident field strikes the slab can be measured by the modulus $\mathfrak A$ of the coefficient of the reflected evanescent mode.  Amplification occurs at wavevector-frequency pairs $\kw$ near those of a {real} pair $\kwz\approx(0.5,0,0.26015)$ that admits a guided mode.
{\bfseries Above:} The incident field is a propagating mode of the ambient space, and $\mathfrak A$ is plotted versus frequency $\omega$ for different values of $\kk$.  Maximal field amplification is of order $1/|\kk|$ as long as $k_2\not=0$, and the interval of amplification shrinks as $|\kk|^2$.  At $\kk=\kk_0$, $\mathfrak A=0$ and there is no field amplification.
{\bfseries Below:} The incident field is an evanescent harmonic, and $\log\mathfrak A$ is plotted versus frequency $\omega$ for different values of $\kk$ (left).  Maximal field amplification is of order $1/|\kk|^2$ (for $k_2\not=0$).  At $\kk=\kk_0$, field amplification is of order $1/|\omega-\omega_0|$.}
\label{fig:amplification}
\end{center}

\section{Electrodynamics in lossless layered media}\label{sec:ODEs} 

This section develops the concepts from electrodynamics of lossless layered media that will be needed for the analysis of resonant scattering.  The reader may wish to skip directly to section~\ref{sec:scattering} and refer back to this material as needed.  

Key results on the non-degeneracy of guided modes (Theorems~\ref{thm:scattering} and~\ref{thm:nondegeneracy}) require a careful treatment of energy density, flux, and velocity for time-harmonic electromagnetic fields for real and complex frequencies.
Section~\ref{sec:CanonicalODEs} reviews the reduction of Maxwell equations in linear, nondispersive, lossless anisotropic layered media to a linear ODE system, which we call the canonical Maxwell ODEs; details are relegated to Appendix~\ref{sec:reduction}.
Two new results, Theorems~\ref{thm:EnergyConservationLaw} and~\ref{thm:EnergyDensity} establish relationships between energy and the Maxwell ODEs.
Section~\ref{sec:periodic} discusses the Floquet theory for the periodic ambient medium and defines the rightward and leftward modes.

\subsection{The canonical Maxwell ODEs}\label{sec:CanonicalODEs}

The Maxwell equations for time-harmonic electromagnetic fields $(\EE(\rr),\HH(\rr),\DD(\rr),\BB(\rr))e^{-i\omega t}$ ($\omega\not=0$) in linear anisotropic media without sources are
\begin{equation}\label{MaxwellEqsTimeHar}
\mat{1.1}{0}{\nabla\times}{-\nabla\times}{0}
\col{1.1}{\EE}{\HH}
=
-\frac{i\omega}{c}
\col{1.1}{\DD}{\BB},\qquad
\col{1.1}{\DD}{\BB}
=
\mat{1.1}{\epsilon}{0}{0}{\mu}
\col{1.1}{\EE}{\HH}
\end{equation}
(in Gaussian units), where $c$ denotes the speed of light in a vacuum. We consider only non-dispersive and lossless media, which means that the dielectric permittivity $\epsilon$ and magnetic permeability $\mu$ are $3\times3$ Hermitian matrices that depend only on the spatial variable $\rr=(x,y,z)$.
A stratified medium is one for which $\epsilon$ and $\mu$ depend only on~$z$.  Thus
\begin{align}\label{LosslessLayeredMedia}
\epsilon=\epsilon(z)=\epsilon(z)^*,\;\;\mu=\mu(z)=\mu(z)^*,
\end{align}
where $*$ denotes the Hermitian conjugate (adjoint) of a matrix. 
Typically, a stratified medium consists of layers of different homogeneous materials.
We assume that each layer is passive.  Thus, for some positive constants $c_1$ and $c_2$, 
\begin{align}\label{PassiveMedia}
0<c_1I\leq \epsilon(z), \mu(z) \leq c_2 I
\end{align}
for all $z\in\RR$, where $I$ denotes the $3\times 3$ identity matrix.  The positive-definiteness of the material tensors has implications for energy density and flux which enters the analysis of transmission anomalies through Theorem~\ref{thm:nondegeneracy} on the non-degeneracy of guided modes.


Because of the translation invariance of stratified media along the $xy$ plane, solutions of equation (\ref{MaxwellEqsTimeHar}) are sought in the form
\begin{equation}\label{TangentialPlaneWavesRepr}
\col{1.1}{\EE}{\HH}=
\col{1.1}{\EE(z)}{\HH(z)}
e^{i(k_1x+k_2y)}\,,
\end{equation}
in which
%
  $\kk=(k_1,k_2)$
%
is the tangential wavevector.
The Maxwell equations (\ref{MaxwellEqsTimeHar})) for this type of solution can be reduced to a system of ordinary differential equations for the tangential electric and magnetic field components (see Appendix~\ref{sec:reduction}) denoted by
\[
\psi(z)=\left[E_1(z), E_2(z), H_1(z),H_2(z)\right]^T\,,
\]
\begin{equation}\label{canonical}
  -iJ\, \frac{d}{dz} \psi(z) \,=\, A(z;\kk,\omega)\, \psi(z),
\end{equation}
in which
\begin{equation*}
  J=
  \renewcommand{\arraystretch}{1.2}
\left[
  \begin{array}{cccc}
    0 & 0 & 0 & 1 \\
    0 & 0 & -1 & 0 \\
    0 & -1 & 0 & 0 \\
    1 & 0 & 0 & 0
  \end{array}
\right], \quad
J^*=J^{-1}=J\,.
\end{equation*}
The $4\times 4$ matrix $A(z;\kk,\omega)$ is given by (\ref{SchurComp}) in Appendix~\ref{sec:reduction} and is a Hermitian matrix for real $\kw$, $\omega\not=0$. We will refer to the ODEs in (\ref{canonical}) as the canonical \emph{Maxwell ODEs}.

Boundary conditions for electromagnetic fields in layered media require that tangential electric and magnetic field components $\psi(z)$ be continuous across the layers, which means $\psi$ is a continuous function of~$z$.
Thus solutions are those satisfying the integral equation
\begin{equation}\label{integralequation}
\psi(z)=\psi(z_0)+\int_{z_0}^z iJ^{-1}A(s;\kk,\omega)\psi(s)ds.
\end{equation}


The initial-value problem
\begin{equation}
  -iJ\, \frac{d}{dz} \psi(z) \,=\, A(z;\kk,\omega)\, \psi(z),
  \qquad
  \psi(z_0)=\psi_0
\end{equation} 
for the Maxwell ODEs (\ref{canonical}) has a unique solution
\begin{equation}
  \psi(z) \,=\, T(z_0,z)\, \psi(z_0)
\end{equation}
for each initial condition $\psi_0\in \CC^4$.   The $4\times 4$ matrix $T(z_0,z)$ is called the \emph{transfer matrix}.
It satisfies
%
\begin{align}
T(z_0,z)=T(z_1,z)T(z_0,z_1),\quad T(z_0,z_1)^{-1}=T(z_1,z_0),\quad T(z_0,z_0)=I,
\end{align}
for all $z_0,z_1,z\in \RR$.  As a function of $z$ it is continuous and satisfies the integral equation~(\ref{integralequation}) with $T(z_0,z)$ in place of $\psi(z)$.
%
%
As a function of the wavevector-frequency pair $(\kk,\omega)\in \CC^2\times \CC\setminus\{0\}$, it is analytic.
Perturbation analysis of analytic matrix-valued functions and their spectrum is central to the study of scattering problems, particularly 
those involving guided modes, as in this work, or slow light, as discussed in \cite{FigotinVitebskiy2006,Welters2011,Welters2011a,ShipmanWelters2012}, for instance.

\subsection{Electromagnetic energy flux and energy density}\label{sec:EnergyFluxAndDensity}

Two of the most important attributes of electrodynamics of layered media are energy flux and energy density.
The energy-conservation law for electromagnetics, known as Poynting's theorem, 
for time-harmonic fields $\EE(\rr,t)=\EE(\rr)e^{-i\omega t}$, $\HH(\rr,t)=\HH(\rr)e^{-i\omega t}$ with complex frequency $\omega\not=0$
%
%
in {\em linear, non-dispersive, lossless} media is
\begin{eqnarray}\label{APoyntingsThm}
  -\frac{d}{dt}\int_{V}U\left(\rr,t\right)d^3\rr
  \,=\, \int_{V}\nabla \cdot \operatorname{Re}\SS\left(\rr,t\right) d^3\rr
  \,=\, \int_{\partial V}\operatorname{Re}\SS\left(\rr,t\right) \cdot \mathbf{n}\, da\;.
\end{eqnarray}
The volume $V\subseteq \mathbb{R}^3$ is bounded, its boundary $\partial V$ has outward directed unit normal $\mathbf{n}$, and
\begin{equation*}
  \renewcommand{\arraystretch}{1.1}
\left.
  \begin{array}{ll}
    \SS\left( \rr,t\right)=\SS(\rr)e^{2\operatorname{Im}\omega t}\;, &
    \SS(\rr)=\displaystyle\frac{c}{8\pi }(\EE\left( \rr\right) \times \overline{\EE\left( \rr\right) })\;; \\
    \vspace{-1.5ex}\\
    U\left( \rr,t\right)=U(\rr)e^{2\operatorname{Im}\omega t}\;, &
    U(\rr)=\displaystyle\frac{1}{16\pi } \big( \left(\EE\left( \rr\right) ,\DD\left( \rr\right) \right) +\left( \BB\left( \rr\right) ,\HH\left( \rr\right) \right)
\big).
  \end{array}
\right.
\end{equation*}
The overline denotes complex conjugation and $(\cdot,\cdot)$ is the usual inner product in $\CC^3$. As in the case of purely oscillatory fields, i.e., $\operatorname{Im}\omega=0$, it is fruitful to make an interpretation of \ref{APoyntingsThm}. With the real vector $\operatorname{Re}\SS(\rr,t)$ interpreted as energy flux associated with the field and the real scalar $U(\rr,t)$ interpreted as energy density, equality \ref{APoyntingsThm} is then interpreted as an energy conservation law which says that the rate of decay of electromagnetic energy in a volume is a result of the outflow of electromagnetic energy through the surface of the volume.

For a damped oscillation, that is, a time-harmonic field with a complex frequency $\omega$ with $\Im\omega< 0$, the quality factor, or Q-factor, is commonly defined as the reciprocal of the relative rate of energy dissipation per temporal cycle 
\begin{equation}\label{Qfactor}
Q_\omega:=2\pi\frac{\text{Energy stored in system}}{\text{Energy lost per cycle}}=|\Re\omega|\frac{\int_{V}U\left(\rr,t\right)d^3\rr}{\int_{\partial V}\operatorname{Re}\SS\left(\rr,t\right) \cdot \mathbf{n}\, da}=\frac{|\Re\omega|}{-2\,\Im\omega} \,,
\end{equation}
where the latter equality follow from Poynting's theorem (\ref{APoyntingsThm}). For purely oscillatory fields, i.e., $\Im\omega=0$, the Q factor is defined to be $Q_\omega=+\infty$. 

For time-harmonic fields with real frequency, i.e., $\operatorname{Im}\omega=0$, this conservation law is the well known Poynting's theorem for time-harmonic fields. In this case, $\SS(\rr)$ is the complex Poynting vector and $\operatorname{Re}\SS\left(\rr\right)$, $U(\rr)$ are the time-averaged Poynting vector and total energy density, respectively, of the sinusoidally varying fields,
\begin{equation*}
  \renewcommand{\arraystretch}{1.1}
\left.
  \begin{array}{rcl}
    \operatorname{Re}\SS\left(\rr\right)
    &=& \left<c/4\pi \operatorname{Re} \EE(\rr,t)\times \operatorname{Re} \HH(\rr,t)\right> \\
    \vspace{-2ex} \\
    U\left(\rr\right)
    &=& \left<1/8\pi \left(\operatorname{Re} \EE(\rr,t)\cdot \operatorname{Re} \DD(\rr,t)+ \operatorname{Re} \BB(\rr,t)\cdot \operatorname{Re} \HH(\rr,t)\right)\right>
  \end{array}
\right.
\end{equation*}
in which $\cdot$ denotes the dot product of two vectors in $\RR^3$ and $\left<\;\right>$ denotes the time average over one period of a periodic function.  Thus \ref{APoyntingsThm} is the well known energy conservation law for the time-averaged energy flux and energy density for the sinusoidally varying fields with the same interpretation as above.

For fields of the form $[\EE(z),\HH(z)]^Te^{i(k_1x+k_2y)}$ (\ref{TangentialPlaneWavesRepr}) in a stratified medium, with $\kk=(k_1,k_2)\in \RR^2$, the energy conservation law (\ref{APoyntingsThm}) for the flow across the layers has a simplified form.  For the tangential electric and magnetic field components $\psi=\left[E_1, E_2, H_1,H_2\right]^T$, the region $V=[x_0,x_1]\times[y_0,y_1]\times[z_0,z_1]$, and the positively-directed normal vector $\mathbf{e}_3$, the energy conservation law \ref{APoyntingsThm} yields
\begin{eqnarray}\label{EnergyConsLawLayeredMedia}
-2\operatorname{Im}\omega\int_{z_0}^{z_1}U\left(z\right)dz
= \int_{z_0}^{z_1}\frac{d}{dz}(\operatorname{Re}\SS\left(z\right)\cdot \mathbf{e}_3)\,dz
= [\psi(z_1),\psi(z_1)]-[\psi(z_0),\psi(z_0)]\,,
\end{eqnarray}
in which the energy-flux form $[\cdot,\cdot]$ and the $z$-dependent Poynting vector $\SS(z)$ and energy flux $U(z)$ are
\begin{equation}\label{FluxSU}
\renewcommand{\arraystretch}{1.2}
\left.
  \begin{array}{rcl}
 [\psi_1,\psi_2] &:=& \displaystyle\frac{c}{16\pi}(J\psi_1,\psi_2),
  \qquad
  \psi_1, \psi_2\in \mathbb{C}^4\,, \\
\vspace{-2ex}\\
\SS(\rr)=\SS\left(z\right) &=& \displaystyle\frac{c}{8\pi }(\EE\left( z \right) \times \overline{\HH\left( z\right) }),\\
 \vspace{-2ex}\\
\Re \SS(\rr)\cdot\mathbf{e}_3 = \Re \SS(z)\cdot\mathbf{e}_3&=&[\psi(z),\psi(z)],\\
\vspace{-2ex}\\
U(\rr) = U\left( z\right) &=& \displaystyle\frac{1}{16\pi }\left( \left( \epsilon
\left( z\right) \EE\left( z\right) ,\EE\left( z\right) \right) +\left( \mu \left( z\right)
\HH\left( z\right) ,\HH\left( z\right) \right)
\right).    
  \end{array}
\right.
\end{equation}
where $(\cdot,\cdot)$ is the standard complex inner product in $\mathbb{C}^4$ with the convention of linearity in the second component and conjugate-linearity in the first.

We have derived an important result: Up to multiplication by  $(x_1-x_0)(y_1-y_0)e^{2\operatorname{Im}\omega t}$, the energy conservation law (\ref{EnergyConsLawLayeredMedia}) represents the differential energy flow across two parallel planes (RHS of (\ref{EnergyConsLawLayeredMedia})) in terms of the energy between these planes (LHS of (\ref{EnergyConsLawLayeredMedia})) for time-harmonic fields with complex frequency $\omega\not=0$ and real tangential wavevector $\kk$.

For damped oscillations, {\itshape i.e.,} $\operatorname{Im}\omega<0$, the positivity (\ref{PassiveMedia}) of the media make the LHS of (\ref{EnergyConsLawLayeredMedia}) positive, and so the RHS of (\ref{EnergyConsLawLayeredMedia}) may be interpreted as the outward flux of electromagnetic energy as it radiates from the interval $(z_0,z_1)$, decaying exponentially in time as $e^{2\operatorname{Im}\omega t}$.  For undamped oscillations, i.e., $\operatorname{Im}\omega=0$, one has the usual conservation of energy for lossless layered media, that is, $\Re \SS(\rr)\cdot\mathbf{e}_3=[\psi(z),\psi(z)]$ is the $z$-independent time-averaged energy flux across the layers.

\medskip
The energy flux $[\cdot,\cdot]$ is an indefinite sesquilinear form associated with the canonical ODEs (\ref{canonical}).  It will play a central role in the analysis of scattering and resonance.
%
%
The adjoint of a matrix $M$ with respect to $[\cdot,\cdot]$ is denoted by $M^\fadj$ and will be called the {\em flux-adjoint} of $M$.  It is equal to $M^\fadj=J^{-1}M^*J$, where $M^*=\overline{M}^T$ is the adjoint of $M$ with respect to the standard inner product $(\cdot, \cdot)$. 

If $\omega$ is real and nonzero and $\kk$ is real, then the matrix $A$ is self-adjoint with respect to $(\cdot,\cdot)$, i.e., $A^*=A$.
The matrix $J^{-1}A$ is self-adjoint and the transfer matrix $T$ is unitary with respect to $[\cdot,\cdot]$, i.e., for any $\psi_1, \psi_2\in \mathbb{C}^4$,
\begin{equation*}
\renewcommand{\arraystretch}{1.1}
\left.
  \begin{array}{ll}
  {} [J^{-1}A\psi_1,\psi_2] = [\psi_1,J^{-1}A\psi_2]\,, &\hspace{1em} (J^{-1}A)^\fadj = J^{-1}A\,, \\
  {} [T\psi_1,\psi_2] = [\psi_1,T^{-1}\psi_2]\,, &\hspace*{1em} T^\fadj = T^{-1}\,.
  \end{array}
\right.
\end{equation*}
The flux-unitarity of $T$ follows from the energy conservation law (\ref{EnergyConsLawLayeredMedia}) and expresses the conservation of energy in a $z$-interval $[z_0,z_1]$ through the principle of energy-flux invariance for lossless systems,
\begin{equation}\label{flux1}
  [\psi(z_1),\psi(z_1)]=[T(z_0,z_1)\psi(z_0),T(z_0,z_1)\psi(z_0)]=[\psi(z_0),\psi(z_0)]
  \quad \big((\kk,\omega) \,\text{ real}\big).
\end{equation}

If $\omega$ is complex, 
this flux is not invariant; instead one has (\ref{EnergyConsLawLayeredMedia}) which, in terms of the matrix $A$ in the Maxwell ODEs, is stated in the following theorem.

\begin{Theorem}[Energy Conservation Law]\label{thm:EnergyConservationLaw}
Let $\psi$ be a solution of the canonical Maxwell ODEs (\ref{canonical}) with nonzero frequency $\omega\in\mathbb{C}$ and real tangential wave vector $\kk\in \RR^2$.  For any $z_0,z_1\in \mathbb{R}$,
\begin{eqnarray*}
[\psi(z_1),\psi(z_1)]-[\psi(z_0),\psi(z_0)]=-\frac{c}{8\pi}\int_{z_0}^{z_1}\left(\psi(z),\operatorname{Im}(A(z;\kk,\omega))\psi(z)\right)dz=-2\operatorname{Im(\omega)}\int_{z_0}^{z_1}U\left(z\right)dz.
\end{eqnarray*}
\end{Theorem}

\begin{proof}
By the fact $\psi$ is a solution to the Maxwell ODEs (\ref{canonical}) then it is an integrable solution and hence satisfies for almost every $z\in\mathbb{R}$ the differential equations \ref{canonical}. Together with the definition of the energy-flux form in (\ref{FluxSU}) this implies for a.e.\ $z\in\mathbb{R}$,
\begin{eqnarray}
\frac{\partial}{\partial z}\left[\psi(z),\psi(z)\right]&=&\left[\frac{\partial}{\partial z}\psi(z),\psi(z)\right]+\left[\psi(z),\frac{\partial}{\partial z}\psi(z)\right]\notag\\
&=&\frac{c}{16\pi}\left(\left(iA(z;\kk,\omega)\psi(z),\psi(z)\right)+\left(\psi(z),iA(z;\kk,\omega)\psi(z)\right)\right)\notag\\
&=&-\frac{2c}{16\pi}\left(\psi(z),\frac{A(z;\kk,\omega)-A(z;\kk,\omega)^*}{2i}\psi(z)\right)\notag\\
&=&-\frac{c}{8\pi}\left(\psi(z),\operatorname{Im}(A(z;\kk,\omega))\psi(z)\right). \label{one}
\end{eqnarray}
with the LHS having the antiderivative the function $\left[\psi(z),\psi(z)\right]$.
On the other hand, by (\ref{EnergyConsLawLayeredMedia}),
%
%
the result follows by integrating (\ref{one}).
\end{proof}

Using Theorem \ref{thm:EnergyConservationLaw} one derives the following formula for the energy density $U(z)$ of any time-harmonic electromagnetic field with real $\kw$ in terms of its tangential components $\psi(z)$.

\begin{Theorem}[Energy Density]\label{thm:EnergyDensity}
Let $\psi$ be a solution of the canonical Maxwell ODEs (\ref{canonical}) with nonzero real frequency $\omega\in\RR$ and real tangential wave vector $\kk\in \RR^2$.  Then the energy density $U(z)$ is given by 
%
%
%
\begin{eqnarray*}
  U\left(z\right) \,=\, \frac{c}{16\pi} \left(\psi(z),\,\frac{\partial A}{\partial \omega}(z;\kk,\omega)\psi(z)\right)\,.
\end{eqnarray*}
\end{Theorem}

\begin{proof}
Let $\psi$ be a solution of the canonical Maxwell ODEs (\ref{canonical}) with $\kk\in \RR^2$ and nonzero $\omega_0\in\RR$. We will show that 
\begin{eqnarray*}
  \frac{c}{16\pi}\left(\psi(z),\frac{\partial A}{\partial \omega}(z;\kk,\omega_0)\psi(z)\right)=U\left(z\right)
\end{eqnarray*}
for a.e.\ $z\in\mathbb{R}$.  We will prove this, suppressing the explicit dependence on $\kk$.

The idea is to use a type of limiting absorption principle with the integral form of the energy conservation law, that is, to consider Theorem \ref{thm:EnergyConservationLaw} for time-harmonic fields with frequency $\omega=\omega_0+i\eta$ with nonzero $\omega_0\in \RR$, $0<\eta\ll 1$ and take the limit as $\eta\to 0$. 

Because $\epsilon$ and $\mu$ depend only on $z$, it is proved in Section \ref{sec:reduction} that
$
A\in \mathcal{O}(\mathbb{C}\setminus \left\{ 0\right\} ,M_{4}\left( L^{\infty}\left( \mathbb{R}\right) \right))
$,
where $\mathcal{O}$ denotes holomorphic functions.
Thus for any $\omega _{0}\in \mathbb{R}\mathcal{\setminus }\left\{ 0\right\}$ there exist coefficient functions of $z$, $\left\{ A_{j}\left(\cdot \right) \right\} _{j=0}^{\infty }\subseteq M_{4}\left( L^{\infty}\left( \mathbb{R}\right) \right)$, such that the power series
\[
A\left( \cdot\, ;\omega \right) =\sum_{j=0}^{\infty }A_{j}\left( \cdot \right)
\left( \omega -\omega _{0}\right) ^{j},\text{ \ \ }\left\vert \omega -\omega
_{0}\right\vert \ll 1
\]
converges uniformly and absolutely in $M_{4}\left(L^{\infty}\left(\mathbb{R}\right) \right) $ to $A\left( \cdot\, ;\omega \right) $.  This
implies that in $M_{4}\left( \mathbb{C}\right) $ with the standard Euclidean metric the partial derivative $\frac{\partial A}{\partial \omega }\left( z;\omega_0 \right)=\lim_{\omega\to \omega_0}\frac{1}{\omega-\omega_0}\left(A(z;\omega)-A(z;\omega_0)\right)$
exists for a.e. $z\in \mathbb{R}$ and 
\[
A_{1}\left( z\right) =\frac{\partial A}{\partial \omega }\left( z;\omega_0 \right)\text{ for a.e. }z\in 
\mathbb{R}.
\]
Moreover, by the Hermitian properties (\ref{MatrixAinMaxwellODEsIsHermitian}) of $A$ it follows that  
$
A_{j}\left( z\right) ^{\ast }=A_{j}\left( z\right) \text{ for a.e. }z\in \mathbb{R}
$
for each $j$. This implies
\[
\operatorname{Im}A\left( \cdot\, ;\omega _{0}+i\eta \right) =\sum_{j=0}^{\infty
}A_{j}\left( \cdot \right) \operatorname{Im}\left( i^{j}\right) \eta ^{j},\text{ \ \ 
}0<\eta \ll 1
\]
and hence the normal limit%
\[
\lim_{\eta \downarrow 0}\frac{1}{\eta }\operatorname{Im}A\left( \cdot\, ;\omega
_{0}+i\eta \right) =A_{1}\left( \cdot \right) \text{ converges in }
M_{4}\left( L^{\infty }\left( \mathbb{R}\right) \right) .
\]

Now we know that $\psi \left( \cdot \right) $ is a solution to the Maxwell ODEs (\ref{canonical}) with nonzero real frequency $\omega _{0}$. \ Let $U\left(\cdot \right) $ denote its associated energy density. \ By the properties of the transfer matrix $T\left( z_0,z;\omega \right) $ we know that there exists $%
\psi_0 \in M_{4}\left( \mathbb{C}\right) $ such that $\psi \left( \cdot \right) =T\left( z_0,\cdot\, ;\omega
_{0}\right) \psi_0 $. \ Define $\psi \left( \cdot\, ;\omega _{0}+i\eta \right)
=T\left( z_0,\cdot\, ;\omega _{0}+i\eta \right) \psi_0 $ and its associated energy density by $U\left( \cdot\, ;\omega _{0}+i\eta \right) $. Then by continuity of the energy density as proven in Section \ref{sec:reduction} [see (\ref{ParaDepEnergyDens})] it follows that 
\begin{eqnarray*}
\lim_{\eta \downarrow 0}\psi \left( \cdot\, ;\omega _{0}+i\eta \right)=\psi \left( \cdot \right) \ \text{converges in }\left( L_{loc}^{\infty
}\left( \mathbb{R}\right) \right) ^{4}\text{ and }\lim_{\eta \downarrow 0}U\left( \cdot\, ;\omega _{0}+i\eta \right)=U\left(
\cdot \right) \text{ converges in }L_{loc}^{\infty }\left( \mathbb{R}\right) .
\end{eqnarray*}
Thus by the integral form of the energy conservation law, i.e., Theorem \ref{thm:EnergyConservationLaw}, it follows that%
\[
\frac{c}{16\pi }\left( \psi \left( \cdot\, ;\omega _{0}+i\eta \right) ,\frac{1%
}{\eta }\operatorname{Im}A\left( \cdot;\omega _{0}+i\eta \right) \psi \left( \cdot\,
;\omega _{0}+i\eta \right) \right) =U\left( \cdot\, ;\omega _{0}+i\eta \right) 
\]%
with equality in the sense of in $L_{loc}^{\infty }\left( \mathbb{R}\right)$ and hence by continuity 
\[
\frac{c}{16\pi }\left( \psi \left( \cdot \right) ,A_{1}\left( \cdot \right)
\psi \left( \cdot \right) \right) =\lim_{\eta \downarrow 0}\frac{c}{16\pi }%
\left( \psi \left( \cdot\, ;\omega _{0}+i\eta \right) ,\frac{1}{\eta }\operatorname{Im}
A\left( z;\omega _{0}+i\eta \right) \psi \left( \cdot\, ;\omega _{0}+i\eta
\right) \right) =U\left( \cdot \right) 
\]%
converges in $L_{loc}^{\infty }\left( \mathbb{R}\right) $. Therefore, from these facts it follows that%
\[
\frac{c}{16\pi }\left( \psi \left( z\right) ,\frac{\partial A}{\partial \omega }\left( z;\omega_0 \right)\psi \left( z\right)
\right) =U\left( z\right) \text{ for a.e. }z\in \mathbb{R}.
\]
This completes the proof.
\end{proof}

\subsection{A periodic ambient medium}\label{sec:periodic}  

Let the materials $\epsilon$ and $\mu$ of the ambient space (Fig.~1 outside the slab) be periodic with period $d$, i.e.,
\begin{align}
\epsilon(z+d)=\epsilon(z),\;\;\;\mu(z+d)=\mu(z).
\end{align}
Then for the Maxwell ODEs (\ref{canonical}) the propagator $iJ^{-1}A(z;\kk,\omega)$ for the field $\psi(z)$ along the $z$-axis is $d$-periodic.  According to the Floquet theory (see, {\it e.g.}, \cite[Ch.\!~II]{YakubovichStarzhinsk1975}), the general solution of the Maxwell ODEs is pseudo-periodic, meaning that the transfer matrix $T(0,z)$ is the product of a periodic matrix and an exponential matrix, $T(0,z)=F(z)e^{iKz}$,
\begin{equation*}
  \psi(z) \,=\, F(z) e^{iKz} \psi(0)\,,
  \quad
  F(z+d) = F(z)\,,
  \quad
  F(0) = I\,.
\end{equation*}
For a real pair $(\kk,\omega)$, $F(z)$ can be chosen to be flux-unitary and the constant-in-$z$ matrix $K$ to be flux-self-adjoint:
\begin{equation*}
  [F(z)\psi_1,\psi_2] = [\psi_1,F(z)^{-1}\psi_2]\,,
  \quad
  [K\psi_1,\psi_2] = [\psi_1,K\psi_2]\,.
\end{equation*}
In concise notation, $F(z)^\fadj=F(z)^{-1}$ and $K^\fadj=K$.  The matrix $T(0,d)=e^{iKd}$ is called the monodromy matrix for the sublattice $d\ZZ\subset\RR$. The flux-self-adjointness of $K$ implies that the conjugate of any eigenvalue is also an eigenvalue. Moreover, for a pair $\kwz\in\RR^2\times \RR/\{0\}$ these functions $F,K$ can also be chosen such that they are analytic in a complex neighborhood of that pair which follows from the properties (\ref{SmoothnessOfAinMaxwellDAEs}), (\ref{MatrixAinMaxwellODEsIsHermitian}) of $A$ and \cite{YakubovichStarzhinsk1975}[\S III.4.6]. We shall assume that this is the case in a complex neighborhood of a pair $\kwz$.

In the example of section \ref{sec:example} the ambient medium has $F(z)=I$ and $K=J^{-1}A$, with $A$ a constant matrix.  The condition that the ambient medium simultaneously allows a pair of propagating modes and a pair of evanescent modes is generalized for a periodic ambient medium by the condition that $K$ has two real eigenvalues $k_{-p}, k_{+p}$ and a pair of complex conjugate eigenvalues $k_{-e},k_{+e}$.  We assume that this is the case in a real neighborhood of a pair $(\kk_0,\omega_0)$.
Because $K\kw$ is analytic and $[K\psi_1,\psi_2]=[\psi_1,K\psi_1]$, it suffices to make the following assumption at $\kwz$:

\begin{assumption}\label{Assumption1}
For the real pair $(\kk_0,\omega_0)$, the matrix $\mathring{K}=K(\kk_0,\omega_0)$ is diagonalizable with diagonal form
\begin{equation}
\mathring{\tilde K} \,=\,
  \renewcommand{\arraystretch}{1.1}
\left[
  \begin{array}{cccc}
    \mathring{k}_{-p} & 0 & 0 & 0 \\
    0 & \mathring{k}_{+p} & 0 & 0 \\
    0 & 0 & \mathring{k}_{-e} & 0 \\
    0 & 0 & 0 & \mathring{k}_{+e}
  \end{array}
\right],
\qquad
\renewcommand{\arraystretch}{1.1}
\left\{
  \begin{array}{ll}
    \mathring{k}_{-p}\,, \mathring{k}_{+p}\in\RR\,, & \mathring{k}_{-p}\not=\mathring{k}_{+p}\,, \\
    \vspace{-2ex}  \\
    \mathring{k}_{-e}=\overline{\mathring{k}_{+e}}\,, & \Im\mathring{k}_{+e}>0\,.
  \end{array}
\right.
\end{equation}
In addition, the eigenvectors $\mathring{v}_{\pm p}$ of $\mathring K$ corresponding to $\mathring{k}_{\pm p}$ satisfy
\begin{equation*}
  [\mathring{v}_{+p},\mathring{v}_{+p}] > 0
  \qquad\text{and}\qquad
  [\mathring{v}_{-p},\mathring{v}_{-p}] < 0\,.
\end{equation*}
\end{assumption}

The eigenvalues are extensible to analytic functions $k_{\pm p}\kw$ and $k_{\pm e}\kw$ in a neighborhood of $\kwz$, and the eigenvectors $v_{\pm p}\kw$ and $v_{\pm e}\kw$ can be chosen to be analytic.  The conditions of Assumption~\ref{Assumption1} hold for all $\kw$ near $\kwz$.  The theory of linear algebra in the presence of an indefinite form \cite{GohbergLancasterRodman2005} allows the eigenvectors to be chosen such that, with respect to the basis $\{\vlp,\vrp,\vle,\vre\}$, energy flux and $K$ have the forms
\begin{equation}\label{canonicalform}
  \left([v_i,v_j]\right)_{i,j\in \{-p,+p,-e,+e\}} \,=\,
  \renewcommand{\arraystretch}{1.1}
\left[
  \begin{array}{cccc}
    -1 & 0 & 0 & 0 \\
    0 & 1 & 0 & 0 \\
    0 & 0 & 0 & 1 \\
    0 & 0 & 1 & 0
  \end{array}
\right]\, \; \big(\text{for $(\kk,\omega)$ real}\big),
\qquad
  \tilde K \,=\,
  \renewcommand{\arraystretch}{1.1}
\left[
  \begin{array}{cccc}
    k_{-p} & 0 & 0 & 0 \\
    0 & k_{+p} & 0 & 0 \\
    0 & 0 & k_{-e} & 0 \\
    0 & 0 & 0 & k_{+p}
  \end{array}
\right]\,.
\end{equation}
The diagonal form of $K$ holds in a complex neighborhood of $(\kk_0,\omega_0)$, whereas the form for energy flux holds only in a real neighborhood.

The general solution to the Maxwell ODE system is
\begin{equation*}
  \psi(z) \,=\, F(z)\left( a\,\vlp e^{ik_{-p}z} + b\,\vrp e^{ik_{+p}z} + c\,\vle e^{ik_{-e} z} + d\,\vre e^{ik_{+p} z} \right),
  \quad
  a,b,c,d\in\CC.
\end{equation*}
Equation (\ref{canonicalform}) indicates the energy-flux interactions among the four modes, which is independent of $z$:
\begin{equation}\label{conservation}
  [\psi(z),\psi(z)] \,=\, -|a|^2 + |b|^2 + c\bar d + \bar c d\qquad
  \text{for $(\kk,\omega)$ real}.
\end{equation}
The modes are designated as rightward or leftward as follows:
\begin{equation}\label{modes}
\renewcommand{\arraystretch}{1}
\left.
  \begin{array}{rl}
    \text{leftward propagating:} & \wlp(z) = F(z)\,\vlp e^{ik_{-p} z} \\
    \text{rightward propagating:} & \wrp(z) = F(z)\,\vrp e^{ik_{+p} z} \\
    \text{leftward evanescent:} & \wle(z) = F(z)\,\vle e^{ik_{-e} z} \\
    \text{rightward evanescent:} & \wre(z) = F(z)\,\vre e^{ik_{+e} z}
  \end{array}
\right.
\end{equation}
The modes $w_{\pm e}$ remain evanescent near $(\kk_0,\omega_0)$ because the eigenvalues $k_{\pm e}$ have nonzero imaginary parts.  The modes $w_{\pm p}$ attain exponential growth or decay, depending on the sign of $\Im\omega$.

\smallskip
Theorem~\ref{thm:EnergyIndRep}, proved later on, relates the energy density $U(z)$, as given by Theorem~\ref{thm:EnergyDensity} to the corresponding one for which the $z$-dependent propagator matrix $A$ is replaced by the ``effective" $z$-independent propagator matrix $K$.
Both give the same total energy in a full period,
\begin{equation}\label{AK}
\frac{c}{16\pi}\int_{0}^{d}\left(\frac{\partial A}{\partial \omega}(z;\kk_0,\omega_0)\psi_1(z),\psi_2(z)\right)dz=\int_{0}^{d}\left[\frac{dK}{d\omega}(\omega_0) e^{iK(\omega_0) z}\psi_1(0),e^{iK(\omega_0)z}\psi_2(0)\right]dz,
\end{equation}
This fact is used to show how the imaginary parts of the exponents $k_{\pm p}$ change as $\omega$ attains a small imaginary part.

\begin{Theorem}[Analytic continuation of modes]\label{thm:exponents}
  For $(\kk,\omega)$ near $(\kk_0,\omega_0)$ with $\kk$ real, both $\Im k_{+p}$ and $-\Im k_{-p}$ have the same sign as $\Im\omega$.
\end{Theorem}

\begin{proof}
The proof is essentially based on the well-known principle that for propagating electromagnetic Bloch waves in a passive and lossless periodic medium, the group velocity equals the energy velocity which points in the same direction as the energy flux since energy density must be positive.

Begin by differentiating the relation $(K(\omega)-k_{\pm p}(\omega))v_{\pm p}(\omega)=0$ with respect to $\omega$ and applying $[v_{\pm p},\cdot]$ to the result yields
\begin{equation}\label{kK}
  \left[ v_{\pm p},\left(\frac{\partial K}{\partial\omega}-\frac{\partial k_{\pm p}}{\omega}\right) v_{\pm p} \right]
  + \left[ v_{\pm p}, (K-k_{\pm p}) \frac{\partial v_{\pm p}}{\partial\omega} \right] \,=\, 0.
\end{equation}
By the flux-self-adjointness of $K$ and the reality of $k_{\pm p}$ for real $\kk$ and $\omega$, the second term vanishes, leaving
\begin{equation*}
  \frac{\partial k_{\pm p}}{\partial\omega}[v_{\pm p},v_{\pm p}] = \left[v_{\pm p}, \frac{\partial K}{\partial\omega}v_{\pm p} \right].
\end{equation*}
In (\ref{AK}), put $\psi_1(0)=\psi_2(0)=v_{\pm p}$.  Then $\psi_1(z)=\psi_2(z)$, and the left-hand side is positive because the integrand, by Theorem~\ref{thm:EnergyDensity}, is equal to the positive energy density $U(z)$ given by~(\ref{FluxSU}).  Since $e^{iKz}\psi_1(0)=e^{ik_{\pm p}z}v_{\pm p}$, the right-hand side of (\ref{AK}) is
\begin{equation*}
  \int_0^{d} \left[ \frac{dK}{d\omega}v_{\pm p},\,v_{\pm p} \right]dz
  = d \left[ \frac{dK}{d\omega}v_{\pm p},\,v_{\pm p} \right]\,,
\end{equation*}
and the right-hand side of (\ref{kK}) is therefore positive.
Thus $\partial k_{\pm p}/\partial\omega$ and $[v_{\pm p},v_{\pm p}]$ have the same sign and the theorem follows.
\end{proof}

Denote by $\Plp$, $\Prp$, $\Ple$, and $\Pre$ the rank-1 complementary projections (their sum is the identity) onto the corresponding eigenspaces of $K$ and by $\Pl$ and $\Pr$ the rank-2 complementary projections onto the {\em leftward and rightward spaces} of $\CC^4$:
\begin{equation}\label{projections}
  \renewcommand{\arraystretch}{1}
\left.
  \begin{array}{ll}
  \Pl = \Plp + \Ple\,, & \Range(\Pl) = \sspan\{\vlp,\vle\} \\
  \vspace{-2ex}\\
  \Pr = \Prp + \Pre\,, & \Range(\Pr) = \sspan\{\vrp,\vre\}    
  \end{array}.
\right.
\end{equation}

The form of the flux interaction matrix in (\ref{canonicalform}) plays an important role in the way fields are scattered by an obstacle, particularly at resonance.  The critical fact is that each oscillatory mode carries energy in isolation while the evanescent modes induce energy flux only when superimposed with one another.  This idea is manifest in the flux-adjoints of the projection operators: 
Both projections $\Plp$ and $\Prp$ onto the ``propagating subspaces" are flux-self-adjoint, and the projections $\Ple$ and $\Pre$ onto the ``evanescent subspaces" are flux-adjoints of each other:
\begin{equation}\label{adjoints}
  [\Plp\psi_1,\psi_2] = [\psi_1,\Plp\psi_2]\,,\;\;\;
  [\Prp\psi_1,\psi_2] = [\psi_1,\Prp\psi_2]\,,\;\;\;
  [\Ple\psi_1,\psi_2] = [\psi_1,\Pre\psi_2]\,.
\end{equation}
In concise notation, $P_{\pm p}^\fadj=P_{\pm p}$ and $\Ple^\fadj=\Pre$.

\section{Resonant scattering by a defect layer}\label{sec:scattering} 

This section, especially~\ref{sec:anomalies}, contains the heart of this work, namely a rigorous analysis of guided modes and the anomalous scattering behavior exemplified by the system in section~\ref{sec:example}, specifically, resonant transmission anomalies and field amplification.

We investigate scattering of a time-harmonic electromagnetic field by a defective layer, or slab, embedded in a layered ambient medium (Fig.~\ref{fig:slab}).  The materials $\epsilon,\mu$ are lossless and passive---they satisfy (\ref{LosslessLayeredMedia}) and (\ref{PassiveMedia}).  We work near a real wavevector-frequency pair $\kwz$ at which the ambient medium admits rightward and leftward propagating and evanescent modes, according to Assumption \ref{Assumption1}.  $P_\rt$ and $P_\lf$ are the projections (\ref{projections}) onto the two-dimensional rightward and leftward mode spaces, as described on page~\pageref{modes}, and $T=T(0,L)$ denotes the transfer matrix across the slab.

The analysis of transmission anomalies is based on a few key points, which are developed rigorously in the subsections below.
\begin{enumerate}
  \item  An analytic eigenvalue $\ell\kw$ (of the matrix $T\Pl-\Pr$ introduced below), whose zero set defines the dispersion relation for generalized guided modes, is {\em algebraically simple} at a real pair $\kwz$ of a guided slab mode.
  This was hypothesized in previous works~\cite{ShipmanVenakides2005,Shipman2010} on transmission anomalies and is established for layered media in Theorem~\ref{thm:scattering}.
  \item  The far-field scattering matrix is the ratio of analytic functions of $\kw$
  whose {\em zero sets intersect at a real pair} $\kwz$ of a guided slab mode,
\begin{equation}\label{S0}
  S_0\kw
  \,=\, \mat{1.3}{t_{+p}\kw}{r_{+p}\kw}{r_{-p}\kw}{t_{-p}\kw}
  \,=\; \frac{1}{\ell\kw} \mat{1.3}{b_1\kw}{a_2\kw}{a_1\kw}{b_2\kw}\,.
\end{equation}
  \item  For $\kk\in\RR^2$, the poles in the $\omega$-plane of the full scattering matrix are in the lower half-plane.
  \item  The eigenvalue $\ell\kw$ is nondegenerate in the sense that
\begin{equation*}
  \frac{\partial\ell}{\partial\omega}\kwz \,\not=\, 0\,.  
\end{equation*}
The proof is nontrivial and requires a careful treatment of electromagnetic energy density and flux in layered media.  This generic condition was hypothesized in~\cite{ShipmanVenakides2005,Shipman2010} and is now established for layered media in Theorem~\ref{thm:nondegeneracy}.
\end{enumerate}

\subsection{The scattering problem and guided modes}

Let the slab extend from $z=0$ to $z=L$ as in Fig.~\ref{fig:layered}.
It is convenient to define the point $z=L$ so that the electromagnetic coefficients $\epsilon(z)$ and $\mu(z)$ in the period $[-d,0]$ are identical to those in the period $[L,L+d]$ ({\itshape i.e.}, $\epsilon(z+L+d)=\epsilon(z)$ and $\mu(z+L+d)=\mu(z)$ for all $z\in[-d,0]$). Thus, in the ambient space, any solution $\psi$ of the Maxwell ODEs (\ref{canonical}) has the form
\begin{equation}\label{ambientfield}
\renewcommand{\arraystretch}{1.0}
\left.
  \begin{array}{rl}
  z<0: & \psi(z) \;=\; F(z)e^{iKz}\psi(0) \,, \\
  \vspace{-1.5ex}\\
  z>L: & \psi(z) \;=\; F(z-L)e^{iK(z-L)}\psi(L)\,,    
  \end{array}
\right.
\end{equation}
with $\psi(0)$ and $\psi(L)$ related through the flux-unitary transfer matrix across the slab $T=T(0,L)$,
\begin{equation}\label{SlabTransferMatrix}
  T\psi(0)=\psi(L),
  \qquad
  T := T(0,L),
  \qquad
  [T\psi_1,\psi_2] = [\psi_1,T^{-1}\psi_2].
\end{equation}

The scattering problem for this stratified medium, as illustrated in Fig.~\ref{fig:scattering}, can be described as follows.  A~harmonic source located to the left of some point $z_-<0$, emits a field that, in the region $z_-<z\leq0$ is of the rightward form $\psi_+^\inc(z) = j_{+p}\wrp(z) + j_{+e}\wre(z)$ (see~(\ref{modes})), and a source to the right of a point $z_+>L$ emits a field that, in the region $L\leq z<z_+$ is of the leftward form $\psi_-^\inc(z) = j_{-p}\wlp(z-L) + j_{-e}\wle(z-L)$.  In the specified domains, these fields are designated as {\em incoming}, or incident, as they are directed from their sources toward the slab.  In the absence of the slab, with $L=0$, they would continue indefinitely to $\pm\infty$ where they would become outgoing fields.  But with the slab present, the {\em outgoing field} is modified.  To the left of the slab ($z_-<z\leq0$), it is of the leftward form $\psi_-^\out(z)=u_{-p}\wlp(z)+u_{-e}\wle(z)$, and to the right ($L\leq z<z_+$) it is rightward, $\psi_+^\out(z)=u_{+p}\wrp(z-L)+u_{+e}\wre(z-L)$. 

A solution $\psi(z)$ to the scattering problem is a solution of the Maxwell ODEs (\ref{canonical}) that is decomposed outside the slab into incoming and outgoing parts.  This decomposition expands the solution~(\ref{ambientfield}) into physically meaningful modes,
\begin{equation}\label{totalfield}
\renewcommand{\arraystretch}{1.0}
\left.
  \begin{array}{rrl}
  z<0: & \psi(z) &=\; \psi_-^\out(z) + \psi_+^\inc(z) \\
           && =\; u_{-p}\wlp(z)+u_{-e}\wle(z) + j_{+p}\wrp(z) + j_{+e}\wre(z) \\
           && =\; F(z)\big( u_{-p}\vlp e^{ik_{-p}z} +u_{-e}\vle e^{i k_{-e} z} + j_{+p}\vrp e^{ik_{+p}z} + j_{+e}\vre e^{ik_{+e} z} \big)\,, \\
  \vspace{-1ex}\\
  z>L: & \psi(z) &=\; \psi_-^\inc(z) + \psi_+^\out(z) \\
            && =\; j_{-p}\wlp(z-L) + j_{-e}\wle(z-L) + u_{+p}\wrp(z-L)+u_{+e}\wre(z-L) \\   
            && =\; F(z-L) \big( j_{-p}\vlp e^{ik_{-p}(z-L)} + j_{-e}\vle e^{ik_{-e}(z-L)} + u_{+p}\vrp e^{ik_{+p}(z-L)}+u_{+e}\vre e^{ik_{+e}(z-L)} \big)\,.    
  \end{array}
\right.
\end{equation}

\begin{figure}
\centerline{\scalebox{0.4}{\includegraphics{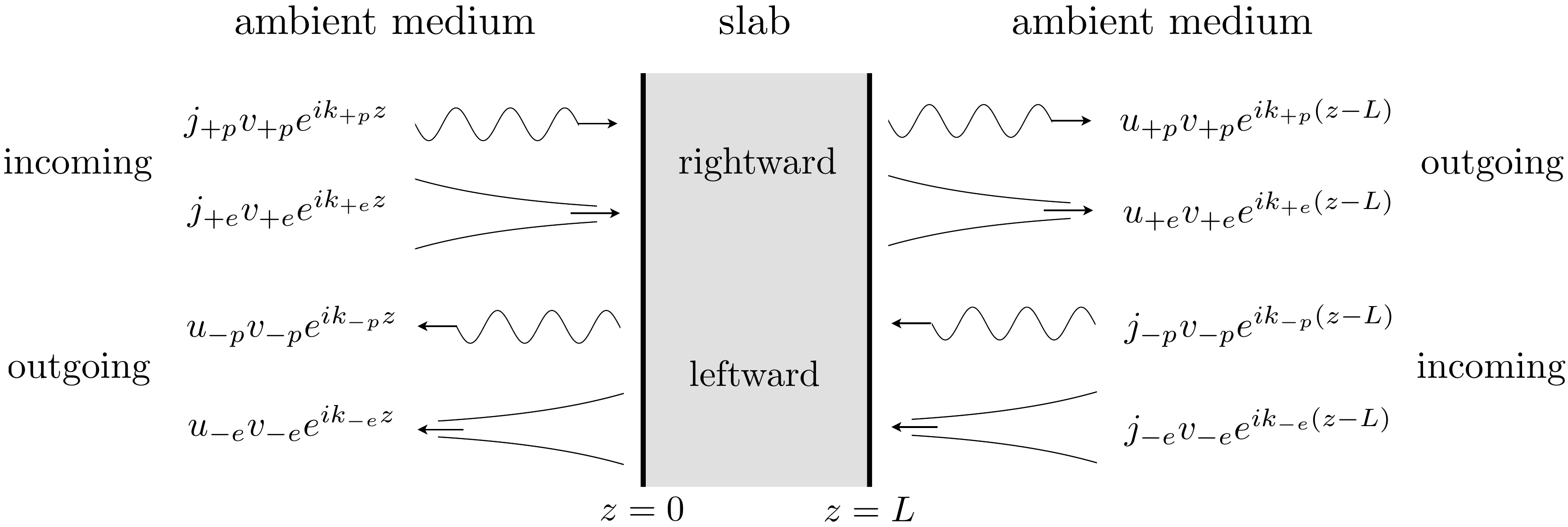}}}
\caption{\small  Scattering of harmonic source fields in a periodic ambient medium by a defect layer (slab), at frequency $\omega$ and wavevector $\kk$ parallel to the slab.  Sources on each side of the slab emit fields, which arrive at the slab as {\em incoming} propagating and evanescent modes.  Scattering by the slab results in an {\em outgoing} field, whose modes are directed in the opposite direction to the incoming modes.  The periodic factor $F(z)$ (left) and $F(z-L)$ (right) are omitted.}
\label{fig:scattering}
\end{figure}

The field $\psi(z)$ is determined by $\psi(0)$ or $\psi(L)$ alone, and these boundary values are related through the slab transfer matrix $T$ by $T\psi(0)=\psi(L)$, which has the block form
\begin{equation}\label{transfer2}
  \mat{1.3}{T_{--}}{T_{-+}}{T_{+-}}{T_{++}}
  \col{1.3}{\psi_-^\out(0)}{\psi_+^\inc(0)}
  \,=\,
  \col{1.3}{\psi_-^\inc(L)}{\psi_+^\out(L)},
\end{equation}
\begin{equation*}
  T_{--} = P_- T P_-\,,\quad
  T_{-+} = P_- T P_+\,,\quad
  T_{+-} = P_+ T P_-\,,\quad
  T_{++} = P_+ T P_+\,.
\end{equation*}
This equation can be rearranged to separate the  incoming (given) outgoing (unknown) fields,
\begin{equation}\label{scattering1}
  \mat{1.3}{-I}{T_{-+}}{0}{T_{++}}
  \col{1.3}{\psi_-^\inc(L)}{\psi_+^\inc(0)}
  +
  \mat{1.3}{T_{--}}{0}{T_{+-}}{-I}
  \col{1.3}{\psi_-^\out(0)}{\psi_+^\out(L)}
  \,=\, 0\,.
\end{equation}
By defining incoming and outgoing vectors in $\CC^4$ by
\begin{equation*}
  \Psi^\inc = \psi_-^\inc(L) + \psi_+^\inc(0)\,,
  \qquad
  \Psi^\out = \psi_-^\out(0) + \psi_+^\out(L)\,,
\end{equation*}
the form (\ref{scattering1}) of the scattering problem may be written concisely as
\begin{equation}\label{scattering2}
  (T\Pl - \Pr) \Psi^\out + (T\Pr - \Pl)\Psi^\inc \,=\, 0\,.
  \qquad
  \text{(scattering problem)}
\end{equation}
This formulation reduces the scattering problem to finite-dimensional linear algebra involving the analytic slab transfer matrix $T=T\kw$ and the eigenspaces of $K=K\kw$, where $e^{iKd}$ is the analytic monodromy matrix for the periodic ambient medium.

The scattering problem is uniquely solvable for $\Psi^\out$ whenever $T\Pl-\Pr$ is invertible.  When it is not invertible, there is a nonzero vector $\Psi^\out$ that solves (\ref{scattering2}) with $\Psi^\inc=0$.
In this case, the projections of the vector $\Psi^\out$ onto the leftward and rightward subspaces are the traces of a solution $\psig(z)$ of the Maxwell ODEs that is outgoing as $|z|\to\infty$.  Because it has no incoming component on either side, conservation of energy for {\em real} $\kw$ (\ref{conservation}) requires that it be exponentially decaying as $|z|\to\infty$.  Such a field is a {\em guided mode} of the slab, and its vector $\Psi^\out$ of traces at $z=0$ and $z=L$ satisfies
\begin{equation}\label{guidedmode1}
  (T\Pl - \Pr) \Psi^\guided  \,=\, 0\,.
  \qquad
  \text{(guided-mode equation)}
\end{equation}

Observe that (\ref{scattering2}) and (\ref{guidedmode1}) for scattering fields and guided modes remain unchanged if the factor $F(z)$ is removed from the fields (\ref{totalfield}) in the ambient medium since $F(0)=I$.  This amounts to replacing the transfer matrix $F(z)e^{iKz}$ for the ambient medium with the simple exponential $e^{iKz}$, or, equivalently, replacing the $z$-periodic matrix $A(z)$ with the constant matrix $JK$ in the Maxwell ODEs to obtain
\begin{equation*}
  \frac{d}{dz}\psi = iK\psi.
\end{equation*}
So we might as well be working with a homogeneous ``effective" ambient medium in place of the periodic one.  Solutions need only be multiplied by $F(z)$ for $z<0$ and $F(z-L)$ for $z>L$ to obtain solutions of the original problem.

When $(\kk,\omega)$ is not real, $\psi^\guided(z)$ is known as a {\em generalized guided mode} or a {\em guided resonance} \cite{FanJoannopoul2002}.  For real wavevector $\kk$, one necessarily has $\Im\omega\leq0$, and, if $\Im\omega<0$, the mode grows exponentially as $z\to\infty$ or $z\to-\infty$.  That $\Im\omega\leq0$ is the statement that the resolvent of the scattering matrix (developed below) has its poles in the lower half plane.

\begin{Theorem}[Poles of the resolvent and generalized guided modes]\label{thm:resolvent}
  If $(\kk,\omega)$ with $\kk\in\RR^2$ is near a point $(\kk_0,\omega_0)$ at which Assumption~\ref{Assumption1} is satisfied and if $(T\Pl-\Pr)|_{\kw}\Psi^\guided=0$, then either $\Im\omega=0$ and $\psi^\guided(z)$ is exponentially decaying as $|z|\to\infty$ or $\Im\omega<0$ and $\psi^\guided(z)$ is exponentially growing as $z\to\infty$ or $z\to-\infty$.
\end{Theorem}

\begin{proof}
  The proof is based on Theorem \ref{thm:EnergyConservationLaw} and Theorem \ref{thm:exponents}.  By the latter, $\psi^\guided$ is exponentially decaying in $|z|$ if $\Im\omega>0$, and this contradicts Theorem \ref{thm:EnergyConservationLaw}.  So $\Im\omega\leq0$.  If $\Im\omega=0$, we have seen that the propagating modes vanish by conservation of energy (\ref{conservation}), so that $\psi^\guided$ decays with $|z|$.  If $\Im\omega<0$, Theorem \ref{thm:EnergyConservationLaw} guarantees that $\psi^\guided$ does not decay as $|z|\to\infty$, and thus the coefficients of $v_{\pm p}$ do not both vanish.  But for $\Im\omega<0$, the mode $v_{\pm p}$ grows exponentially as $z\to\pm\infty$ by Theorem \ref{thm:exponents}.  
\end{proof}

The time-averaged energy of a guided mode is positive, as stated in the following theorem.  This will be instrumental in the proof of the nondegeneracy of the dispersion relation for slab modes stated in Theorem~\ref{thm:nondegeneracy}.

\begin{Theorem}[Energy of a guided mode]\label{thm:energy}
If $\psi^\guided(z)$ is a guided mode solution to the Maxwell ODEs at a real pair $\kwz$ satisfying Assumption~\ref{Assumption1}, then
\begin{equation}
\frac{c}{16\pi}\int_{-\infty}^{\infty}\left(\psi^\guided(z),\frac{\partial A}{\partial \omega}\kw\psi^\guided(z)\right)\,dz = \int_{-\infty}^\infty U^g(z)\,dz>0,
\end{equation}
where $U^g(z)$ is the time-averaged energy density of the time-harmonic electromagnetic field with tangential wavevector $\kk$, frequency $\omega$, and tangential electric and magnetic field components $\psi^g(z)$.
\end{Theorem}
\begin{proof}
By Theorem \ref{thm:resolvent} we know that $\psi^\guided(z)$ is exponentially decaying as $|z|\rightarrow \infty$ and so by Theorem \ref{thm:EnergyConservationLaw} the equality follows by taking $z_0\rightarrow -\infty$, $z_1\rightarrow \infty$. The fact that the integral is positive follows from passivity hypothesis (\ref{PassiveMedia}) for the layered media.
\end{proof}

Perfectly guided modes, exponentially decaying away from the defect layer, occur by Theorem~\ref{thm:resolvent} at real pairs $\kw$.  It turns out that they are always simple as the next theorem states; the proof is deferred to section~\ref{sec:nondegeneracy}.  Part (2) of the theorem identifies the three-dimensional space of incoming fields for which the scattering problem has a (nonunique) solution at any pair $\kwz$ at which a perfect guided mode exists.

\begin{Theorem}[Simplicity of guided modes]\label{thm:scattering}  
If $T\Pl\!-\!\Pr$ is noninvertible at the real pair $(\kk,\omega)$, then the following statements hold.
\begin{enumerate}
  \item  The zero eigenvalue of $T\Pl\!-\!\Pr$ is of algebraic (and geometric) multiplicity $1$.  A nonzero solution to (\ref{guidedmode1}) has only evanescent components, that is,
\begin{equation*}
  \Pl\Psig = u_{-e}^\guided\vle\,,
  \qquad
  \Pr\Psig = u_{+e}^\guided\vre\,,
\end{equation*}
($u_{-e}^\guided$ and $u_{+e}^\guided$ are complex constants) and $\Psig$ corresponds to a guided-mode solution $\psig(z)$ to the Maxwell ODEs,
\begin{equation}\label{guidedmode2}
\renewcommand{\arraystretch}{1.0}
\left.
  \begin{array}{rll}
  z<0: & \psig(z) = u_{-e}^\guided\wle(z) &= u_{-e}^\guided F(z)\vle e^{ik_{-e} z} \,, \\
  \vspace{-1.5ex}\\
  z>L: & \psig(z) = u_{+e}^\guided\wre(z-L) &= u_{+e}^\guided F(z-L)\vre e^{ik_{+e}(z-L)} \,,  
  \end{array}
\right.
\end{equation}
that decays exponentially as $|z|\to\infty$.
  \item  The scattering problem (\ref{scattering2}) admits a solution $\Psi^\out$ if and only if the coefficients $j_{+e}$ and $j_{-e}$ of the evanescent modes of the incident field at $z=L$ and $z=0$,
\begin{eqnarray*}
  && \Pl\Psi^\inc = j_{-p}\vlp + j_{-e}\vle\,, \\
  && \Pr\Psi^\inc = j_{+p}\vrp + j_{+e}\vre\,,
\end{eqnarray*}
are related to the coefficients $u_{+e}^\guided$ and $u_{-e}^\guided$ of the guided mode by
\begin{equation*}
  \det\mat{1.2}{\overline{j_{+e}}}{\overline{j_{-e}}}{u_{+e}^\guided}{u_{-e}^\guided}\,=\,0\,,
\end{equation*}
that is, if and only if
\begin{equation*}
  \Psi^\inc \in {\cal N} := \sspan\{ \vrp,\vlp,\Psi^\guided \}\,.
\end{equation*}
The general solution is of the form
\begin{equation*}
  \Psi^\out = \Psi^\out_\text{\tiny partic} + c \Psig\,,
\end{equation*}
in which $c$ is an arbitrary complex number.
\end{enumerate}
\end{Theorem}

Part (2) says that, in order for the system to have a time-harmonic response to an incident field at real $\kw$ for which the slab supports a guided mode, the evanescent mode components of the incident field at $z=0$ and $z=L$ are equal to the complex conjugates of the  values of a guided mode at $z=L$ and $z=0$, respectively.

\commentsps{Question: Must $u_{+e}^\guided$ and $u_{-e}^\guided$ have the same magnitude?}

\subsection{The full scattering matrix around a guided mode}  

The delicate behavior of resonance phenomena are revealed through an analysis of the behavior of the {\em scattering matrix} near a real pole $\kwz$, where the slab admits a perfect guided mode.

Let us set $B:=T\Pl-\Pr$ and $C:=T\Pr-\Pl$, so that the scattering and guided-mode equations become
\begin{equation*}
\renewcommand{\arraystretch}{1.1}
\left.
  \begin{array}{ll}
  B(\kk,\omega)\Psi^\out + C(\kk,\omega)\Psi^\inc = 0\,, &  \text{(scattering)} \\
  B(\kk,\omega)\Psi^\guided = 0 \,. & \text{(generalized guided mode)}
  \end{array}
\right.
\end{equation*}
$B$ and $C$ are analytic in $\kw$ near $\kwz$.
In terms of the scattering matrix $S:=-B^{-1}C$, the scattering problem has the simple form
\begin{equation*}
  \Psi^\out = S\kw \Psi^\inc\,.
  \qquad
  \text{(scattering matrix formulation)}
\end{equation*}
The complex dispersion relation for guided slab modes gives the locus of $\kw$ pairs for which a generalized guided mode exists,
\begin{equation*}
  D\kw := \det B\kw=0\,. \quad \text{(dispersion relation for guided modes)}
\end{equation*}
One can think of this relation as defining a pole of the scattering matrix, located in the closed lower $\omega$-halfplane by Theorem~\ref{thm:resolvent}, which varies with the wavevector parameter $\kk$.

Suppose now that $D\kwz=0$ for a {\em real} pair $\kwz$, so that $B\kwz\Psi^\guided=0$, as in the example of section~\ref{sec:example}.  Theorem \ref{thm:scattering} says that {\em the zero-eigenvalue of $B\kwz$ is algebraically simple}; thus it extends to a simple analytic eigenvalue $\ell\kw$ in a complex neighborhood of $\kwz$.  In this neighborhood, the locus of $\kw$ pairs that admit a generalized guided mode is given by the complex dispersion relation
\begin{equation*}
  \ell\kw=0\,.
  \qquad \text{(guided-mode relation near $\kwz\in\RR^3$)}
\end{equation*}
It follows from Theorem \ref{thm:nondegeneracy} is that $\ell$ cannot be identically zero, i.e.,
 $  \ell\kw\not\equiv 0 $.
Thus, the scattering matrix $S\kw$ is meromorphic in $\kw$ near $\kwz$.

\smallskip

Resonance phenomena near a guided mode pair $\kwz$, is tied to the behavior of the poles of the scattering matrix $S$ (compare the Auger states in quantum mechanics~\cite{ReedSimon1980d}).  We follow \cite{ShipmanVenakides2005,Shipman2010} and
denote by $\Pres$ the analytic Riesz projection onto the one-dimensional eigenspace of $B\kw$ corresponding to the eigenvalue $\ell\kw$ and by $\Preg$ the complementary projection,
\begin{equation*}
  \Pres\kw = \frac{1}{2\pi i}\oint_{\cal C} (\lambda - B\kw)^{-1} d\lambda\,,
  \qquad
  \Pres + \Preg = I,
\end{equation*}
where $\cal C$ is a circle in the complex $\lambda$-plane enclosing $\ell\kw$ and no other eigenvalues of $B\kw$.  These projections commute with $B$ and reduce it analytically:
\begin{equation*}
  B\kw = \ell\kw \Pres\kw + \tilde B\kw \Preg\kw,
\end{equation*}
in which
\begin{equation*}
  \tilde B\kw = \Pres\kw + B\kw  
\end{equation*}
is analytic {\em and invertible} near $\kwz$.
The inverse of $\tilde B$,
\begin{equation}
  R\kw=\tilde B\kw^{-1}
\end{equation}
is analytic at $\kwz$ and $\Preg R\Preg = R\Preg=\Preg R$, and from this one obtains a meromorphic representation of~$B^{-1}$,
\begin{equation*}
  B\kw^{-1} = \ell\kw^{-1} \Pres\kw + R\kw\Preg\kw\,.
\end{equation*}
In block-matrix form with respect to the images of $\Pres$ and $\Preg$, $B^{-1}$ has the form
\begin{equation*}
  B\kw^{-1} =
\renewcommand{\arraystretch}{1.3}
\left[
  \begin{array}{cc}
    \ell\kw^{-1} & 0 \\
    0 & R\kw
  \end{array}
\right]\,.
\end{equation*}
Thus the scattering matrix $S$ has the form
\begin{equation}\label{Smatrix0}
  S  \,=\,  \ell^{-1} \tilde S\,,
  \qquad
  \tilde S\kw = -(\Pres + \ell R\Preg)C\,,
\end{equation}
with $\tilde S$ analytic in $\kw$ at $\kwz$.
%
%

\begin{figure}
\centerline{\scalebox{0.45}{\includegraphics{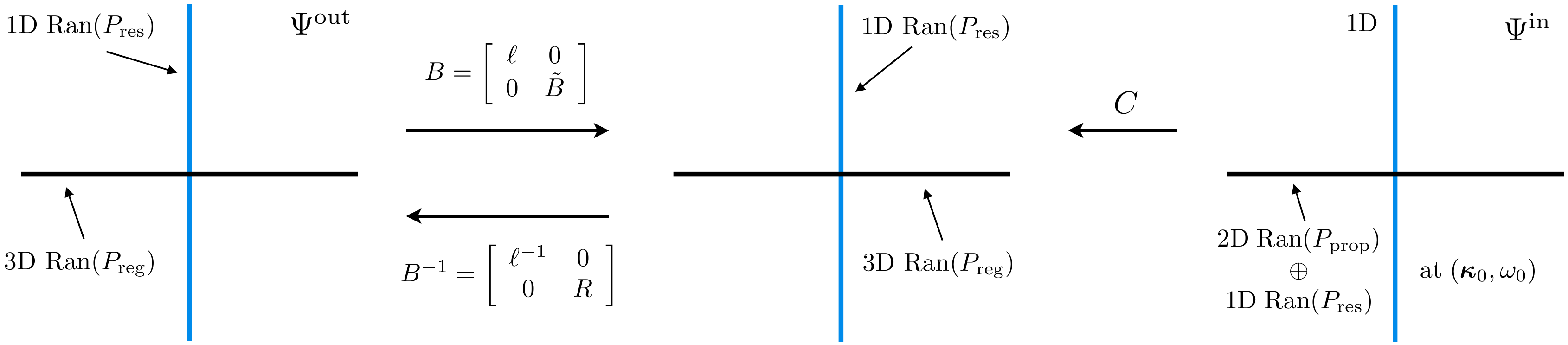}}}
\caption{\small The matrix $B$ in the scattering problem $B\kw\Psi^\out + C\kw\Psi^\inc=0$ is noninvertible at a guided mode wavevector-frequency $\kwz$, where an analytic eigenvalue $\ell\kw$ of $B\kw$ vanishes.  Near $\kwz$, $B\kw$ is decomposed analytically by ``resonant" and ``regular" projections $\Pres$ and $\Preg$.  At {\em real} $\kwz$, the preimage of $\Range(\Preg)$ under $C$ is equal to the three-dimensional space $\cal N$ spanned by $\Range(\Pres)$ and the propagating harmonics (Theorem~\ref{thm:scattering}).  This is the space of incident fields $\Psi^\inc$ for which the scattering problem admits a solution at $\kwz$.}
\label{fig:BCoperators}
\end{figure}

\smallskip
{\slshape\bfseries Interpretation.}\
Let $\Phi\kw$ be analytic, and set $\Psi^\inc=\ell\Phi$, so that the source field vanishes on the guided-mode relation.  The outgoing field trace
\begin{equation}\label{analyticconnection}
  \Psi^\out \,=\, -\Pres(C\,\Phi) - \ell R\Preg(C\,\Phi)
  \qquad
  (\Psi^\inc=\ell\Phi)\,
\end{equation}
provides an {\em analytic connection} between scattering states (on $\ell\kw\not=0$) and generalized guided modes (on $\ell\kw=0$).  This was the crux of the analysis of resonant transmission anomalies in~\cite{ShipmanVenakides2005}.

\smallskip
The importance of the nondegeneracy of guided modes in the sense of the following theorem was mentioned in point (4) at the beginning of section~\ref{sec:scattering}.

\begin{Theorem}[Nondegeneracy of guided modes]\label{thm:nondegeneracy}
If $B\kw$ is noninvertible at a real pair $\kwz$ then its zero eigenvalue can be extended analytically about $\kwz$ to an algebraically simple eigenvalue $\ell\kw$ that satisfies
\begin{equation*}
  \frac{d\ell}{d\omega}\kwz = \frac{1}{\overline{u_{+e}^\guided} u_{-e}^\guided} \frac{ic}{16\pi} \int_{-\infty}^\infty \left( \frac{\partial A}{\partial \omega}\kwz\psig(z),\psig(z) \right)dz
  = \frac{1}{\overline{u_{+e}^\guided} u_{-e}^\guided}i\int_{-\infty}^\infty U^g(z)\,dz \not= 0\,.
\end{equation*}
Here $\psig(z)$ is any guided mode associated with the nullspace of $B\kwz$, the coefficients $u_{-e}^\guided$, $u_{+e}^\guided$ are related to this guided mode by $\psig(0)=u_{-e}^\guided\vle$, $\psig(L)=u_{+e}^\guided\vre$, and $U^g(z)$ is the time-averaged energy density of the time-harmonic electromagnetic field $(\EE,\HH)$ with tangential wavevector $\kk_0$, frequency $\omega_0$, and tangential electric and magnetic field components $[E_1,E_2,H_1,H_2]^T=\psi^g(z)$.
\end{Theorem}
\begin{proof}
The proof is technical and is deferred to section~\ref{sec:nondegeneracy}. 
\end{proof}

A local description of the scattering matrix near a real point on the guided mode dispersion relation is given in the following theorem.

\begin{Theorem}[Scattering matrix: local representation]\label{thm:SmatrixLocal}
Let $\kwz$ be a real pair such that $\ell\kwz=0$. Then the generic condition
\begin{equation}
  \frac{\partial\ell}{\partial\omega}\kwz\,\not=\,0.
\end{equation}
is satisfied and $\ell$ has the local representation
\begin{equation}
  \ell\kw = (\omega-\omega_0+g(\kk))\,\hat\ell\kw\,,
\end{equation}
where $\hat\ell\kw$, $g(\kk)$ are both analytic with $\hat\ell\kwz\not=0$ and $g(\kk_0)=0$. Furthermore, the scattering matrix has the local representation
\begin{equation}\label{Smatrix}
  S\kw = \ell^{-1}\kw\tilde S\kw =
  \frac{h\kw}{\omega-\omega_0+g(\kk)}\big( \Pres + \ell R\Preg \big)C\,,
\end{equation}
with $h=-\hat\ell^{-1}$ analytic and nonzero at $\kwz$. Moreover, at $\kk=\kk_0$ the scattering matrix $S(\kk_0,\omega)$ has a simple pole at $\omega=\omega_0$ with
\begin{equation}\label{SmatrixOmegaPole}
\lim_{\omega\rightarrow\omega_0}(\omega-\omega_0)S(\kk_0,\omega)=\left(\frac{\partial\ell}{\partial\omega}\Pres C\right)\kwz\not=0.
\end{equation}
\end{Theorem}

\begin{proof}
Theorem~\ref{thm:nondegeneracy} states that the generic condition $\frac{\partial\ell}{\partial\omega}\kwz\,\not=\,0$ always holds at real $\kwz$ at which $\ell$ vanishes.  The Weierstra{\ss} Preparation Theorem then provides the local representation $\ell\kw=(\omega-\omega_0+g(\kk))\,\hat\ell\kw$ with $\hat\ell\kwz\not=0$ and $g(\kk_0)=0$ and both analytic.  By these facts and the representation (\ref{Smatrix0}) above, the local representation (\ref{Smatrix}) for the scattering matrix now follows.  In particular $S(\kk_0,\omega)$ is meromorphic in $\omega$ at $\omega_0$ and the equality in (\ref{SmatrixOmegaPole}) holds.  To complete the proof we need only show that $\left(\Pres C\right)\kwz\not=0$.  This follows from Theorem \ref{thm:scattering}, which implies that at $\kwz$, ${\cal N}=\sspan\{ \vrp,\vlp,\Psi^\guided \}$ is a three-dimensional subspace of $\CC^4$ with the property that $\Psi\in {\cal N}$ if and only if $\Pres C\Psi=0$.
\end{proof}

When $\ell\kwz=0$ at a real pair $\kwz$, there are two cases, resulting in different scattering behavior.  They are distinguished by whether the incident field $\Psi^\inc$ yields a (nonunique) solution at $\kwz$ or no solution at all.

\begin{enumerate}
\item  If $(\Pres C\Psi^\inc)\kwz=0$, there exist multiple solutions to the scattering problem at $\kwz$, and any two solutions differ by a multiple of $\Psi^\guided$,
\begin{equation}\label{multiplesolutions}
  \Psi^\out = \text{const.}\Psi^\guided - R\Preg C\Psi^\inc\,.
  \qquad
  \text{(at $\kwz$)}
\end{equation}
The projection of $\Psi^\out$ to the propagating modes is independent of the solution, that is, independent of the constant multiplying $\Psi^\guided$ because $\Psi^\guided$ is evanescent.  {\em Thus the far-field behavior is unique.}  The three-dimensional space $\cal N$ of incoming field traces $\Psi^\inc$ at $\kwz$ for which a solution~(\ref{multiplesolutions}) exists
satisfies $C\Psi^\inc\in \Range B$ and is characterized by Theorem \ref{thm:scattering} as
\begin{equation*}
  {\cal N} = \sspan\{ \vrp,\vlp,\Psi^\guided \}\,.
  \qquad
  \text{(non-resonant incident fields)}
\end{equation*}
In particular, $\cal N$ contains the propagating modes.  Scattering of propagating modes manifests interesting anomalous behavior in the transmission coefficient and field amplification.
\item  If $(\Pres C\Psi^\inc)\kwz\not=0$, the scattering problem has no solution at $\kwz$ or at any real $\kw$ on the dispersion relation near $\kwz$.  Field amplification is more pronounced than in the first case (see section~\ref{sec:amplification}).
\end{enumerate}

In case (1), not only is the far-field scattering behavior unique at $\kwz$, but, by fixing $\kk=\kk_0$ and varying $\omega$, a unique constant in (\ref{multiplesolutions}) is determined, making the scattering solution well defined at $\kwz$.  To see this, suppose that $\Psi^\inc\kw$ is analytic with
\begin{equation}\label{solvable}
  \Psi^\inc\kwz\in{\cal N}\,.
  \qquad \text{($\Psi^\out=S\Psi^\inc$ is solvable)}
\end{equation}
Then $\Psi^\out(\kk_0,\omega)$ is analytic in $\omega$ at $\omega_0$.  To see this, observe that
(\ref{solvable}) means $\Pres C\Psi^\inc\kwz = 0$ so that 
\begin{equation}\label{Psiout}
  \Pres C\Psi^\inc(\kk_0,\omega) = (\omega-\omega_0)\hat\Psi(\omega)
\end{equation}
in which $\hat\Psi(\omega)$ is an analytic eigenvector of $B(\kk_0,\omega)$ near $\omega_0$.  Applying (\ref{Smatrix}) at $\kk=\kk_0$ to $\Psi^\inc$ yields
\begin{equation}\label{omegaanalytic}
  \Psi^\out(\kk_0,\omega) = S\Psi^\inc = h(\kk_0,\omega)\hat\Psi(\omega) - R\Preg C\Psi^\inc(\kk_0,\omega).
\end{equation}
Comparing this to (\ref{multiplesolutions}) shows that
$\Psi^\out(\kk_0,\omega)$ extends analytically through $\omega_0$ as a solution (\ref{multiplesolutions}) to the scattering problem, and, at $\kwz$, selects a particular solution from the family (\ref{multiplesolutions}).

{\bfseries Remark.} Depending on the path in real $\kw$ space through $\kwz$ along which $\ell$ has a nonzero derivative at $\kwz$, a different multiple of the guided mode is determined in the limit at $\kwz$, changing the evanescent part of the response field but not its far-field behavior.

\subsection{The far-field scattering matrix}\label{sec:farfield}

Because the traces of the propagating modes at $\kwz$ are in $\sspan\{\vrp,\vlp\}\subset\cal N$ there is always a solution to the scattering problem for any {\em incident propagating field}, i.e., $\Psi^\inc \in \sspan\{\vrp,\vlp\}$.
By projecting the outgoing field of an incident propagating field onto the propagating modes, one obtains the far-field, or reduced, scattering matrix $S_0$.  It is an operator from $\sspan\{\vrp,\vlp\}$ to itself given~by
\begin{equation*}
  S_0\kw = P_pSP_p,
  \qquad
  \text{(far-field scattering matrix)}
\end{equation*}
where $P_p:=\Prp+\Plp$. More precisely, an incident propagating field with trace $\Psi^\inc_p\kw = j_{+p}\vrp+j_{-p}\vlp$ is mapped to the trace of propagating (far-field) modes of the outgoing field $\Psi^\out_p\kw = u_{+p}\vrp+u_{-p}\vlp$ by
\begin{equation*}
 \Psi^\out_p = S_0 \Psi^\inc_p
  = -P_p\big( \ell^{-1}\Pres (C\Psi^\inc_p) + R\Preg (C\Psi^\inc_p) \big).
\end{equation*}
As we have said, $\Pres C\Psi^\inc_p\kwz=0$ so that $\Psi^\out_p$ is well defined at $\kwz$ as the limit of the $\omega$-analytic function $\Psi^\out_p(\kk_0,\omega)$ at $\omega_0$.

By identifying $S_0$ with its matrix with respect to the basis $\{\vrp,\vlp\}$, one has
\begin{equation}
\begin{bmatrix}
u_{+p}\\
u_{-p}
\end{bmatrix}
=S_0\begin{bmatrix}
j_{+p}\\
j_{-p}
\end{bmatrix}.
\end{equation}
In view of $S_0=\ell^{-1}P_p\tilde SP_p$, in the basis $\{\vrp,\vlp\}$, $S_0$ has the form
\begin{equation*}
  \mat{1.3}{t_{+p}\kw}{r_{+p}\kw}{r_{-p}\kw}{t_{-p}\kw}=S_0\kw \,=\, -\ell^{-1} P_p\left( \Pres + \ell\,R\Preg \right)CP_p
  \,=\, \frac{1}{\ell\kw} \mat{1.3}{b_1\kw}{a_2\kw}{a_1\kw}{b_2\kw}\,,
\end{equation*}
and $\tilde S_0 = P_p \tilde S P_p$ has the form
\begin{equation}
  \tilde S_0\kw \,:=\, \mat{1.3}{b_1\kw}{a_2\kw}{a_1\kw}{b_2\kw}\,.
\end{equation}
The entries $a_i$ and $b_i$ are analytic at $\kwz$ because $\tilde S$, $\vrp$, and $\vlp$ are, and the transmission and reflections coefficients $\{t_{+p},t_{-p}\}$ and $\{r_{+p},r_{-p}\}$ are meromorphic functions of $\kw$ near $\kwz$.
 
If $\kw$ is real, then $S_0\kw$ is a unitary matrix.  Indeed, $S_0[j_{+p},j_{-p}]^\textrm{T}=[u_{+p},u_{-p}]^\textrm{T}$ means that there exists a solution $\psi$ of the Maxwell ODEs with
\begin{eqnarray*}
  && \psi(0) = j_{+p}\vrp + u_{-p}\vlp + u_{-e}\vle\,,\\
  && \psi(L) = j_{-p}\vlp + u_{+p}\vrp + u_{+e}\vre\,,
\end{eqnarray*}
and, since $[\psi(z),\psi(z)]$ is invariant in $z$, (\ref{conservation}) yields $|j_{+p}|^2-|u_{-p}|^2=|u_{+p}|^2-|j_{-p}|^2$.

At the real pair $\kwz$ on the dispersion relation $\ell\kw=0$, we have seen that $S_0$ is well defined and, because of (\ref{omegaanalytic}),
\begin{equation*}
  S_0\kwz=\lim_{\omega\to\omega_0}S_0(\kk_0,\omega)=-(P_pR\Preg CP_p)\kwz.
\end{equation*}
Since $\ell\kwz=0$, we have
\begin{equation*}
  \tilde S_0\kwz=0\,.
\end{equation*}

We have arrived at the key point (2) at the beginning of section~\ref{sec:scattering}.
{\em The sensitive nature of resonant scattering for real $\kw$ near a guided-mode pair $\kwz$ is embodied in a common zero of $a_j$, $b_j$, and $\ell$ at $\kwz$.}  This was the basis of the analysis of scattering resonances in~\cite{ShipmanVenakides2005}.
That work shows another way to see the confluence of the zero sets of $\ell$ and $\tilde S_0$ at $\kwz$ without directly utilizing the fact that $\sspan\{\vrp,\vlp\}\subset{\cal N}$.  One starts with the field $\Psi^\out$ in (\ref{analyticconnection}) that connects scattering states analytically to guided modes near $\kwz$.  The matrix $\tilde S_0\kw$ is obtained by projecting this field analytically onto the propagating modes.  The fact that $\tilde S_0\kwz=0$ comes from $\ell\kwz=0$ and the energy-flux conservation laws $|\ell|^2=|a_j|^2+|b_j|^2$ for real $\kw$.  The latter follows from setting the incident field in \ref{totalfield} to be a propagating mode from the right or from the left with appropriate coefficient $j_{\pm p}=\ell$.

\medskip
{\bfseries The transmitted energy.}\, Let ${\cal T}\kw^2$, for real $\kw$, denote the ratio of energy flux of the transmitted field to that of an incident propagating field from the left of the slab or the right.  By setting the incident field in \ref{totalfield} to be propagating from the left, i.e., $j_{+p}\not=0,j_{-p}=j_{-e}=j_{+e}=0$, or from the right, i.e., $j_{-p}\not=0,j_{+p}=j_{-e}=j_{+e}=0$,
the unitarity of $S_0\kw$ and the flux relations (\ref{canonicalform}) yield
\begin{align}\label{def:transmission}
   {\cal T}\kw^2&=|t_{\pm p}\kw|^2=1-|r_{\pm p}\kw|^2  \hspace{5em}
   \text{(transmitted energy)} \\
    & = \frac{|b_j\kw|^2}{|\ell\kw|^2} = 1-\frac{|a_j\kw|^2}{|\ell\kw|^2},
    \quad j=1,2 \nonumber
\end{align}
provided $\ell\kw\not=0$.  Thus ${\cal T}\kw=|t_{\pm p}|$ is the modulus of the transmission coefficient, independently of the side of the slab on which the incident field impinges.

At a guided-mode pair $\kw$, the transmission can be defined in the $\omega$-limit as
\begin{eqnarray}\label{TransmittanceAtGuidedModePair}
   {\cal T}\kwz^2:=\lim_{\omega\rightarrow \omega_0}{\cal T}(\kk_0,\omega)^2= \frac{|b'_j|^2}{|\ell'|^2} = 1-\frac{|a'_j|^2}{|\ell'|^2},
   \quad j=1,2,
\end{eqnarray}
where
\begin{eqnarray}
b'_j := \frac{\partial b_j}{\partial\omega}\kwz,
\quad
a'_j := \frac{\partial a_j}{\partial\omega}\kwz,
\quad
\ell':=\frac{\partial\ell}{\partial\omega}\kwz\not=0.
\end{eqnarray}
The fact that $\ell'\not=0$ follows from Theorem~\ref{thm:nondegeneracy}.

Energy-flux conservation takes the form
\begin{eqnarray}
|\ell\kw|^2=|a_j\kw|^2+|b_j\kw|^2,\;\;j=1,2,
\end{eqnarray}
for real $\kw$. 

\subsection{Scattering anomalies}\label{sec:anomalies}

The following two subsections analyze two fundamental resonance phenomena---transmission anomalies and field amplification.
The former deals with the resonant transmission of energy of a propagating mode across a slab and is based on the sensitivity of the reduced scattering matrix $S_0$ to frequency around a guided slab mode.  The latter requires consideration of the full scattering matrix $S$ because resonant amplification of a field in the slab is due to enhanced excitation of the evanescent modes in the ambient space just outside the~slab.

\subsubsection{Transmission anomalies}\label{sec:transmission}

The transmission graph in Figure~\ref{fig:resonance} exhibits a characteristic sharp double spike, or peak-dip shape, often called a ``Fano lineshape".  It is a sharp deviation, localized at the bound-state frequency, from a certain ``background transmission".  This background can be identified with the transmission graph of the unperturbed system ($\kk=\kk_0$ in the present context), which features no anomaly at all (the dotted graph).  Most of this subsection is an analysis of this anomaly, which leads to a transparent formula for it (\ref{T2b} and \ref{a/b}).  The analysis builds on \cite{ShipmanVenakides2005,Shipman2010}, which treats the case of a one-dimensional wavevector.

\smallskip
The relations (\ref{TransmittanceAtGuidedModePair}) imply the following conditions on the  
derivatives $a'_j$ and $b'_j$.

\begin{Lemma}\label{derivatives}
Either both of $a'_j$ ($j=1,2$) vanish or neither does; the same holds for~$b'_j$.  But at least one of the pairs $(a'_1,a'_2)$, $(b'_1,b'_2)$ must be nonzero.
\end{Lemma}

The extreme cases of full background transmission and full background reflection at the guided-mode pair $\kwz$ lead to limiting forms of the Fano-type shape.  Full reflection at $\kwz$, i.e., ${\cal T}\kwz=0$, means
\begin{equation*}
  b'_j=0,\;\;j=1,2
\end{equation*}
and leads to a ``Lorentzian lineshape" with a single peak at its center, and full transmission, i.e., ${\cal T}\kwz=1$, means
\begin{equation*}
  a'_j=0\,, \qquad j=1,2,
\end{equation*}
and results in an inverted Lorentzian shape.

\medskip
{\bfseries The case of a peak-dip anomaly.}\
When the derivatives of all the entries of $\tilde S_0$ are nonzero at $\kwz$, which is the case of neither full nor vanishing background transmission, the Weierstra{\ss} Preparation Theorem provides the following local forms:
\begin{equation}\label{Weierstrass1}
  \renewcommand{\arraystretch}{1.1}
\left.
  \begin{array}{ll}
    \ell\kw = \big(\tomega + \ell_1(\tkk) + \ell_2(\tkk,\tkk) + \dots \big)\hat\ell\kw\,, & \hat\ell\kwz=\ell'\not=0 \\
    a_1\kw = \big(\tomega + r_1(\tkk) + r_2(\tkk,\tkk) + \dots \big)\hat a_1\kw\,, & \hat a_1\kwz=a'_1\not=0 \\
    b_1\kw = \big(\tomega + t_1(\tkk) + t_2(\tkk,\tkk) + \dots \big)\hat b_1\kw\,, & \hat b_1\kwz=b'_1\not=0  \\
    a_2 = -e^{i\theta} \overline{a_1},\;\;b_2 = e^{i\theta} \overline{b_1},\;\;\theta=\theta\kw\in\RR & \text{for real } \kw
  \end{array}
\right.
\end{equation}
in which
\begin{equation*}
  \tilde\kk = \kk-\kk_0\,,
  \qquad
  \tilde\omega = \omega-\omega_0\,,
\end{equation*}
and the functions $\ell_1$, $r_1$, $t_1$ are linear, $\ell_2$, $r_2$, $t_2$ are bilinear forms, {\itshape etc.} of $\tkk$. The relation between $a_1$ and $a_2$ and that between $b_1$ and $b_2$ follows from the unitarity of the reduced scattering matrix $S_0\kw$ for real~$\kw$.

The analysis of resonant scattering anomalies relies on certain relations among the coefficients in the expansions~(\ref{Weierstrass1}).

\begin{Proposition}\label{prop:coefficients}\renewcommand{\labelenumi}{\alph{enumi}.}
\hspace{1em}
\begin{enumerate}
  \item $|\ell'|^2 = |a'_1|^2+|b'_1|^2$\,.
  \item $\ell_1(\tkk) = r_1(\tkk) = t_1(\tkk)\,\in\RR\,$ for all $\tkk\in\RR^2$.
  \item $\Im\ell_2(\tkk,\tkk)\geq0$ for all $\tkk\in\RR^2$.
  \item For all $\tkk\in\RR^2$,
\begin{equation*}
  \ell_2(\tkk,\tkk)\in\RR \iff r_2(\tkk,\tkk) = t_2(\tkk,\tkk)\in\RR \iff \ell_2(\tkk,\tkk) = r_2(\tkk,\tkk) = t_2(\tkk,\tkk)\in\RR
\end{equation*}
  \item  Given $\xxi\in\RR^2$,\ \ $\left| a'_1 \right|^2\Re\!\left( r_2(\xxi,\xxi)-\ell_2(\xxi,\xxi) \right) + \left| b'_1 \right|^2\Re\!\left( t_2(\xxi,\xxi)-\ell_2(\xxi,\xxi) \right) \,=\, 0$
  \item Given $\xxi\in\RR^2$, if $r_2(\xxi,\xxi), t_2(\xxi,\xxi)\in\RR$ then
\begin{align}
\big(r_2(\xxi,\xxi)-t_2(\xxi,\xxi)\big)^2&=2\big(\Im\ell_2(\xxi,\xxi)\big)^2+\big(r_2(\xxi,\xxi)-\operatorname{Re}\ell_2(\xxi,\xxi)\big)^2+\big(t_2(\xxi,\xxi)-\operatorname{Re}\ell_2(\xxi,\xxi)\big)^2,\label{IneqBetwePeakDip}\\
-\big(\Im\ell_2(\xxi,\xxi)\big)^2&=(r_2(\xxi,\xxi)-\operatorname{Re} \ell_2(\xxi,\xxi))(t_2(\xxi,\xxi)-\operatorname{Re}\ell_2(\xxi,\xxi)).\label{IneqCentrFreqBetwePeakDip}
\end{align}
  \item Given $\xxi\in\RR^2$, if $r_2(\xxi,\xxi), t_2(\xxi,\xxi)\in\RR$ then
\begin{equation*}
  \Re\ell_2(\xxi,\xxi)-t_2(\xxi,\xxi)=\pm\frac{\left| a'_1 \right|}{\left| b'_1 \right|} \Im\ell_2(\xxi,\xxi)\,,
  \qquad
  \Re\ell_2(\xxi,\xxi)-r_2(\xxi,\xxi)=\mp\frac{\left| b'_1 \right|}{\left| a'_1 \right|} \Im\ell_2(\xxi,\xxi)\,.
\end{equation*}
In particular, these differences are of opposite sign.
  \item Given $\xxi\in\RR^2$, if $a_1(\kk_0+\tau\xxi,\omega)=a_2(\kk_0+\tau\xxi,\omega)$ and $b_1(\kk_0+\tau\xxi,\omega)=b_2(\kk_0+\tau\xxi,\omega)$ for $|\tau|\ll1$ and $|\omega-\omega_0|\ll1$ and if $\Im\ell_2(\xxi,\xxi)>0$ and $r_2(\xxi,\xxi)\not=t_2(\xxi,\xxi)$, then all of the $k$-th order tensors $r_k$ and $t_k$ take on real values at~$\xxi$.
\end{enumerate}
\end{Proposition}

\begin{proof}
Items (a--d) are stated in Theorem~20 of~\cite{Shipman2010} in the case of a one-dimensional wavevector.  The proofs are straightforwardly extensible to the multi-dimensional case.  Briefly, (a) comes from~(\ref{TransmittanceAtGuidedModePair}), and that $\ell_1$ is real and $\ell_2$ has a nonnegative imaginary part comes from Theorem~\ref{thm:resolvent}, which disallows a frequency of a generalized slab mode to be in the upper half plane.  The other parts of (b--d) are proved by applying flux conservation $|\ell|^2=|a_1|^2+|b_1|^2$ for real $\kw$ along the relation $\tomega + \ell_1(\tkk)=0$, then along the relation $\tomega + \ell_1(\tkk) + r_2(\tkk,\tkk)=0$, {\itshape etc.} 
Equality (e) is proved by expanding (a) evaluated at $(\kk_0+\tau\xxi,\omega)$ in the variables $\tau$ and $\hat\omega=\tomega+\tau\ell_1(\xxi)$ and taking the coefficient of $\hat\omega\tau^2$.
The first equality in (f) and item (h) are proved in the one-dimensional case in Theorem~20 of \cite{Shipman2010} and Theorem~12 of \cite{ShipmanTu2012}, respectively.  Their proofs are extensible to the multi-dimensional case by considering a one-dimensional subspace of wavevectors $\{\tau\xxi : \tau\in\RR \}$, with $\xxi\in\RR^2$ fixed. The second equality in (f) follows from the first, which is proved by applying flux conservation $|\ell|^2=|a_1|^2+|b_1|^2$ for real $\kw$ along the relation $\tomega + \ell_1(\tkk)+ r_2(\tkk,\tkk)=0$ and then along $\tomega + \ell_1(\tkk)+t_2(\tkk,\tkk)=0$ for wavevectors $\tkk=\tau \xxi$, $|\tau|\ll 1$.
Part (g) comes from (e) and the second equation of (f).
\end{proof}

\smallskip
Because of part (b), the zero-sets of $\ell$, $a_j$ and $b_j$ are
\begin{equation}
  \renewcommand{\arraystretch}{1.1}
\left.
  \begin{array}{rll}
    \ell=0: & \tomega + \ell_1(\tkk) + \ell_2(\tkk,\tkk) + \dots = 0 &\text{(dispersion relation)} \\
    a_j=0: & \tomega + \ell_1(\tkk) + r_2(\tkk,\tkk) + \dots = 0 &\text{(zero reflection)} \\
    b_j=0: & \tomega + \ell_1(\tkk) + t_2(\tkk,\tkk) + \dots = 0 &\text{(zero transmission)}
  \end{array}
\right\}
\;
\text{with $\ell_1(\tkk)\in\RR$ for all $\tkk\in\RR^2$.}
\end{equation}
For real $\kw$, the transmission is 
\begin{equation}\label{T2b}
\renewcommand{\arraystretch}{1.1}
\left.
  \begin{array}{lcl}
  {\cal T}\kw^2 &=& \displaystyle\frac{|b_1|^2}{|\ell|^2}=\frac{|b_1|^2}{|a_1|^2+|b_1|^2}\\
\vspace{-1ex}\\
&=&
\displaystyle\frac{\left|b'_1\right|^2\,\left|\tomega + \ell_1(\tkk) + t_2(\tkk,\tkk) + \dots\right|^2}
    {\,\left|b'_1\right|^2\,\left|\tomega + \ell_1(\tkk) + t_2(\tkk,\tkk) + \dots\right|^2
              \;+\; \left|a'_1\right|^2\,\left|\tomega + \ell_1(\tkk) + r_2(\tkk,\tkk) + \dots\right|^2 \left|q(\tkk,\tomega)\right|^2\,}\,,    
  \end{array}
\right.
\end{equation}
in which $q\tkw$ is analytic at $(0,0)$ and $q(0,0)=1$.  It can be written alternatively as
\begin{equation}\label{a/b}
\renewcommand{\arraystretch}{1.1}
\left.
  \begin{array}{rcl}
     {\cal T}\kw^2 &=& \displaystyle\frac{1}{1+ |a_1/b_1|^2} \,,\\
  \vspace{-1ex}\\
       \displaystyle\frac{a_1}{b_1} &=&
  \displaystyle\frac{a_1'}{b_1'}\cdot
  \frac{\tomega + \ell_1(\tkk) + r_2(\tkk,\tkk) + \cdots}{\tomega + \ell_1(\tkk) + t_2(\tkk,\tkk) + \cdots}\cdot
  q\tkw\,.
  \end{array}
\right.
\end{equation}
At $\kk=\kk_0$, the transmission is continuous as a function of $\omega$,
\begin{equation*}
  {\cal T}(\kk_0,\omega)^2 = \frac{|b'_1|^2}{|b'_1|^2 + |a'_1|^2 |q(0,\tomega)|^2}=\frac{1}{1 + |a'_1/b'_1|^2 |q(0,\tomega)|^2}\,.
\end{equation*}

\medskip
{\bfseries Fano-type resonance.}\
Often, the transmission attains $0\%$ and $100\%$ at frequencies near $\omega_0$ when $\kk$ is perturbed from $\kk_0$, resulting in a clean transmission anomaly, as in the example in section~\ref{sec:example}.
This happens when the scattering problem is reciprocal, that is, when $a_1\!=\!a_2\!=\!a$ and $b_1\!=\!b_2\!=\!b$, which is guaranteed when the Maxwell ODEs admit symmetry in~$z$ about the centerline of the slab, as in Fig.~\ref{fig:layered}(right).
This is seen by Proposition~\ref{prop:coefficients}(h): if $\Im\ell_2(\xxi,\xxi)\!>\!0$ and $r_2(\xxi,\xxi)\not=t_2(\xxi,\xxi)$ for some unit wavevector $\xxi\in\RR^2$, then the transmission and reflection coefficients $a$ and $b$ vanish at {\em real} frequencies associated with the real wavevectors $\kk=\kk_0+\tau \xxi$.  Denote these frequencies by $\omega_\text{\tiny max}$ and $\omega_\text{\tiny min}$:
\begin{equation}
  \renewcommand{\arraystretch}{1.1}
\left.
  \begin{array}{ll}
    \omega_\text{\tiny max} = \omega_0 - \tau\ell_1(\xxi) - \tau^2r_2(\xxi,\xxi) + \cdots & \text{($100\%$ transmission).} \\
    \omega_\text{\tiny min} = \omega_0 - \tau\ell_1(\xxi) - \tau^2t_2(\xxi,\xxi) + \cdots & \text{($0\%$ transmission).}
  \end{array}
\right.
\end{equation}
If $\ell_1(\xxi)\not=0$, the anomaly is detuned linearly from the real guided-mode frequency $\omega_0$.

Along the line of wavevectors $\kk=\kk_0+\tau \xxi$, the frequency of the generalized guided mode is
\begin{eqnarray}
\omega_g=\omega_0-\tau\ell_1(\xxi)-\tau^2\ell_2(\xxi,\xxi) + \dots\,,
\qquad \text{(complex resonance frequency)}
\end{eqnarray} 
which lies below the real $\omega$-axis since $\Im\ell_2(\xxi,\xxi)\!>\!0$.
Its real and imaginary parts are interpreted as the {\em central frequency} and the {\em width} of the transmission resonance,
\begin{eqnarray}
\omega_\text{\tiny cent} &:=&\Re \omega_g=\omega_0-\tau\ell_1(\xxi)-\tau^2\Re\ell_2(\xxi,\xxi)+O(\tau^3)\,,
   \hspace{3em} \text{($\omega_\text{\tiny cent}=$ central frequency)} \label{centralfrequency}\\
-\textstyle\half\Gamma &:=& \Im \omega_g= -\tau^2\Im\ell_2(\xxi,\xxi)+\O(\tau^3)\,, \label{width}
   \hspace{8.5em} \text{($\Gamma=$ width)} \\
 Q &=& \frac{\left|\Re\omega_g\right|}{-2\,\Im\omega_g} = \frac{\left|\omega_\text{\tiny cent}\right|}{\Gamma}
                    = \frac{\left|\omega_0\right|}{-\tau^2\,2\,\Im\ell_2(\xxi,\xxi)} + O(|\tau|^{-1})\,. \label{Q}
               \hspace{1.4em}     \text{(quality factor)}
\end{eqnarray}

Importantly, \emph{the central frequency $\omega_\text{\tiny cent}=\omega_\text{\tiny cent}(\kk)=\omega_\text{\tiny cent}(\kk_0+\tau\xxi)$ lies between the frequencies of $0\%$ and $100\%$ transmission}.  This is seen through the relations
%
%
%
\begin{eqnarray*}
  \omega_\text{\tiny max}-\omega_\text{\tiny cent} &=& \pm\tau^2\frac{\left| a'_1 \right|}{\left| b'_1 \right|}\Im\ell_2(\xxi,\xxi) + O(\tau^3)\,,\\ 
  \omega_\text{\tiny min}-\omega_\text{\tiny cent} &=& \mp\tau^2\frac{\left| b'_1 \right|}{\left| a'_1 \right|}\Im\ell_2(\xxi,\xxi) + O(\tau^3)\,,
\end{eqnarray*}
for $0<|\tau|\ll 1$, which follows from Proposition~\ref{prop:coefficients}(g).
The width
\begin{equation}\label{Gamma}
  \Gamma\,\approx\,2\tau^2\Im\ell_2(\xxi,\xxi)
\end{equation}
is quadratically small in $\tau$.
A different but related measure of the width (also quadratic in $\tau$) is the distance between its peak and its dip,
\begin{equation}
  \omega_\text{\tiny max} - \omega_\text{\tiny min}
  \,=\, \tau^2 \left( t_2(\xxi,\xxi) - r_2(\xxi,\xxi) \right) + \O(\tau^3)\,. 
\end{equation}

The quality factor of a resonance, from (\ref{Qfactor}) in section~\ref{sec:EnergyFluxAndDensity} is the ratio of the energy of a generalized guided mode contained in a finite volume to the energy dissipated from that volume per cycle, 
It is inversely proportional to the spectral width of the resonance and thus blows up like $1/\tau^2$ as the parallel wavevector $\kk=\kk_0+\tau\xxi$ approaches a wavevector $\kk_0$ that supports a perfect guided mode.

\smallskip
One draws the following conclusions about transmission anomalies near a guided-mode wavevector-frequency pair $\kwz$, as $\kk$ is perturbed along $\xxi$, {\itshape i.e.,} $\kk=\kk_0+\tau \xxi$, given that $\Im\ell_2(\xxi,\xxi)\not=0$.
\begin{enumerate}\renewcommand{\labelenumi}{\alph{enumi}.}
  \item  $\cal T$ takes on all values between $0$ and $1$ in each real $\tau\omega$-ball around $(0,\omega_0)$.
  \item  The width of the anomaly is quadratic in $\,\tau$.
  \item  The order in which the peak and dip occur is independent of $\,\tau$.
  \item  If $\ell_1(\xxi)\not=0$, the central frequency of the anomaly is detuned linearly from the guided-mode frequency by $\Delta\omega=\omega_\text{\tiny cent}-\omega_0\sim-\ell_1(\xxi)\tau$. Thus the frequencies of the peak and the dip, $\omega_\text{\tiny max}$ and $\omega_\text{\tiny min}$, and central frequency $\omega_\text{\tiny cent}$, which differ by $O(\tau^2)$ are located to the same side of $\omega_0$ with $\omega_\text{\tiny cent}$ lying between $\omega_\text{\tiny max}$ and~$\omega_\text{\tiny min}$,
\begin{equation*}
    (\omega_\text{\tiny max}\!-\omega_0) \;\sim\; (\omega_\text{\tiny min}\!-\omega_0) \;\sim\; (\omega_\text{\tiny cent}-\omega_0) \;\sim\; -\ell_1(\xxi)\tau
    \qquad (\ell_1(\xxi)\not=0).
\end{equation*}
  \item  The Q-factor is inversely proportional to the width of the resonance, $Q\sim|\omega_0|/\Gamma$.
\end{enumerate}

A connection to the Fano formula~\cite{Fano1961} is made by writing the transmission in the form
\begin{equation*}
  {\cal T}\kw^2 \,=\, \frac{|b|^2}{|\ell|^2}
      \,=\, \left| b'_1 \right|^2 \frac{\left| \varpi + \tau^2\big( t_2(\xxi,\xxi) - \Re\ell(\xxi,\xxi) \big) + \dots \right|^2}
                                                 {\left| \varpi + i\tau^2 \big( \Im\ell_2(\xxi,\xxi) + \dots \big) \right|^2}
              \left| 1 + \dots \right|^2.
\end{equation*}
(In the fraction, the higher-order terms involve $\tau$, and in the last factor, they involve $\tau$ and $\omega-\omega_0$).
By ignoring the big-$O$ terms, one obtains simple peak-dip formulas for the transmission and refelection ${\cal R}^2 = 1- {\cal T}^2$,
\begin{eqnarray}\label{TR}
  {\cal T}\kw^2 &\approx& t_0^2\,\frac{(q+e)^2}{1+e^2}\,, \\
  {\cal R}\kw^2 &\approx& r_0^2\,\frac{(q'+e)^2}{1+e^2}\,,
\end{eqnarray}
through the following definitions.
\begin{equation*}
  \renewcommand{\arraystretch}{1.1}
\left.
  \begin{array}{ll}
    \varpi = \omega-\omega_\text{\tiny cent}(\kk_0+\tau\xxi) & \text{(deviation from central frequency)} \\
    \vspace{-2ex}\\
    \Gamma = 2\tau^2\,\Im\ell_2(\xxi,\xxi) +\dots & \text{(resonance width)} \\
    \vspace{-2ex}\\
    e = \displaystyle\frac{\varpi}{\Gamma/2} & \text{(normalized frequency)} \\
    \vspace{-2ex}\\
    q = \displaystyle\frac{t_2(\xxi,\xxi)-\Re\ell_2(\xxi,\xxi)}{\Im\ell_2(\xxi,\xxi)} & \text{(transmission asymmetry parameter)} \\
    \vspace{-2ex}\\
    q' = \displaystyle\frac{r_2(\xxi,\xxi)-\Re\ell_2(\xxi,\xxi)}{\Im\ell_2(\xxi,\xxi)} & \text{(reflection asymmetry parameter)} \\
    t_0 = |b'_j| \quad (j=1,2) & \text{(background transmission)} \\
    r_0 = |a'_j| \quad (j=1,2) & \text{(background reflection)} \\    
  \end{array}
\right.
\end{equation*}

\smallskip
{\bfseries Example of Section~\ref{sec:example} revisited.}
In the example, a true (exponentially decaying) guided mode was constructed for $k_2=0$ and for real $k_1$.  At $\kwz=(0.5,0;\,0.26015...)$, for instance, the slab supports a guided mode, and one has $\tkk=(\tilde k_1,k_2)=(k_1\!-\!0.5,k_2)$.  Because of the symmetry of the problem in the spatial variable $y$, the dispersion relation is even in $k_2$.
Thus, one can make the replacements $\ell_1(\tkk) \mapsto \ell_1\tilde k_1$ and $\ell_2(\tkk,\tkk) \mapsto \ell_{21}\tilde k_1^2 + \ell_{22}k_2^2$, with $\ell_{21}$ real,
\begin{equation}\label{omegag}
  \omega_g \,=\, \omega_0 - \ell_1\tilde k_1 - \ell_{21}\tilde k_1^2 - \Re\ell_{22}\,k_2^2 - i\,\Im\ell_{22}\,k_2^2 + \dots.
\end{equation}
Since the structure is symmetric about the center of the slab, $100\%$ and $0\%$ transmission is achieved along real curves
\begin{equation}\label{maxmin}
  \renewcommand{\arraystretch}{1.1}
\left.
  \begin{array}{ll}
    \omega_\text{\tiny max} = \omega_0 - \ell_1\tilde k_1 - \ell_{21}\tilde k_1^2 - r_2 k_2^2 + \cdots & \text{($100\%$ transmission).} \\
    \omega_\text{\tiny min} = \omega_0 - \ell_1\tilde k_1 - \ell_{21}\tilde k_1^2 - t_2 k_2^2 + \cdots & \text{($0\%$ transmission).}
  \end{array}
\right.
\end{equation}
These curves coincide with that for $\omega_g$ when $k_2\!=\!0$ because, by construction, the slab admits a perfect guided mode at a {\em real} frequency $\omega_g$, when $k_2=0$,
\begin{equation*}
  \omega_g = \omega_\text{\tiny cent} = \omega_\text{\tiny max} = \omega_\text{\tiny min}
  \qquad \text{(when $k_2=0$).}
\end{equation*}
When $k_2$ vanishes, the frequencies of $100\%$ and $0\%$ transmission coalesce and the anomaly disappears.  This apparent contradiction in the value of $\cal T$ is reflected in the fact that $\cal T$ takes on all values between $0$ and $1$ in every neighborhood of a point $(k_1,0;\,\omega_g)$ of a perfect guided mode.  Depending on the path along which $\kw$ approaches $(k_1,0;\,\omega_g)$, different limits of ${\cal T}\kw$ are attained.

As discussed in section~\ref{sec:example}, the the central frequency of a resonance is controlled by varying $k_1$, whereas the width is controlled by varying $k_2$.  This is true up to $O(|\tkk|^3)$, as seen either from the difference between the frequencies $\omega_\text{\tiny max}$ and $\omega_\text{\tiny min}$ in (\ref{maxmin}) or from the real and imaginary parts of $\omega_g$ (\ref{omegag}) with the definitions (\ref{centralfrequency},\ref{width}).
Fig.~\ref{fig:ComplexDispersion} shows two different paths $\kk=\kk_0+\tau\xxi$ in real $\kk$ space, one with $\xxi=(0, 1)$ and one with $\xxi=( 2/\sqrt{29}, 5/\sqrt{29} )$.  Along the first path, only $k_2$ varies, and the difference between the peak and dip of the transmission widens as the central frequency stays put, as shown in Fig.~\ref{fig:transmission}.  Along the second path, both $k_2$ and $k_1$ vary, effecting both the width and central frequency of the anomaly.  Notice that the ordering of the peak and the dip is independent of $\tau$.

\medskip
{\bfseries The Lorentzian case: $100\%$ or $0\%$ background transmission.}\
Full transmission at the guided-mode pair (${\cal T}\kwz=1$) occurs when $a_1'=a_2'=0$, and the expression (\ref{Weierstrass1}) is no longer valid for $a_1$.  Let us assume the generic condition
\begin{equation*}
  \frac{\partial^2 a_j}{\partial \omega^2}\kwz\not=0\,.
\end{equation*}
The Weierstra{\ss} Preparation Theorem provides the forms
\begin{equation}\label{Weierstrass2}
  \renewcommand{\arraystretch}{1.1}
\left.
  \begin{array}{ll}
    \ell\kw = \big(\tomega + \ell_1(\tkk) + \ell_2(\tkk,\tkk) + \dots \big)\,\hat\ell\kw\,, & \hat\ell\kwz\not=0 \\
    a_1\kw = \big(\tomega^2 + \tomega\alpha_1(\tkk) + \alpha_0(\tkk) \big)\,\hat a_1\kw\,, & \hat a_1\kwz\not=0 \\
    b_1\kw = \big(\tomega + \ell_1(\tkk) + t_2(\tkk,\tkk) + \dots \big)\,\hat b_1\kw\,, & |\hat b_1\kwz|=|\hat\ell\kwz|  
  \end{array}
\right.
\end{equation}
Here, $\alpha_0(\tkk)$ and $\alpha_1(\tkk)$ are analytic and vanishing at the origin.
One shows that the linear terms in $\ell$ and $b_1$ are indeed equal and that $|b'_1|=|\ell'|$ by using the flux-conservation relation $|\ell|^2 = |a_1|^2+|b_1|^2$ for real $\kw$.
In addition, by setting $\tomega=0$ and using flux conservation once again, one finds that
\begin{equation}\label{alpha0}
  \alpha_0(\tkk) = \O(|\tkk|^2).
\end{equation}

The zero-set of $a_1$ (and $a_2=-e^{i\theta}\overline{a_1}$) is obtained by factoring the quadratic Weierstra{\ss} polynomial.  To do this, write $\tkk=\tau\xxi$, where $\xxi$ is a vector and $\tau$ a complex scalar variable.  The factors are power series in $\tau^\onehalf$, but, because of (\ref{alpha0}), the leading term in $\tau$ is linear or higher,
\begin{equation}\label{puiseux}
  \tomega^2 + \tomega \alpha_1(\tau\xxi) + \alpha_0(\tau\xxi)
  = \left( \tomega + r_1^{(1)}(\xxi)\,\tau  + \O(\tau^\threehalves) \right)
     \left(  \tomega + r_1^{(2)}(\xxi)\,\tau  + \O(\tau^\threehalves) \right) \,.
\end{equation}
We are interested in the case that $\xxi$ is a real vector and $\tau$ is a real scalar.
In the case that the coefficients of the two Puiseux series in (\ref{puiseux}) and the coefficients of 
the polynomial $\tomega + \ell_1(\tau\xxi) + t_2(\tau\xxi,\tau\xxi) + \cdots$ in (\ref{Weierstrass2}), as a function of $\tau$, are real, 
one obtains one frequency of $0\%$ transmission and two frequencies of $100\%$ transmission,
\begin{equation}
  \renewcommand{\arraystretch}{1.1}
\left.
  \begin{array}{rll}
    a_j=0: &
    \renewcommand{\arraystretch}{1.1}
\left\{
  \begin{array}{l}
    \tomega + r_1^{(1)}(\xxi)\,\tau  + \O(\tau^\threehalves) = 0 \\
    \tomega + r_1^{(2)}(\xxi)\,\tau  + \O(\tau^\threehalves) = 0
  \end{array}
\right.
    & \text{(zero reflection)} \\
    \vspace{-2ex} \\
    b_j=0: & \tomega + \ell_1(\xxi)\,\tau + t_2(\xxi,\xxi)\,\tau^2 + \dots = 0 &\text{(zero transmission)}
  \end{array}
\right.
\end{equation}

\commentsps{How can we show that $\ell_1(\xxi)$ lies between the two coefficients $r_1^{(j)}(\xxi)$?  Using full expansion of flux conservation, the result would be in a higher order term that involves the factor multiplying the Weierstrass polynomial.}

\medskip
Zero transmission at the guided-mode pair (${\cal T}\kwz=0$) occurs when $b_1'=b_2'=0$.  This case is analyzed similarly.

\subsubsection{Resonant amplification}\label{sec:amplification}

Resonant amplification can be measured by the ratio $|\Psi^\out|/|\Psi^\inc|$ when it becomes unbounded. We assume now that $\Psi^\inc$ is analytic at $\kw=\kwz$ with $\Psi^\inc\kwz\not=0$. In (\ref{Smatrix}), the term $h\Pres C\Psi^\inc$ is an analytic multiple of an analytic eigenvector $\hat\Psi$ of $B$ corresponding to the simple eigenvalue $\ell$, 
\begin{equation*}
  \left( h \Pres C\Psi^\inc \right)\!\kw \,=\, \alpha\kw\hat\Psi\kw\,,
\end{equation*}
so that
\begin{equation*}
  \Psi^\out = S\Psi^\inc = \frac{\alpha\kw}{\omega-\omega_0+g(\kk)}\hat\Psi\kw
  \,+\,
  \left( R\Preg C\Psi^\inc \right)\kw,
\end{equation*}
in which $g(\kk) := \ell_1(\tkk) + \ell_2(\tkk,\tkk) + \dots$.
The second term is bounded, and thus we may neglect it when seeking the dominant behavior of the ratio $|\Psi^\out|/|\Psi^\inc|$, when this is large.  We might as well take $|\hat\Psi\kwz|=1$ to obtain
\begin{equation}\label{A}
  {\cal A} := \left| \frac{\alpha\kw}{\omega-\omega_0 + g(\kk)} \right|
  \,\sim\,
  \frac{|\Psi^\out|}{|\Psi^\inc|}
  \qquad
  \text{(when $|\Psi^\out|\gg|\Psi^\inc|$)}
\end{equation}
The quantity $\cal A$ is the {\em field amplification indicator}.

The dominant behavior of $\cal A$ near $\kwz$ is distinguished according to the two cases above, namely, whether the scattering problem is solvable at $\kwz$ or not.  Solvability means that $\alpha\kwz=0$, or that the coefficient $\alpha_0$ in the Taylor series of $\alpha$ vanishes,
\begin{equation*}
  \alpha\kw \;=\; \alpha_0 + \alpha_\kk\cdot\tilde\kk \,+\, \alpha_\omega \tilde\omega + \dots\,.
\end{equation*}
We have seen that $\alpha_0=0$ exactly when $\Pres C\Psi\inc\kwz=0$, or
\begin{equation*}
  \alpha_0=0
  \qquad \iff \qquad
  \Psi^\inc\kwz \in {\cal N} = \sspan\{ \vrp, \vlp, \Psi^\guided \}\,.
\end{equation*}

When $\kk=\kk_0$, i.e., $\tilde\kk=0$, $\omega_0$ is a frequency embedded in the continuum, and $\cal A$ simplifies to
\begin{equation}\label{A1}
  {\cal A} = \left| \frac{\alpha_0 + \alpha_\omega\tilde\omega + \dots}{\tilde\omega} \right|\,.
  \qquad
  \text{(along\; $\tilde\kk=0$)}
\end{equation}
Field amplification has particularly singular behavior when one approaches $\kwz$ along the transmission anomaly, which, to linear order, 
is the relation $\tomega+\ell_1(\tkk)=0$. Write $g$ in its real and imaginary parts (when evaluated at real $\kw$):
\begin{equation*}
  g(\kk) = g_1(\tkk) + ig_2(\tkk)\,.
\end{equation*}
When $\kw$ lies on the surface $\tilde\omega+g_1(\tilde\kk)=0$, $\cal A$ reduces~to
\begin{equation}\label{A2}
  {\cal A} = \left| \frac{\alpha_0 + \alpha_\kk\cdot\tilde\kk - \alpha_\omega g_1(\tkk) + \dots}{\Im\ell_2(\tilde\kk,\tilde\kk)+\O(|\tilde\kk|^3)} \right|\,.
  \qquad
  \text{(along\; $\tilde\omega+g_1(\tilde\kk)=0$)}
\end{equation}
The amplification is generically either linear or quadratic in $|\tkk|^{-1}$, depending on whether $\alpha_0$ vanishes or not. The behavior of the field amplification indicator is summarized in the following table, where in the right column we assume $\Im\ell_2(\tilde\kk,\tilde\kk)\not=0$ for $\tilde \kk\not= 0$.

\begin{center}
\begin{tabular}{l||c|c}
   & $\tilde\kk=0$\; ($\tilde\omega\to0$) & $\tilde\omega+g_1(\tilde\kk)=0$\; ($\tilde\kk\to0$) \\\hline\hline
\parbox{3em}{\hspace{1em}\\$\alpha_0=0$\vspace{1ex}} & no amplification & generically ${\cal A}\sim\frac{c}{|\tilde\kk|}$ \\\hline
\parbox{3em}{\hspace{1em}\\$\alpha_0\not=0$\vspace{1ex}} & ${\cal A}\sim\frac{\alpha_0}{\tilde\omega}$ & ${\cal A}\sim\frac{c}{|\tilde\kk|^2}$ \\\hline
\end{tabular}
\end{center}


\section{Non-degeneracy of guided modes}\label{sec:nondegeneracy}

True guided modes, which decay exponentially with distance from the defect layer, occur at real wavevector-frequency pairs $\kw$ on the dispersion relation $D\kw=\det B\kw=0$.  They are non-degenerate in two ways, described by Theorems~\ref{thm:scattering} and \ref{thm:nondegeneracy}, whose proofs are given below.

Theorem~\ref{thm:scattering} states that any zero eigenvalue of $B\kw$ for $\kw\in\RR^3$ is simple, and it gives necessary and sufficient conditions on the incoming field $\Psi^\inc$ that guarantee that the scattering problem has a solution.
Theorem~\ref{thm:nondegeneracy} states that the analytic analytic eigenvalue $\ell\kw$ of $B\kw$ that vanishes at $\kwz\in\RR^3$ has a nonzero derivative in $\omega$ at $\kwz$,
$
  \frac{\partial \ell}{\partial \omega}\kwz\,\not=\, 0\,.
$

\medskip

\begin{proof}[\bfseries Proof of Theorem~\ref{thm:scattering}]
  Let $\Psig$ be a nonzero solution of $(T\Pl-\Pr)\Psig=0$, and set
\begin{eqnarray*}
  && \Pl\Psig = u_{-p}\vlp + u_{-e}\vle\,, \\
  && \Pr\Psig = u_{+p}\vrp + u_{+e}\vre\,.
\end{eqnarray*}
Using the energy-interaction matrix (\ref{canonicalform}) and the flux-unitarity of $T$, one computes
\begin{equation*}
  -|u_{-p}|^2 = [\Pl\Psig,\Pl\Psig]
                = [T\Pl\Psig,T\Pl\Psig]
                = [\Pr\Psig,\Pr\Psig]
  =|u_{+p}|^2,
\end{equation*}
which implies that $u_{-p}=u_{+p}=0$.
The equation $(T\Pl-\Pr)\Psig=0$ now implies that $u_{-e}T\vle = u_{+e}\vre$, and since $\Psig$ is nonzero and $T$ is invertible, one concludes that $T^{-1}\vre = s\,\vle$, where $s=u_{-e}/u_{+e}\not=0$.
This proves that the nullspace is one-dimensional.

The solution (\ref{guidedmode2}) to the Maxwell ODEs is seen to have boundary values
$\psi^g(0) = u_{-e}^\guided\vle = \Ple\Psig$ and $\psi^g(L) = u_{+e}^\guided\vre = \Pre\Psig$ and therefore is the field corresponding to the boundary conditions defined by $(\Psi^\out,\Psi^\inc)=(\Psig,0)$.

By the Fredholm alternative, there exists a solution of (\ref{scattering2}) if and only if $(T\Pr-\Pl)\Psi^\inc$ is flux-orthogonal to the nullspace of $(T\Pl-\Pr)^\fadj$.  Now, $(T\Pl-\Pr)^\fadj=\Pl^\fadj T^{-1} - \Pr^\fadj$, and by (\ref{adjoints}),
\begin{equation*}
  \Pl^\fadj = \Plp + \Pre\,,
  \qquad
  \Pr^\fadj = \Prp + \Ple\,,
\end{equation*}
and $\Pl^\fadj$ and $\Pr^\fadj$ are complementary projections.  Thus,
\begin{eqnarray*}
  && \Psi\in\Null\left( (T\Pl-\Pr)^\fadj \right)
    \iff \left\{ \Pr^\fadj\Phi=0 \,\text{ and }\, \Pl^\fadj T^{-1}\Phi=0 \right\} \\
  && \iff \left\{ \Phi\in\Range\Pl^\fadj=\sspan\{\vlp,\vre\} \,\text{ and }\, T^{-1}\Phi\in\Range\Pr^\fadj=\sspan\{\vrp,\vle\} \right\} \\
  && \iff \left\{ \Phi=a\vre \,\text{ and }\, T^{-1}\Phi=b\vle \,\text{ for some }\, a,\,b\in\CC \right\}\,.
\end{eqnarray*}
The final equivalence is due to the equality $[\Phi,\Phi]=[T^{-1}\Phi,T^{-1}\Phi]$ and the flux interactions between modes given in (\ref{canonicalform}).  In words, $\Null\left( (T\Pl\!-\!\Pr)^\fadj \right)$ consists of the set of boundary data at $z=L$ of all guided modes.  In conclusion,
\begin{equation*}
    \Null\left( (T\Pl\!-\!\Pr)^\fadj \right) = \Pr \left(\Null(T\Pl\!-\!\Pr)\right)=\sspan\{\vre\}\,.
\end{equation*}

Now let $\Psi^\inc=j_{+p}\vrp + j_{+e}\vre + j_{-p}\vlp + j_{-e}\vle$ be given, then
\begin{equation*}
  (T\Pr-\Pl)\Psi^\inc \,=\, j_{+p} T\vrp + j_{+e} T\vre - j_{-p}\vlp - j_{-e}\vle\,.
\end{equation*}
The condition that $(T\Pr-\Pl)\Psi^\inc$ is flux-orthogonal to the nullspace of $(T\Pl-\Pr)^\fadj$ is that
$[(T\Pr-\Pl)\Psi^\inc,\vre]=0$, under the condition that $T^{-1}\vre = s\,\vle$.  Using the flux relations again yields
\begin{multline*}
  0 = [j_{+p} T\vrp + j_{+e} T\vre - j_{-p}\vlp - j_{-e}\vle,\vre]
     = [j_{+p}\vrp+j_{+e}\vre,T^{-1}\vre] - \overline{j_{-e}}[\vle,\vre] \\
     = \overline{j_{+e}}\,s[\vre,\vle] - \overline{j_{-e}}[\vle,\vre] = \overline{j_{+e}}\,s - \overline{j_{-e}}\,,
\end{multline*}
which, in view of $s=u_{-e}/u_{+e}$ proves the relation in part 2.

To prove that the algebraic multiplicity of the zero eigenvalue is $1$, one shows that the nullspace and the range of $T\Pl\!-\!\Pr$ intersect trivially.  Now, $\Range(T\Pl-\Pr) = \Null(\Pl^\fadj T^{-1}-\Pr^\fadj)^\fperp$, and
\begin{multline*}
  \Phi\in\Null(\Pl^\fadj T^{-1}\!-\!\Pr^\fadj)
  \iff \left\{ \Phi\in\Range\Pl^\fadj=\sspan\{\vlp,\vre\} \,\text{ and }\, T^{-1}\Phi\in\Range\Pr^\fadj=\{\vrp,\vle\} \right\} \\
  \iff \left\{ \Phi = a\vre \,\text{ and }\, T^{-1}\Phi = b\vle \,\text{ for some }\, a,\,b\in\CC \right\}\,.
\end{multline*}
Since $T^{-1}\vre = s\,\vle$, one concludes that $\Null(\Pl^\fadj T^{-1}\!-\!\Pr^\fadj)=\sspan\{\vre\}$.
Therefore,
\begin{equation*}
  \Psi\in\Range(T\Pl-\Pr)
  \iff 0 = [\Psi,\vre] = [\Ple\Psi,\vre]
\end{equation*}
and thus $\Range(T\Pl-\Pr)=\sspan\{\vlp,\vrp,\vre\}$, which does not contain the generator of the nullspace of $(T\Pl\!-\!\Pr)$, which is $s\vle+\vre$ ($s\not=0$).
\end{proof}

\medskip

\begin{proof}[\bfseries Proof of Theorem~\ref{thm:nondegeneracy}]\label{prf:thm:nondegeneracy}
Set $B=T\Pl\!-\!\Pr$. Suppose $B\kw$ is noninvertible at a real pair $\kwz$. Then by Theorem \ref{thm:scattering} the zero eigenvalue of $B\kwz$ is an algebraically simple eigenvalue. Thus because $B\kw$ is analytic in a complex neighborhood of $\kwz$ then by perturbation theory there exists a unique eigenvalue $\ell\kw$ of $B\kw$ such that $\lim_{\kw\rightarrow\kwz}\ell\kw=0$. Furthermore, by perturbation theory this eigenvalue of $B\kw$ is algebraically simple and analytic in a complex neighborhood of $\kwz$. Now for the rest of the proof we fix $\kappa=\kappa_0$ with primed functions denoting the derivative with respect to $\omega$ at $\omega_0$.

Let $\Psig$ be any eigenvector in the one-dimensional nullspace of $B(\omega)$. Then by perturbation theory there exists an analytic eigenvector $\Psi(\omega)$ of $B(\omega)$ corresponding to the eigenvalue $\ell(\omega)$ in a complex neighborhood of $\omega=\omega_0$ such that $\Psi(\omega)=\Psig$. Now differentiating the relation
\begin{equation*}
  (\ell(\omega)I-B(\omega))\Psi(\omega) = 0
\end{equation*}
with respect to $\omega$, evaluating at $\omega_0$, and using the fact that $\ell(\omega_0)=0$ yields 
\begin{equation}\label{diffeig}
  (\ell'(\omega_0)I-B'(\omega_0))\Psig - B(\omega_0)\Psi'(\omega_0) \,=\, 0\,.
\end{equation}
Convenient is the notation $\Psig=\Psigl+\Psigr$, with $T\Psigl=\Psigr$, where $\Pl\Psig=\Psigl=u_{-e}^\guided\vle$ and $\Pr\Psig=\Psigr=u_{+e}^\guided\vre$. In particular, since $T$ is invertible this implies $u_{-e}^\guided,u_{+e}^\guided\not=0$.

Applying $[\Psigr,\cdot]$ to (\ref{diffeig}) leads to
\begin{equation*}
  \ell'(\omega_0)[\Psigr,\Psig]
  - [\Psigr,B'(\omega_0)\Psig]
  - [B(\omega_0)^\fadj\Psigr,\Psi'(\omega_0)]
  \,=\, 0\,.
\end{equation*}
Because $B^\fadj=\Pl^\fadj T^{-1} - \Pr^\fadj$ for real $\omega$, one obtains
\begin{equation*}
  \ell'(\omega_0)[\Psigr,\Psigl] = [\Psigr,B'(\omega_0)\Psig]\,.
\end{equation*}
Observe that $[\Psigr,\Psigl]=[u_{+e}^\guided\vre,u_{-e}^\guided\vle]=\overline{u_{+e}^\guided}u_{-e}^\guided\not=0$.
To prove that $\ell'(\omega_0)\not=0$, it suffices to prove that $[\Psigr,B'(\omega_0)\Psig]\not=0$. We prove this by considering the three terms in the right side of the identity
\begin{eqnarray}
  [(T\Pl-\Pr)'\Psig,\Psigr] &=& [T'\Pl\Psig,\Psigr] + [T\Pl'\Psig,\Psigr] - [\Pr'\Psig,\Psigr] \notag \\
     &=& [T'\Psigl,\Psigr] + [\Pl'\Psig,\Psigl] - [\Pr'\Psig,\Psigr]\,. \label{threeterms}
\end{eqnarray}

To treat the first term in the latter expression, use the fact that for the canonical Maxwell ODEs (\ref{canonical}) its transfer matrix $T(0,z)$ satisfies the matrix ODE $-iJdT(0,z)/dz = A(z;\omega)T(0,z)$ with $T(0,0)=I$ so that for real $\omega$ by differentiating both sides of this matrix ODE with respect to $\omega$ we obtain for $T(0,z)$ the identity
\begin{equation}
  \frac{\partial }{\partial z}\left( T^*(iJ)\frac{\partial T}{\partial \omega} \right) = T^*\frac{\partial A}{\partial \omega}T,
\end{equation}
from which it follows that
\begin{equation*}
  \int_0^L \left( A'\psi_1(z), \psi_2(z) \right)dz = -i\left( JT'\psi_1(0), \psi_2(L) \right).
\end{equation*}
where $A'$ is an abbreviation for $\frac{\partial A}{\partial \omega}(z;\omega_0)$.  Application of this identity to the guided-mode solution $\psi^g(z)$ yields
\begin{equation}\label{term1}
  [T'\Psigl,\Psigr] = \frac{c}{16\pi}\left( JT'\Psigl,\Psigr \right)
  = \frac{c}{16\pi i} \int_0^L \left( \frac{\partial A}{\partial \omega}(z;\omega_0)\psig(z), \psig(z) \right)dz\,.
\end{equation}

To deal with the other two terms of (\ref{threeterms}), one needs a general fact about derivatives of projections:
\begin{equation*}
  \Plr' = \Pr\Plr'\Pl + \Pl\Plr'\Pr\,.
\end{equation*}
The term $[\Pl'\Psig,\Psigl]$ becomes
\begin{eqnarray*}
  [\Pl'\Psig,\Psigl] &=& [\Pl'\Psigl,\Pr^\fadj\Psigl] + [\Pl'\Psigr,\Pl^\fadj\Psigl] \\
             &=& [\Pl'\Psigl,\Psigl] \,=\, [\Ple'\Psigl,\Psigl]\,.
\end{eqnarray*}
The final equality follows from $\Pl=\Ple+\Plp$ and
\begin{eqnarray*}
  [\Plp'\Psigl,\Psigl] &=& [(\Pr+\Ple)\Plp'\Plp\Psigl,\Psigl] + [\Plp\Plp'(\Pr+\Ple)\Psigl,\Psigl] \\
            &=& [\Plp\Plp'\Psigl,\Psigl] \,=\, [\Plp'\Psig,\Plp\Psigl] \,=\, 0\,.
\end{eqnarray*}
Define an analytic eigenvector of $K(\omega)$ for $\omega$ in a neighborhood of $\omega_0$ by
\begin{equation*}
  \Psil(\omega) := \Ple(\omega)\Psigl\,,
  \qquad
  \Psil(\omega_0) = \Psigl\,.
\end{equation*}
The corresponding analytic eigenvalue $k_{-e}(\omega)$ has negative imaginary part for $\omega$ near $\omega_0$.
Differentiating the equation $(K-k_{-e})\Psil = 0$ and evaluating at $\omega_0$ yields
\begin{equation*}
  \left( K'-k_{-e}' \right) \Psigl + \left( K-k_{-e} \right)\Ple'\Psigl \,=\, 0\,.
\end{equation*}
which, upon applying $[\cdot,\Psigl]$ yields
\begin{equation*}
  [K'\Psigl,\Psigl] + [K\Ple'\Psigl,\Psigl] = k_{-e}^*[\Ple'\Psigl,\Psigl]\,,
\end{equation*}
and then using the flux-self-adjointness of $K$,
\begin{multline*}
  (k_{-e}^*-k_{-e})[\Ple'\Psigl,\Psigl] = [K'\Psigl,\Psigl] = \frac{c}{16\pi}\left( JK'\Psigl,\Psigl \right) \\
           = \frac{c}{16\pi i}(k_{-e}^*-k_{-e}) \int_{-\infty}^0 e^{i(k_{-e}-k_{-e}^*)z}\left( JK'\Psigl,\Psigl \right) dz \\
           = \frac{c}{16\pi i}(k_{-e}^*-k_{-e}) \int_{-\infty}^0 \left( JK'\Psigl e^{ik_{-e}z},\Psigl e^{ik_{-e}z} \right) dz
           = \frac{c}{16\pi i}(k_{-e}^*-k_{-e}) \int_{-\infty}^0 \left( A'\psig(z),\psig(z) \right) dz\,
\end{multline*}
where the latter equality follows from Theorem \ref{thm:EnergyIndRep}. Thus we conclude that
\begin{equation}\label{term2}
  [\Pl'\Psig,\Psigl] \,=\, [\Ple'\Psigl,\Psigl] \,=\, \frac{c}{16\pi i} \int_{-\infty}^0 \left( A'\psig(z),\psig(z) \right)dz\,.
\end{equation}

A similar calculation is valid for the third term of (\ref{threeterms}):
Define $\Psir(\omega):=\Pre(\omega)\Psigr$.  The differentiated eigenvalue problem at $\omega=\omega_0$ is
\begin{equation*}
    \left( K'-k_{+e}' \right) \Psigr + \left( K-k_{+e} \right)\Pre'\Psigr \,=\, 0\,,
\end{equation*}
and $\Im(k_{+e})>0$ near $\omega_0$.
Application of $[\cdot,\Psigr]$ yields
\begin{equation*}
  [K'\Psigr,\Psigr] + [K\Pre'\Psigr,\Psigr] = k_{+e}^*[\Pre'\Psigr,\Psigr]\,,
\end{equation*}
and the calculation
\begin{multline*}
  (k_{+e}^*-k_{+e})[\Pre'\Psigr,\Psigr] = [K'\Psigr,\Psigr] = \frac{c}{16\pi}\left( JK'\Psigr,\Psigr \right) \\
           = -\frac{c}{16\pi i}(k_{+e}^*-k_{+e}) \int_L^\infty e^{i(k_{+e}-k_{+e}^*)(z-L)}\left( JK'\Psigr,\Psigr \right) dz \\
           = -\frac{c}{16\pi i}(k_{+e}^*-k_{+e}) \int_L^\infty \left( JK'\Psigr e^{ik_{+e}(z-L)},\Psigr e^{ik_{+e}(z-L)} \right) dz
           = -\frac{c}{16\pi i}(k_{+e}^*-k_{+e}) \int_L^\infty \left( A'\psig(z),\psig(z) \right) dz,
\end{multline*}
where the latter equality follows from Theorem \ref{thm:EnergyIndRep}. Thus we conclude that
\begin{equation}\label{term3}
  -[\Pr'\Psig,\Psigr] \,=\, -[\Pre'\Psigr,\Psigr] \,=\, \frac{c}{16\pi i} \int_L^\infty \left( A'\psig(z),\psig(z) \right)dz\,.  
\end{equation}
In conclusion, since the integral is real-valued, (\ref{threeterms}) together with equations (\ref{term1},\ref{term2},\ref{term3}) yield
\begin{equation*}
  [\Psigr,(T\Pl-\Pr)'\Psig] \,=\, \frac{ic}{16\pi} \int_{-\infty}^\infty \left( A'\psig(z),\psig(z) \right)dz\,.
\end{equation*}
and Theorem \ref{thm:energy} implies 
\begin{equation*}
\frac{c}{16\pi} \int_{-\infty}^\infty \left( A'\psig(z),\psig(z) \right)dz=\int_{-\infty}^\infty U^g(z)\,dz>0.
\end{equation*}
The proof now follows from these facts.
\end{proof}


\section{The Maxwell ODEs and EM energy in layered media}\label{sec:reduction} 

In this section we discuss the reduction of the time-harmonic Maxwell's equations (\ref{MaxwellEqsTimeHar}) to the canonical Maxwell ODEs (\ref{canonical}) (i.e., the reduced Maxwell's equations) and the dependence of the solutions and energy density on the system parameters. 

The reduction method essentially comes from the theory of linear differential-algebraic equations (DAEs) where, via a linear transformation, solutions to a system of DAEs are the solutions to a system of linear ODEs. Although it is known for layered media, Maxwell's equations can be reduced to a system of linear ODEs and their solutions are given in terms of the transfer matrix \cite{Berreman1972}, it is has not been clearly shown until now how these ODEs and their solutions are connected to the most important physical attributes that of energy flux and energy density. This connection is of vital importance in the study of dissipation, dispersion, and scattering in layered media. Thus we carry out a rigorous analysis of these connections deriving results that are needed in this paper, for instance, to prove Theorems \ref{thm:EnergyConservationLaw}, \ref{thm:EnergyDensity}, and \ref{thm:EnergyIndRep} which are needed to prove two of our main results, namely, Theorems \ref{thm:nondegeneracy}, \ref{thm:SmatrixLocal}.

In this paper, frequency-independent and $z$-dependent material tensors $\epsilon,\mu$ are Hermitian matrix-valued functions which are bounded (measurable) and coercive, that is, there exists constants $c_1,c_2>0$ such that
\begin{eqnarray*}
0<c_1I\leq \epsilon(z),\mu(z)\leq c_2I
\end{eqnarray*}
for all $z\in\mathbb{R}$. We are considering Maxwell's equations in Cartesian coordinates, Gaussian units, and one-dimensional layered media whose plane parallel layers have positively-directed normal vector $\mathbf{e}_{3}$. The electromagnetic fields will be written as column vectors.

We consider solutions of the time-harmonic Maxwell's equations (\ref{MaxwellEqsTimeHar}) of the form
\begin{equation*}
\begin{bmatrix}
\EE(\mathbf{r}) \\
\HH(\mathbf{r})
\end{bmatrix}
=
\begin{bmatrix}
\EE\left(  z\right) \\
\HH\left(  z\right)
\end{bmatrix}
e^{i\left(  k_{1}x+k_{2}y-\omega t\right)  }
\end{equation*}
for ${\bf k}_{\parallel}=\left(  k_{1},k_{2},0\right)  \in\mathbb{C}^{2}$ and $\omega\in\mathbb{C}\setminus\left\{  0\right\}  $. Set $\kk=(k_1,k_2)$. \ Their $z$-dependent factors are solutions to
the linear differential--algebraic equations (DAEs)
\begin{equation}
i^{-1}\frac{d}{dz}G
\begin{bmatrix}
\EE\left(  z\right) \\
\HH\left(  z\right)
\end{bmatrix}
=V\left(  z;\kk,\omega\right)
\begin{bmatrix}
\EE\left(  z\right) \\
\HH\left(  z\right)
\end{bmatrix}
\label{DAEs}
\end{equation}
where $G$ and $V\left(  z;\kk,\omega\right)  $ are the block
matrices
\begin{eqnarray}\label{def:BlockMatricesInDAEs}
G=
\begin{bmatrix}
0 & i\mathbf{e}_{3}\times\\
-i\mathbf{e}_{3}\times & 0
\end{bmatrix}
,\text{ \ \ }V\left(  z;\kk,\omega\right)  =\frac{\omega}{c}
\begin{bmatrix}
\epsilon\left(  z\right)  & 0\\
0 & \mu\left(  z\right)
\end{bmatrix}
+%
\begin{bmatrix}
0 & {\bf k}_{\parallel}\times\\
-{\bf k}_{\parallel}\times & 0
\end{bmatrix}
.
\end{eqnarray}
In particular, these are DAEs and not ODEs since $\det G=0$. \ As the theory
of solutions of DAEs are not as well known as for ODEs, we will only briefly
describe below the solutions to these DAEs without explicit reference to the
general theory of solutions of DAEs begin used.

We begin by introducing the matrices 
$P_{\parallel}:\mathbb{C}^{6}\rightarrow\mathbb{C}
^{4}$, $P_{\bot}:\mathbb{C}^{6}\rightarrow\mathbb{C}^{2}$ which map the $z$-dependent factors $F=\left [
\EE, \HH\right]^{T}$ onto their tangential components $\psi=
\left [
E_{1} , E_{2} , H_{1} , H_{2}
\right]
^{T}$ and normal components $\phi=
\left[
E_{3} , H_{3}
\right]
^{T}$, respectively, that is,
\[
P_{\parallel}F=\psi,\text{ \ \ }P_{\bot}F=\phi,
\]
where
\begin{eqnarray}
P_{\parallel}=
\begin{bmatrix}
1 & 0 & 0 & 0 & 0 & 0\\
0 & 1 & 0 & 0 & 0 & 0\\
0 & 0 & 0 & 1 & 0 & 0\\
0 & 0 & 0 & 0 & 1 & 0
\end{bmatrix}
,\text{ \ \ }P_{\bot}=
\begin{bmatrix}
0 & 0 & 1 & 0 & 0 & 0\\
0 & 0 & 0 & 0 & 0 & 1
\end{bmatrix}
.
\end{eqnarray}
Then any factor $F$ may be represented in block form as
\begin{eqnarray}
F=
\begin{bmatrix}
\psi\\
\phi
\end{bmatrix}
\end{eqnarray}
and $G$, $V$ (suppressing the explicit dependence on the parameters) with
respect to this block form can be written as the block matrices
\begin{equation}\label{BlockMatricesInDAEsPrime}
G=
\begin{bmatrix}
J & 0\\
0 & 0
\end{bmatrix}
,\ \ V=
\begin{bmatrix}
V_{\parallel\parallel} & V_{\parallel\bot}\\
V_{\bot\parallel} & V_{\bot\bot}
\end{bmatrix},\tag{\ref{def:BlockMatricesInDAEs}$^\prime$}
\end{equation}
where the block components are defined in (\ref{def:BlockCompsGandV}).
Next, we introduce the Schur complement of this block representation of $V$
with respect to the block $V_{\bot\bot}$, i.e.,
\begin{eqnarray}
A=V_{\parallel\parallel}-V_{\parallel\bot}\left(  V_{\bot\bot}\right)
^{-1}V_{\bot\parallel},\label{SchurComp}
\end{eqnarray}
and the matrix
\begin{eqnarray}
\Phi=-\left(  V_{\bot\bot}\right)  ^{-1}V_{\bot\parallel},
\end{eqnarray}
where
\begin{eqnarray}\label{def:BlockCompsGandV}
&J=P_{\parallel}GP_{\parallel}^{\ast}=
\begin{bmatrix}
0 & 0 & 0 & 1\\
0 & 0 & -1 & 0\\
0 & -1 & 0 & 0\\
1 & 0 & 0 & 0
\end{bmatrix}
,\\
&V_{\bot\bot}=P_{\bot}VP_{\bot}^{\ast}=\frac{\omega}{c}
\begin{bmatrix}
\epsilon_{33} & 0\\
0 & \mu_{33}
\end{bmatrix}
,\nonumber\\
&V_{\parallel\parallel}=P_{\parallel}VP_{\parallel}^{\ast}=\frac{\omega}{c}
\begin{bmatrix}
\epsilon_{11} & \epsilon_{12} & 0 & 0\\
\epsilon_{21} & \epsilon_{22} & 0 & 0\\
0 & 0 & \mu_{11} & \mu_{12}\\
0 & 0 & \mu_{21} & \mu_{22}
\end{bmatrix}
,\nonumber\\
&V_{\parallel\bot}=P_{\parallel}VP_{\bot}^{\ast}=\frac{\omega}{c}\left[
\begin{array}
[c]{cc}%
\epsilon_{13} & 0\\
\epsilon_{23} & 0\\
0 & \mu_{13}\\
0 & \mu_{23}
\end{array}
\right]  +\left[
\begin{array}
[c]{cc}
0 & k_{2}\\
0 & -k_{1}\\
-k_{2} & 0\\
k_{1} & 0
\end{array}
\right]  ,\nonumber\\
&V_{\bot\parallel}=P_{\bot}VP_{\parallel}^{\ast}=\frac{\omega}{c}\left[
\begin{array}
[c]{cccc}
\epsilon_{31} & \epsilon_{32} & 0 & 0\\
0 & 0 & \mu_{31} & \mu_{32}
\end{array}
\right]  +\left[
\begin{array}
[c]{cccc}
0 & 0 & -k_{2} & k_{1}\\
k_{2} & -k_{1} & 0 & 0
\end{array}
\right]  .\nonumber
\end{eqnarray}
In particular, the following Hermitian property holds:
\begin{equation}
A^*=A,\;\;\text{for every }(z,\kk,\omega)\in \RR\times\RR^2\times \RR/\{0\}.\label{MatrixAinMaxwellODEsIsHermitian} 
\end{equation}
Now in the block form (\ref{BlockMatricesInDAEsPrime}) the DAEs (\ref{DAEs}) become
\begin{equation}
\begin{bmatrix}
i^{-1}J\frac{d\psi}{dz}\\
0
\end{bmatrix}
=%
\begin{bmatrix}
V_{\parallel\parallel} & V_{\parallel\bot}\\
V_{\bot\parallel} & V_{\bot\bot}
\end{bmatrix}
\begin{bmatrix}
\psi\\
\phi
\end{bmatrix}
.\tag{\ref{DAEs}$^\prime$}\label{bDAEs}
\end{equation}
Therefore, the solutions $F=\left[  \EE,\HH\right]  ^{T}$ to the DAEs (\ref{DAEs}) are precisely those
satisfying
\begin{eqnarray*}
&i^{-1}J\frac{d\psi}{dz}=A\psi,\text{ \ \ }\psi\in\left(  AC_{loc}\left(\mathbb{R}\right)  \right)  ^{4},\\
&\phi =\Phi\psi,
\end{eqnarray*}
whose tangential and normal components are
\[
P_{\parallel}F=\psi,\text{ \ \ }P_{\bot}F=\phi.
\]
Here $\left(  AC_{loc}\left(\mathbb{R}\right)  \right)  ^{4}$ denotes the space of column vectors whose entries are locally absolutely continuous functions of $z$. On the other hand, from any solution of the canonical Maxwell ODEs
\begin{equation}
i^{-1}J\frac{d\psi}{dz}=A\psi,\text{ \ \ }\psi\in\left( AC_{loc}\left(\mathbb{R}\right)  \right)  ^{4}\label{cODEs}
\end{equation}
we get a solution $F=\left[  E,H\right]  ^{T}$ to the DAEs by setting
\begin{eqnarray}
F=P_{\parallel}^{\ast}\psi+P_{\bot}^{\ast}\Phi\psi
\end{eqnarray}
and, in particular, its tangential and normal components are
\begin{eqnarray}
P_{\parallel}F=\psi,\text{ \ \ }P_{\bot}F=\Phi\psi.
\end{eqnarray}

Now for the proofs in this paper we need to describe the dependency of the
solutions of the canonical Maxwell ODEs and DAEs as well as the energy density on the parameters $\left(
\kk,\omega\right)  \in\mathbb{C}^{2}\times\mathbb{C}\setminus\left\{  0\right\}  $. \ The solutions of the Maxwell ODEs are
given in terms of the $4\times 4$ transfer matrix $T\left(  z_{0},z;\kk,\omega\right)  $ via matrix multiplication%
\[
\psi\left(  z\right)  =T\left(  z_{0},z;\kk,\omega\right)
\psi_0,\text{ \ \ }\psi\left(  z_{0}\right)  =\psi_0,
\]
where $T\left(  z_{0},\cdot;\kk,\omega\right)  $ is the unique
solution to the integral equation
\begin{eqnarray*}
&T\left(  z_{0},z;\kk,\omega\right)  =I_{4}+\int_{z_{0}}^{z}
iJ^{-1}A\left(  z_{1};\kk,\omega\right)  T\left(  z_{0}
,z_{1};\kk,\omega\right)  dz_{1},\\
&T\left(  z_{0},\cdot;\kk,\omega\right)  \in M_{4}\left(
AC_{loc}\left(\mathbb{R}\right)  \right)  .
\end{eqnarray*}
Here $M_{4}\left(AC_{loc}\left(\mathbb{R}\right)\right)$ denotes the space of $4\times 4$ matrices whose entries are locally absolutely continuous functions of $z$. This transfer matrix is given explicitly by the Peano-Baker series%
\begin{equation}
T\left(  z_{0},z;\kk,\omega\right)  =I_{4}+\sum_{j=1}^{\infty
}\mathcal{I}_{j}\left(  z_{0},z;\kk,\omega\right)
,\label{PeanoBakerSeries}
\end{equation}
where
\begin{eqnarray*}
&\mathcal{I}_{1}\left(  z_{0},z;\kk,\omega\right)  =\int_{z_{0}
}^{z}iJ^{-1}A\left(  z_{1};\kk,\omega\right)  
dz_{1},\\
&\mathcal{I}_{j}\left(  z_{0},z;\kk,\omega\right)  =\int_{z_{0}
}^{z}iJ^{-1}A\left(  z_{1};\kk,\omega\right)  \mathcal{I}
_{j-1}\left(  z_{0},z_{1};\kk,\omega\right)  dz_{1}\text{ for
}j\geq 2,
\end{eqnarray*}
and this series converges uniformly and absolutely on any compact interval in
$\mathbb{R}$ in the $L^{\infty}$ norm $\left\Vert \cdot\right\Vert _{\infty}$ since
\begin{eqnarray}
A\left(  \cdot;\kk,\omega\right)  \in M_{4}\left(  L^{\infty
}\left(\mathbb{R}\right)  \right) 
\end{eqnarray}
More specifically, from our hypotheses of our materials
tensors being bounded, coercive, and independent of the parameters $\left(
\kk,\omega\right)  $ then it follows from the representation of
$A$ in (\ref{SchurComp}) that 
\begin{equation}
A\in\mathcal{O}(\mathbb{C}^{2}\times\mathbb{C}\setminus\left\{  0\right\},M_{4}\left(  L^{\infty}\left(\mathbb{R}\right)  \right)  ).\label{SmoothnessOfAinMaxwellDAEs}
\end{equation}
Here $\mathcal{O}$ denotes holomorphic functions. Hence it follows from this and the fact that the
series (\ref{PeanoBakerSeries}) converges uniformly and absolutely on any
compact interval in $\mathbb{R}$ in the $L^{\infty}$ norm $\left\Vert \cdot\right\Vert _{\infty}$, that
\begin{eqnarray*}
T\left(  z_{0},z;\cdot,\cdot\right)   & \in\mathcal{O}(\mathbb{C}^{2}\times\mathbb{C}\setminus\left\{  0\right\}  ,M_{4}\left(\mathbb{C}\right)  )\text{ and}\\
T(z_{0},\cdot;\cdot,\cdot)  & \in\mathcal{O}(\mathbb{C}^{2}\times\mathbb{C}\setminus\left\{  0\right\}  ,M_{4}\left(  L_{loc}^{\infty}\left(\mathbb{R}\right)  \right)  ).
\end{eqnarray*}

Now we can describe the dependency of the solutions of the Maxwell ODEs and
DAEs on the parameters $\left(  \kk,\omega\right)  \in\mathbb{C}^{2}\times\mathbb{C}\setminus\left\{  0\right\}  $ as well as the associated energy density for
such solutions. \ Let $\Psi\in\mathbb{C}^{4}$. Then it follows from the facts above that
\[
\psi\left(  \cdot;\kk,\omega\right)  =T\left(  z_{0}
,\cdot;\kk,\omega\right)  \Psi\in\left(  AC_{loc}\left(\mathbb{R}\right)  \right)  ^{4}
\]
is a solution of the Maxwell ODEs (\ref{cODEs}),
\[
F\left(  \cdot;\kk,\omega\right)  =P_{\parallel}^{\ast}\psi\left(
\cdot;\kk,\omega\right)  +P_{\bot}^{\ast}\Phi\left(
\cdot;\kk,\omega\right)  \psi\left(  \cdot;\kk
,\omega\right)
\]
is a solution to the DAEs (\ref{DAEs}), and
\begin{equation}
\psi\in\mathcal{O}(\mathbb{C}^{2}\times\mathbb{C}\setminus\left\{  0\right\}  ,\left(  L_{loc}^{\infty}\left(\mathbb{R}\right)  \right)  ^{4}),\text{ \ \ }F\in\mathcal{O}(\mathbb{C}^{2}\times\mathbb{C}\setminus\left\{  0\right\}  ,\left(  L_{loc}^{\infty}\left(\mathbb{R}\right)  \right)  ^{6}).\label{ParaDepSolns}
\end{equation}
In particular, this implies the energy density of the electromagnetic field
associated with $F\left(  \cdot;\kk,\omega\right)  =\left[
\EE\left(  \cdot;\kk,\omega\right)  ,\HH\left(  \cdot;\kk,\omega\right)  \right]  ^{T}$, namely,
\begin{equation}
U\left(  \cdot;\kk,\omega\right)  =\frac{1}{16\pi}\left(  \left(
\epsilon\left(  \cdot\right)  \EE\left(  \cdot;\kk,\omega\right)
,\EE\left(  \cdot;\kk,\omega\right)  \right)  +\left(  \mu\left(
\cdot\right)  \HH\left(  \cdot;\kk,\omega\right)  ,\HH\left(
\cdot;\kk,\omega\right)  \right)  \right)
\label{defAssocEnergyDens}
\end{equation}
is continuous on any compact interval in $\mathbb{R}$ in the $L^{\infty}$ norm $\left\Vert \cdot\right\Vert _{\infty}$, i.e.,
\begin{equation}
U\in\mathcal{C}(\mathbb{C}^{2}\times\mathbb{C}\setminus\left\{  0\right\}  ,L_{loc}^{\infty}\left(\mathbb{R}\right)  ),\label{ParaDepEnergyDens}
\end{equation}
a result that is needed in the proof of Theorem \ref{thm:EnergyDensity}.

We now conclude this section with a key new result, needed in the proof of Theorems \ref{thm:exponents} and \ref{thm:nondegeneracy}, on the representation of the energy density in a periodic medium in terms of spatial averaging and the Floquet theory. This result, described in the next theorem, is an extension of Theorems \ref{thm:EnergyConservationLaw}, \ref{thm:EnergyDensity} for one-dimensional photonic crystals.

\smallskip

\begin{Theorem}[Energy Density: Indicator Matrix Representation]\label{thm:EnergyIndRep}
Let $T(0,z)$ be the transfer matrix for the Maxwell ODEs (\ref{canonical}) of a $d$-periodic layered medium satisfying (\ref{PassiveMedia}), i.e., bounded measurable coercive $\epsilon$ and $\mu$ depending only on $z$ and $d$-periodic. Suppose at a parallel wavevector $\kk_0\in\RR^2$ the transfer matrix has a Floquet representation $T(0,z)=F(z)e^{iKz}$, i.e., $F(z+d)=F(z)$, $F(z)=I$ and $K$ is constant in $z$, such that $K=K(\omega)$ is both analytic in a complex neighborhood of a frequency $\omega_0\in\RR/\{0\}$ and self-adjoint with respect to the energy-flux form $[\cdot,\cdot]$ for real $\omega$ near $\omega_0$. Then for any integer $N$,
\begin{equation}
\frac{c}{16\pi}\int_{0}^{Nd}\left(\frac{\partial A}{\partial \omega}(z;\kk_0,\omega_0)\psi_1(z),\psi_2(z)\right)dz=\int_{0}^{Nd}\left[\frac{dK}{d\omega}(\omega_0) e^{iK(\omega_0) z}\psi_1(0),e^{iK(\omega_0)z}\psi_2(0)\right]dz,
\end{equation}
for any two solutions $\psi_1$, $\psi_2$ of the Maxwell ODEs at $\kwz$.
\end{Theorem}

\begin{proof}
The proof essentially follows from the fact that for any canonical ODEs
\begin{equation}
-iJ\frac{d}{dz}\phi(z)=M(z;\omega)\phi(z)
\end{equation}
with $M(\omega):=M(\cdot;\omega)$ analytic in $\omega$ near $\omega_0$ (i.e., bounded measurable entries as a function of $z$ and analytic in $\omega$ as a function into the Banach space $L^{\infty}(\RR)$) which is self-adjoint for a.e.\ $z$ for each real $\omega$ near $\omega_0$, its transfer matrix $Y(0,z;\omega)$ for real $\omega$ near $\omega_0$ satisfies the identity
\begin{equation}
  \frac{\partial }{\partial z}\left( Y^*(iJ)\frac{\partial Y}{\partial \omega} \right) = Y^*\frac{\partial M}{\partial \omega}Y,
\end{equation}
and hence for any $z>0$ this identity combined with the fact $Y(0,0;\omega)=I$ implies
\begin{equation}
\int_{0}^{\pm z}\left(\frac{\partial M}{\partial \omega}(z;\omega)\psi_1(z),\psi_2(z)\right)dz=\pm\left(Y(0,\pm z;\omega)^*(iJ)\frac{\partial Y}{\partial \omega}(0,\pm z;\omega)\psi_1(0),\psi_2(0)\right),
\end{equation}
for any two solutions $\phi_1$, $\phi_2$ of those canonical ODEs at $\omega$.

Now we use the fact that both $M(\cdot;\omega)=A(\cdot;\kk_0,\omega)$ and $M(\cdot;\omega)=JK(\omega)$ satisfy the hypotheses with the latter having transfer matrix $e^{iK(\omega)z}$ and so we find by taking $z=\pm Nd$ with $N$ a positive integer and using the fact $T(0,\pm Nd,\omega)=e^{iK(\omega)(\pm Nd)}$ that
\begin{gather}
\frac{c}{16\pi}\int_{0}^{\pm Nd}\left(\frac{\partial A}{\partial \omega}(z;\kk_0,\omega_0)\psi_1(z),\psi_2(z)\right)dz=\pm\frac{c}{16\pi}\left(T(0,\pm Nd;\omega_0)^*(iJ)\frac{\partial T}{\partial \omega}(0,\pm Nd;\omega_0)\psi_1(0),\psi_2(0)\right)\\
=\pm\frac{c}{16\pi}\left((e^{iK(\omega_0)(\pm Nd)})^*(iJ)\frac{\partial e^{iK(\omega)z}}{\partial \omega}(0,\pm Nd;\omega_0)\psi_1(0),\psi_2(0)\right)\\
=\frac{c}{16\pi}\int_{0}^{\pm Nd}\left(\frac{\partial JK}{\partial \omega}(z;\omega_0)e^{iK(\omega_0)z}\psi_1(0),e^{iK(\omega_0)z}\psi_2(0)\right)dz\\
=\int_{0}^{\pm Nd}\left[\frac{\partial K}{\partial \omega}(z;\omega_0)e^{iK(\omega_0)z}\psi_1(0),e^{iK(\omega_0)z}\psi_2(0)\right]dz,
\end{gather}
for any two solutions $\psi_1$, $\psi_2$ of the Maxwell ODEs at $\kwz$. This completes the proof.
\end{proof}

\bigskip
\bibliography{ShipmanWelters}

\begin{thebibliography}{10}

\bibitem{Beale1973}
J.~Thomas Beale.
\newblock Scattering frequencies of resonators.
\newblock {\em Comm. Pure Appl. Math.}, 26:549--563, 1973.

\bibitem{Berreman1972}
Dwight~W. Berreman.
\newblock Optics in stratified and anisotropic media: 4x4-matrix formulation.
\newblock {\em J. Opt. Soc. Am.}, 62(4):502--510, 1972.

\bibitem{Bonnet-BeStarling1994}
Anne-Sophie Bonnet-Bendhia and Felipe Starling.
\newblock Guided waves by electromagnetic gratings and nonuniqueness examples
  for the diffraction problem.
\newblock {\em Math. Methods Appl. Sci.}, 17(5):305--338, 1994.

\bibitem{DurandPaidarovaGadea2001}
Ph. Durand, I.~Paidarov{\'a}, and F.~X. Gad{\'e}a.
\newblock Theory of fano profiles.
\newblock {\em J. Phys. B}, 34:1953--1966, 2001.

\bibitem{FanVilleneuveJoannopoul2000}
S.~Fan, P.R. Villeneuve, and J.D. Joannopoulos.
\newblock Rate-equation analysis of output efficiency and modulation rate of
  photonic-crystal light-emitting diodes.
\newblock {\em Quantum Electronics, IEEE Journal of}, 36(10):1123--1130, Oct
  2000.

\bibitem{FanJoannopoul2002}
Shanhui Fan and J.~D. Joannopoulos.
\newblock Analysis of guided resonances in photonic crystal slabs.
\newblock {\em Phys. Rev. B}, 65(23):235112--1--8, Jun 2002.

\bibitem{FanSuhJoannopoul2003}
Shanhui Fan, Wonjoo Suh, and J.~D. Joannopoulos.
\newblock Temporal coupled-mode theory for the fano resonance in optical
  resonators.
\newblock {\em J. Opt. Soc. Am. A}, 20(3):569--572, 2003.

\bibitem{Fano1961}
U.~Fano.
\newblock Effects of configuration interaction on intensities and phase shifts.
\newblock {\em Physical Review}, 124(6):1866--1878, 1961.

\bibitem{FigotinVitebskiy2006}
Alex Figotin and Ilya Vitebskiy.
\newblock Slow light in photonic crystals.
\newblock {\em Waves in Random and Complex Media}, 16(3):293--382, 2006.

\bibitem{GohbergLancasterRodman2005}
Israel Gohberg, Peter Lancaster, and Leiba Rodman.
\newblock {\em Indefinite Linear Algebra and Applications}.
\newblock Birkh{\"a}user Verlag AG, 2005.

\bibitem{HausMiller1986}
Hermann~A. Haus and David A.~B. Miller.
\newblock Attenuation of cutoff modes and leaky modes of dielectric slab
  structures.
\newblock {\em IEEE J. Quantum Elect.}, 22(2):310--318, 1986.

\bibitem{KanskarPaddonPacradouni1997}
M.~Kanskar, P.~Paddon, V.~Pacradouni, R.~Morin, A.~Busch, Jeff~F. Young, S.~R.
  Johnson, Jim MacKenzie, and T.~Tiedje.
\newblock Observation of leaky slab modes in an air-bridged semiconductor
  waveguide with a two-dimensional photonic lattice.
\newblock {\em Applied Physics Letters}, 70(11):1438--1440, 1997.

\bibitem{Lamb1900}
Horace Lamb.
\newblock On a peculiarity of the wave-system due to free vibrations of a
  nucleus in an extended medium.
\newblock {\em Proc. Lond. Math. Soc.}, XXXII(723):208--211, 1900.

\bibitem{PengTamirBertoni1975}
S.T. Peng, T.~Tamir, and H.L. Bertoni.
\newblock Theory of periodic dielectric waveguides.
\newblock {\em Microwave Theory and Techniques, IEEE Transactions on},
  23(1):123--133, Jan 1975.

\bibitem{PtitsynaShipman2012}
Natalia Ptitsyna and Stephen~P. Shipman.
\newblock A lattice model for resonance in open periodic waveguides.
\newblock {\em Discret Contin Dyn S}, 5(5), 2012.

\bibitem{PtitsynaShipmanVenakides2008}
Natalia Ptitsyna, Stephen~P. Shipman, and Stephanos Venakides.
\newblock Fano resonance of waves in periodic slabs.
\newblock pages 73--78. Intl. Conf. on Math. Meth. in EM Theory, IEEE, 2008.

\bibitem{ReedSimon1980d}
Michael Reed and Barry Simon.
\newblock {\em Methods of Mathematical Physics: Analysis of Operators},
  volume~IV.
\newblock Academic Press, 1980.

\bibitem{Shipman2010}
Stephen~P. Shipman.
\newblock {\em Resonant Scattering by Open Periodic Waveguides}, volume~1 of
  {\em E-Book, Progress in Computational Physics}.
\newblock Bentham Science Publishers, 2010.

\bibitem{ShipmanRibbeckSmith2010}
Stephen~P. Shipman, Jennifer Ribbeck, Katherine~H. Smith, and Clayton Weeks.
\newblock A discrete model for resonance near embedded bound states.
\newblock {\em IEEE Photonics J.}, 2(6):911--923, 2010.

\bibitem{ShipmanTu2012}
Stephen~P. Shipman and Hairui Tu.
\newblock Total resonant transmission and reflection by periodic structures.
\newblock {\em SIAM J. Appl. Math.}, 72(1):216--239, 2012.

\bibitem{ShipmanVenakides2003}
Stephen~P. Shipman and Stephanos Venakides.
\newblock Resonance and bound states in photonic crystal slabs.
\newblock {\em SIAM J. Appl. Math.}, 64(1):322--342 (electronic), 2003.

\bibitem{ShipmanVenakides2005}
Stephen~P. Shipman and Stephanos Venakides.
\newblock Resonant transmission near non-robust periodic slab modes.
\newblock {\em Phys. Rev. E}, 71(1):026611--1--10, 2005.

\bibitem{ShipmanWelters2012}
Stephen~P. Shipman and Aaron~T. Welters.
\newblock Resonance in anisotropic layered media.
\newblock pages 227--232. Proc. Intl. Conf. on Math. Meth. in EM Theory,
  Kharkov, IEEE, 2012.

\bibitem{TikhodeevYablonskiMuljarov2002}
Sergei~G. Tikhodeev, A.~L. Yablonskii, E.~A. Muljarov, N.~A. Gippius, and
  Teruya Ishihara.
\newblock Quasiguided modes and optical properties of photonic crystal slabs.
\newblock {\em Phys. Rev. B}, 66:045102--1--17, 2002.

\bibitem{Wei-HsuZhenChua2013}
Chia Wei~Hsu, Bo~Zhen, Song-Liang Chua, Steven~G. Johnson, John~D.
  Joannopoulos, and Marin Soljacic.
\newblock Bloch surface eigenstates within the radiation continuum.
\newblock {\em Light: Science App.}, 2(e84 doi:10:1038/Isa.2013.40), 2013.

\bibitem{Wei-HsuZhenLee2013}
Chia Wei~Hsu, Bo~Zhen, Jeongwon Lee, Song-Liang Chua, Steven~G. Johnson,
  John~D. Joannopoulos, and Marin Soljacic.
\newblock Observation of trapped light within the radiation continuum.
\newblock {\em Nature}, (doi:10.1038/nature12289), July 2013.

\bibitem{Welters2011}
Aaron Welters.
\newblock On explicit recursive formulas in the spectral perturbation analysis
  of a jordan block.
\newblock {\em SIAM J. Matrix Anal. Appl.}, 32(1):1--22, 2011.

\bibitem{Welters2011a}
Aaron~Thomas Welters.
\newblock {\em On the Mathematics of Slow Light}.
\newblock PhD thesis, UC Irvine, 2011.

\bibitem{YakubovichStarzhinsk1975}
V.~A. Yakubovich and V.~M. Starzhinskii.
\newblock {\em Linear Differential Equaions with Periodic Coefficients I}.
\newblock Halsted Press, 1975.

\bibitem{Zworski1999}
Maciej Zworski.
\newblock Resonances in physics and geometry.
\newblock {\em Notices of the AMS}, 46(3), 1999.

\bibitem{Zworski2011}
Maciej Zworski.
\newblock Lectures on scattering resonances.
\newblock Lecture Notes, 2011.

\end{thebibliography}

\end{document}